\documentclass[journal, twocolumn]{IEEEtran}

\usepackage{graphicx,amssymb,amstext,amsmath,cite}
\usepackage{tikz}
\usepackage{amsthm}
\usepackage{kbordermatrix}

\newcounter{actr}
{\begin{list}{(\alph{actr})}{\usecounter{actr}}}{\end{list}}

\newcounter{ictr}
{\begin{list}{(\roman{ictr})}{\usecounter{ictr}}}{\end{list}}

\newtheorem{remark}{Remark}
\newtheorem{thm}{Theorem}
\newtheorem{lemma}{Lemma}
\newtheorem{claim}{Claim}
\newtheorem{corol}{Corollary}
\newtheorem{prop}{Proposition}
\newtheorem{defn}{Definition}

\newenvironment{new-proof}[1]
{{\em Proof }:\\}%
{ \noindent\qed }
\newcommand{\defeq}{\stackrel{\Delta}{=}}

\newcommand{\mrm}{\mathrm}

\newcommand{\bA}{{\mathbf{A}}}

\newcommand{\bb}{{\mathbf{b}}}

\newcommand{\cB}{{\mathcal{B}}}

\newcommand{\bB}{{\mathbf{B}}}

\newcommand{\bc}{{\mathbf{c}}}

\newcommand{\cC}{{\mathcal{C}}}

\newcommand{\bd}{{\mathbf{d}}}
\newcommand{\bD}{{\mathbf{D}}}

\newcommand{\cF}{{\mathcal{F}}}

\newcommand{\cG}{{\mathcal{G}}}

\newcommand{\bI}{{\mathbf{I}}}

\newcommand{\cL}{{\mathcal{L}}}

\newcommand{\bM}{{\mathbf{M}}}

\newcommand{\cM}{{\mathcal{M}}}

\newcommand{\bR}{{\mathbf{R}}}

\newcommand{\bs}{{\mathbf{s}}}

\newcommand{\cS}{{\mathcal{S}}}

\newcommand{\bV}{{\mathbf{V}}}

\newcommand{\bx}{{\mathbf{x}}}

\newcommand{\bX}{{\mathbf{X}}}

\newcommand{\bt}{\boldsymbol{t}}

\newcommand{\eps}{\varepsilon}

\DeclareMathAlphabet{\mathbsf}{OT1}{cmss}{bx}{n}
\DeclareMathAlphabet{\mathssf}{OT1}{cmss}{m}{sl}

\DeclareSymbolFont{bsfletters}{OT1}{cmss}{bx}{n}
\DeclareSymbolFont{ssfletters}{OT1}{cmss}{m}{n}
\DeclareMathSymbol{\bsfGamma}{0}{bsfletters}{'000}
\DeclareMathSymbol{\ssfGamma}{0}{ssfletters}{'000}
\DeclareMathSymbol{\bsfDelta}{0}{bsfletters}{'001}
\DeclareMathSymbol{\ssfDelta}{0}{ssfletters}{'001}
\DeclareMathSymbol{\bsfTheta}{0}{bsfletters}{'002}
\DeclareMathSymbol{\ssfTheta}{0}{ssfletters}{'002}
\DeclareMathSymbol{\bsfLambda}{0}{bsfletters}{'003}
\DeclareMathSymbol{\ssfLambda}{0}{ssfletters}{'003}
\DeclareMathSymbol{\bsfXi}{0}{bsfletters}{'004}
\DeclareMathSymbol{\ssfXi}{0}{ssfletters}{'004}
\DeclareMathSymbol{\bsfPi}{0}{bsfletters}{'005}
\DeclareMathSymbol{\ssfPi}{0}{ssfletters}{'005}
\DeclareMathSymbol{\bsfSigma}{0}{bsfletters}{'006}
\DeclareMathSymbol{\ssfSigma}{0}{ssfletters}{'006}
\DeclareMathSymbol{\bsfUpsilon}{0}{bsfletters}{'007}
\DeclareMathSymbol{\ssfUpsilon}{0}{ssfletters}{'007}
\DeclareMathSymbol{\bsfPhi}{0}{bsfletters}{'010}
\DeclareMathSymbol{\ssfPhi}{0}{ssfletters}{'010}
\DeclareMathSymbol{\bsfPsi}{0}{bsfletters}{'011}
\DeclareMathSymbol{\ssfPsi}{0}{ssfletters}{'011}
\DeclareMathSymbol{\bsfOmega}{0}{bsfletters}{'012}
\DeclareMathSymbol{\ssfOmega}{0}{ssfletters}{'012}

\renewcommand{\defeq}{\triangleq}

\newcommand{\rvb}{{\mathssf{b}}}    

\newcommand{\rvbd}{{\mathbsf{d}}}

\newcommand{\rvf}{{\mathssf{f}}}    

\newcommand{\rvm}{{\mathssf{m}}}    

\newcommand{\rvs}{{\mathssf{s}}}    
\newcommand{\rvbs}{{\mathbsf{s}}}

\newcommand{\rvt}{{\mathssf{t}}}    
\newcommand{\rvbt}{{\mathbsf{t}}}

\newcommand{\rvz}{{\mathssf{z}}}    

\author{{Farrokh~Etezadi,  Ashish~Khisti and Mitchell~Trott}
\thanks{Farrokh Etezadi (fetezadi@comm.utoronto.ca) and Ashish Khist ({akhisti@comm.utoronto.ca}) are with the University of Toronto, Toronto, ON, Canada, Mitchell Trott is with HP Labs, Palo Alto, USA. 
The authors are listed alphabetically. This work was supported by an NSERC Discovery Research Grant and a Hewlett-Packard Innovation Research Program award.}}

\title{Sequential Coding of  Markov Sources over Burst Erasure Channels}

\begin{document}

\maketitle

\begin{abstract}
We  study  sequential coding of Markov sources under an error propagation constraint.
An encoder sequentially compresses a sequence of vector-sources that are spatially i.i.d.\ but temporally correlated according to a first-order Markov process. 
The channel erases up to $B$ packets in a single burst, but reveals all other packets to the destination. 
The destination is required to reproduce all the source-vectors   instantaneously and in a lossless manner, except those sequences that occur in an {\em error propagation window}  of length $B+W$ following the start of the erasure burst. We define  the rate-recovery function $R(B,W)$ --- the minimum achievable compression rate per source sample in this framework --- and develop  upper and lower bounds on this function. Our upper bound is obtained using a random binning technique, whereas our lower bound  is obtained by drawing  connections to multi-terminal source coding. Our upper and lower bounds coincide, yielding $R(B,W)$, in some special cases. More generally, both the upper and lower bounds equal the rate for predictive coding plus a term that decreases as $\frac{1}{W+1},$ thus establishing a scaling behaviour of the rate-recovery function.

For a special class of {\em semi-deterministic} Markov sources we propose a new optimal coding scheme:  {\em prospicient} coding. An extension of this coding technique to Gaussian sources is also developed. For the class of symmetric Markov sources and memoryless encoders, we establish the optimality of random binning. When the destination is required to reproduce each source sequence with a fixed delay and when $W=0$ we also establish the optimality of binning.
\end{abstract}

\begin{keywords}
Streaming source coding, Rate-distortion Theory, Sequential coding, Source coding, Video coding.
\end{keywords}

\section{Introduction}
\IEEEPARstart{T}rade-off between compression efficiency and error resilience is fundamental to any video compression system. In live video streaming, an encoder observes a sequence of correlated video frames and produces a compressed bit-stream that is transmitted to the destination. If the underlying channel is an ideal bit-pipe, it is well known that predictive coding~\cite{berger71} achieves the optimum compression rate. Unfortunately  packet losses are unavoidable in many emerging video distribution systems with stringent delay constraints. Predictive coding  is highly sensitive to such packet losses and can lead to a significant amount of error propagation. In practice various mechanisms have been engineered   to prevent such losses. For example video codecs use a group of picture (GOP) architecture, where intra-frames are periodically inserted to limit the effect of error propagation. Forward error correction codes can also be applied to compressed bit-streams to recover any  missing packets~\cite{tan,li:10}.  Modifications to predictive coding, such as leaky-DPCM~\cite{Huang:08,Connor:73}, have been proposed in the literature to deal with packet losses. The robustness of distributed video coding techniques in such situations has been studied in e.g.,~\cite{Sehgal,pradhanRamchandran:03}. 

Information theoretic analysis of video coding has received significant attention in recent times, see e.g.,~\cite{berger,ishwar,yang} and the references therein. These works focus primarily on the source coding aspects of video. The {\em source process} is modeled as a sequence of vectors, each of which is  spatially i.i.d.\ and temporally correlated.  The encoder is generally restricted to be either causal or having a limited look-ahead. The destination is required to output each source vectors in a sequential manner.  However all of these works assume an ideal channel with no packet losses. To our knowledge, even the effect of a single isolated packet loss is not fully understood~\cite{Wang:06}.

In this work, we study a fundamental trade-off between error propagation and compression rate in sequential source coding when the channel introduces packet losses. The encoder compresses the source-vector sequence in a causal manner and the receiver is required to recover each source sequence in an instantaneous and lossless manner. The channel introduces a burst of $B$ erasures and the destination is not required to recover $B+W$ source sequences following the start of the erasure burst. We introduce the {\em rate-recovery function} $R(B,W)$ --- the minimum achievable compression rate in this framework. Upper and lower bounds on this function are developed. The upper bound is obtained using a binning based scheme. The lower bound is obtained by drawing connections to a multi-terminal source coding problem. Conditions under which the upper and lower bounds coincide are discussed. In particular we establish that the rate-recovery function equals the predictive coding rate plus  a term that decreases as $\frac{1}{W+1}$, where $W$ is the length of the error propagation window. 

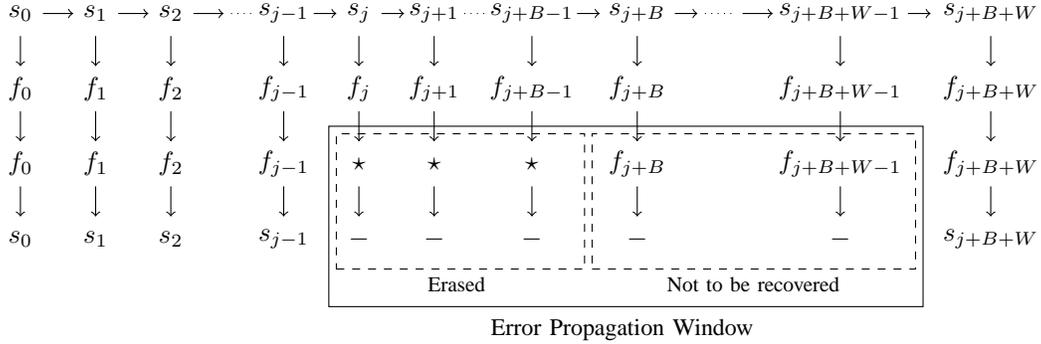
\begin{figure*}
\begin{center}
\vspace{1em}
\begin{tikzpicture}

\draw [white] (0,0) -- (0,0) node {$\color{black}s_0$};
\draw [->](0,-.3) -- (0,-.7);
\draw [white] (0,-1) -- (0,-1) node {$\color{black}f_0$};
\draw [->](0,-1.3) -- (0,-1.7);
\draw [white] (0,-2) -- (0,-2) node {$\color{black}f_0$};
\draw [->](0,-2.3) -- (0,-2.7);
\draw [white] (0,-3) -- (0,-3) node {$\color{black}s_0$};
\draw [->](0.3,0) -- (0.7,0);

\draw [white] (1,0) -- (1,0) node {$\color{black}s_1$};
\draw [->](1,-.3) -- (1,-.7);
\draw [white] (1,-1) -- (1,-1) node {$\color{black}f_1$};
\draw [->](1,-1.3) -- (1,-1.7);
\draw [white] (1,-2) -- (1,-2) node {$\color{black}f_1$};
\draw [->](1,-2.3) -- (1,-2.7);
\draw [white] (1,-3) -- (1,-3) node {$\color{black}s_1$};
\draw [->](1.3,0) -- (1.7,0);

\draw [white] (2,0) -- (2,0) node {$\color{black}s_{2}$};
\draw [->](2,-.3) -- (2,-.7);
\draw [white] (2,-1) -- (2,-1) node {$\color{black}f_2$};
\draw [->](2,-1.3) -- (2,-1.7);
\draw [white] (2,-2) -- (2,-2) node {$\color{black}f_2$};
\draw [->](2,-2.3) -- (2,-2.7);
\draw [white] (2,-3) -- (2,-3) node {$\color{black}s_2$};
\draw [->](2.3,0) -- (2.7,0);

\draw [dotted](2.8,0) -- (3.1,0);

\draw [white] (3.5,0) -- (3.5,0) node {$\color{black}s_{j-1}$};
\draw [->](3.5,-.3) -- (3.5,-.7);
\draw [white] (3.5,-1) -- (3.5,-1) node {$\color{black}f_{j-1}$};
\draw [->](3.5,-1.3) -- (3.5,-1.7);
\draw [white] (3.5,-2) -- (3.5,-2) node {$\color{black}f_{j-1}$};
\draw [->](3.5,-2.3) -- (3.5,-2.7);
\draw [white] (3.5,-3) -- (3.5,-3) node {$\color{black}s_{j-1}$};
\draw [->](3.9,0) -- (4.2,0);

\draw [white] (4.5,0) -- (4.5,0) node {$\color{black}s_{j}$};
\draw [->](4.5,-.3) -- (4.5,-.7);
\draw [white] (4.5,-1) -- (4.5,-1) node {$\color{black}f_{j}$};
\draw [->](4.5,-1.3) -- (4.5,-1.7);
\draw [white] (4.5,-2) -- (4.5,-2) node {$\color{black}\star$};
\draw [->](4.5,-2.3) -- (4.5,-2.7);
\draw [white] (4.5,-3) -- (4.5,-3) node {$\color{black}-$};
\draw [->](4.8,0) -- (5.1,0);

\draw [white] (5.5,0) -- (5.5,0) node {$\color{black}s_{j+1}$};
\draw [->](5.5,-.3) -- (5.5,-.7);
\draw [white] (5.5,-1) -- (5.5,-1) node {$\color{black}f_{j+1}$};
\draw [->](5.5,-1.3) -- (5.5,-1.7);
\draw [white] (5.5,-2) -- (5.5,-2) node {$\color{black}\star$};
\draw [->](5.5,-2.3) -- (5.5,-2.7);
\draw [white] (5.5,-3) -- (5.5,-3) node {$\color{black}-$};

\draw [dotted](5.9,0) -- (6.2,0);

\draw [white] (6.8,0) -- (6.8,0) node {$\color{black}s_{j+B-1}$};
\draw [->](6.8,-.3) -- (6.8,-.7);
\draw [white] (6.8,-1) -- (6.8,-1) node {$\color{black}f_{j+B-1}$};
\draw [->](6.8,-1.3) -- (6.8,-1.7);
\draw [white] (6.8,-2) -- (6.8,-2) node {$\color{black}\star$};
\draw [->](6.8,-2.3) -- (6.8,-2.7);
\draw [white] (6.8,-3) -- (6.8,-3) node {$\color{black}-$};
\draw [->](7.4,0) -- (7.7,0);

\draw [white] (8.2,0) -- (8.2,0) node {$\color{black}s_{j+B}$};
\draw [->](8.2,-.3) -- (8.2,-.7);
\draw [white] (8.2,-1) -- (8.2,-1) node {$\color{black}f_{j+B}$};
\draw [->](8.2,-1.3) -- (8.2,-1.7);
\draw [white] (8.2,-2) -- (8.2,-2) node {$\color{black}f_{j+B}$};
\draw [->](8.2,-2.3) -- (8.2,-2.7);
\draw [white] (8.2,-3) -- (8.2,-3) node {$\color{black}-$};
\draw [->](8.7,0) -- (9,0);

\draw [dotted](9.1,0) -- (9.5,0);

\draw [->](9.6,0) -- (10,0);
\draw [white] (10.9,0) -- (10.9,0) node {$\color{black}s_{j+B+W-1}$};
\draw [->](10.9,-.3) -- (10.9,-.7);
\draw [white] (10.9,-1) -- (10.9,-1) node {$\color{black}f_{j+B+W-1}$};
\draw [->](10.9,-1.3) -- (10.9,-1.7);
\draw [white] (10.9,-2) -- (10.9,-2) node {$\color{black}f_{j+B+W-1}$};
\draw [->](10.9,-2.3) -- (10.9,-2.7);
\draw [white] (10.9,-3) -- (10.9,-3) node {$\color{black}-$};

\draw [->](11.8,0) -- (12.1,0);
\draw [white] (12.9,0) -- (12.9,0) node {$\color{black}s_{j+B+W}$};
\draw [->](12.9,-.3) -- (12.9,-.7);
\draw [white] (12.9,-1) -- (12.9,-1) node {$\color{black}f_{j+B+W}$};
\draw [->](12.9,-1.3) -- (12.9,-1.7);
\draw [white] (12.9,-2) -- (12.9,-2) node {$\color{black}f_{j+B+W}$};
\draw [->](12.9,-2.3) -- (12.9,-2.7);
\draw [white] (12.9,-3) -- (12.9,-3) node {$\color{black}s_{j+B+W}$};

\draw [white] (5.8,-3.6) -- (5.8,-3.6) node {$\color{black}\textrm{\footnotesize{Erased}}$};
\draw [white] (9.75,-3.6) -- (9.75,-3.6) node {$\color{black}\textrm{\footnotesize{Not to be recovered}}$};


\draw [dashed](4.2,-3.4) rectangle (7.5,-1.6);
\draw [dashed](7.6,-3.4) rectangle (11.9,-1.6);
\draw (4.1,-3.9) rectangle (12,-1.5);
\draw [white] (8,-4.2) -- (8,-4.2) node {$\color{black}\textrm{\small{Error Propagation Window}}$};
\end{tikzpicture}
\caption{ Problem Setup: The encoder output $\rvf_j$ is a  function of  the source sequences up to time $j$ i.e., ${\rvs_0^n, \rvs_1^n,\ldots, \rvs_j^n}$. 
The channel introduces an erasure burst of length $B$.
The decoder produces $\hat{s}_j^n $ upon observing the sequence $g_0^j$.  As indicated decoder
is not required to produce those source sequences that fall in a window of length $B+W$ following the start of an erasure burst.}
\label{fig:setup}
\end{center}
\end{figure*}

We study special class of sources for which the binning based upper bound can be improved by exploiting the underlying structure. First we consider the linear semi-deterministic Markov source  and develop a new coding technique --- {\em prospicient coding} that meets the lower bound for all values of $B$ and $W$. In our proposed scheme, we first transform the linear semi-deterministic source into a simpler diagonally-correlated source.  For the latter class, we provide an explicit coding scheme that meets the lower bound.  We also extend the proposed coding technique to an i.i.d.\ Gaussian source process, where the receiver is required to recover source sequences in a sliding window of fixed length. Numerical results indicate significant improvements of the proposed coding scheme over techniques such as FEC based coding and naive binning. 

For the class of {\em symmetric sources}, with an additional assumption of a {\em memoryless encoder}, we establish the optimality of a binning based technique.  This is done by establishing a connection with another multi-terminal source coding problem --- the {\em Zig-Zag source network} with side information~\cite{berger:77}. For our streaming problem, we need to only lower bound the sum-rate for this network, which we do by exploiting the symmetric nature of the underlying sources.

As another extension, we consider the case when the decoder is allowed a fixed decoding delay of $T$ frames. When $W=0$ we again establish the optimality of binning.  For the converse, we introduce a periodic erasure channel of period $B+T+1$, where the first $B$ packets are erased. We argue that the decoder can recover each of the remaining source sequences by their deadline  and invoke the source coding theorem to find a lower bound on the rate-recovery function. 

The remainder of the paper is organized as follows. The problem setup is described in Section~\ref{sec:statement} and a summary of the main results is provided in Section~\ref{sec:Results}. Our upper and lower bounds on the rate-recovery function are established in section~\ref{sec:UBLB}. The prospicient coding scheme is described for the class of diagonally correlated deterministic sources in Section~\ref{sec:deterministic}, for the linear deterministic sources in Section~\ref{sec:LSD} and for Gaussian sources in Section~\ref{sec:Gauss}. The optimality of binning for symmetric sources is established in Section~\ref{sec:Symmetric} whereas the case of delay-constrained decoder is treated in section~\ref{sec:no_W}. Conclusions are provided in section~\ref{sec:concl}.

\section{Problem Statement}
\label{sec:statement}
In this section we describe the source and channel models as well as our notion of an error-propagation window and the associated rate-recovery function. 

\subsection{Source Model}
\label{subsec:Markov}
We consider a semi-infinite stationary vector source process $\{\rvs_t^n\}_{t\ge 0}$ whose\footnote{In establishing our coding theorems we assume that the source process starts at $t=-1$ (or before if required) and that all the sequences with a negative index are  revealed to the destination. We will also assume that the transmission terminates after a sufficiently long period.} symbols (defined over some finite alphabet $\cS$) are drawn independently across the spatial dimension and from a first-order Markov chain across the temporal dimension, i.e., for each $t \ge 1$,
\begin{multline}
\Pr(~\rvs_t^n = s_t^n~|~\rvs_{t-1}^n = s_{t-1}^n,~\rvs_{t-2}^n = s_{t-2}^n,\ldots)  \\= \prod_{j=1}^n p_{\rvs_1|\rvs_0}(s_{t,j}|s_{t-1,j}), \quad \forall t \ge 1.\label{eq:Markov}
\end{multline}
We assume that the underlying random variables $\{\rvs_t\}_{t\ge 0}$ constitute a time-invariant, stationary and a first-order stationary Markov chain with a common marginal distribution denoted by $p_\rvs(\cdot)$.  Such models are  used in earlier works on sequential source coding. See e.g.,~\cite{ishwar} for some justification. We remark that the results for the lossless recovery also generalize when the source sequence is a stationary process (not  necessarily \mbox{i.i.d.\ )} in the spatial dimension. However the extension to higher order Markov process appears non-trivial.


\subsection{Channel Model}
The channel introduces an erasure burst of size $B$, i.e. for some particular $j \ge 0$, it introduces an erasure burst such that  $g_{i}=\star$ for $i\in\{j,j+1,...,j+B-1\}$ and  $g_{i}=f_{i}$ otherwise i.e.,\begin{align}
g_i =\begin{cases}
\star, &i \in [j,j+1,\ldots, j+B-1]\\
f_i, &\text{ else}.
\end{cases}\label{eq:chModel}
\end{align}

\subsection{Rate-Recovery Function}
\label{subsec:rate-recovery}
A rate-$R$ causal encoder maps the sequence $\{\rvs_i^n\}_{i\ge 0}$ to an index $f_i \in [1,2^{nR}]$ according to some function \begin{align}\rvf_i=\mathcal{F}_i\left({s}^n_0,..., {s}^n_{i}\right) \label{eq:f-enc}\end{align} for each $i \ge 0$. For most of our discussion we will assume causal encoders. Furthermore, a {\em memoryless} encoder satisfies $\mathcal{F}_i\left({s}^n_0,..., {s}^n_{i}\right) = \cF_i(s_i^n)$ i.e., the encoder does  not use the knowledge of the past sequences. Naturally a memoryless encoder is very restricted and we will only use it to establish some special results. 

Upon observing the sequence $\{g_i\}_{i\ge 0}$  the decoder is required to perfectly recover all the source sequences  using decoding functions
\begin{align}
\hat{{s}}^{n}_i=\mathcal{G}_{i}(g_0, g_1, \ldots, g_{i}), \quad i \notin \{j, \ldots, j+B+W-1\}.
\end{align}
where $j$ denotes the time at which the erasure burst starts in~\eqref{eq:chModel}.
It is however not required to produce the  source sequences in the window of length $B+W$ following the start of an erasure burst.
We call this period the error propagation window.
The setup is shown in Fig.~\ref{fig:setup}.

A rate $R(B, W)$ is feasible if there exists a sequence of encoding and decoding functions and a sequence $\epsilon_{n}$ that approaches zero as $n \to \infty$ such that, $\Pr({s}_i^n \neq \hat{{s}_i}^{n})\leq\epsilon_{n}$ for all $i\notin\{j,..., j+B+W-1\}$. We seek the minimum feasible rate $R(B,W)$, which we define to be the {\em rate-recovery} function.

\section{Main Results}
\label{sec:Results}
In this section we discuss the main results of this paper.

\subsection{Upper and Lower Bounds}
\label{subsec:gen}
\begin{thm}
For any stationary first-order Markov source process the rate-recovery function satisfies $R^-(B,W) \le R(B,W)\le R^+(B,W)$ where
\begin{align}
R^{+}(B, W) &\!=\! H(\rvs_{1}|\rvs_0)+\frac{1}{W+1}I(\rvs_{B}~;~\rvs_{B+1}|\rvs_{0}), \label{eq:genUB}\\
R^-(B, W) &\!=\! H(\rvs_{1}|\rvs_0)+\frac{1}{W+1}I(\rvs_{B};\rvs_{W+B+1}|\rvs_{0}). \label{eq:genLB}
\end{align}
\label{thm:genUB_LB}
\end{thm}
Notice that the upper and lower bound coincide for $W=0$ and $W \rightarrow\infty$, yielding the rate-recovery function in these cases. More generally we can interpret the term $H(\rvs_1|\rvs_0)$ as the rate associated with ideal predictive coding in absence of any erasures. Theorem~\ref{thm:genUB_LB} suggests that the rate-recovery function equals $H(\rvs_1|\rvs_0)$ plus a term that decreases as $\frac{1}{W+1}$. 

The upper bound is obtained using a  binning based scheme.  At each time the encoding function $f_i$ in~\eqref{eq:f-enc}  is  the bin-index of a Slepian-Wolf codebook~\cite{slepianWolf:73}. Following an erasure burst in $[j,j+B-1]$, the decoder collects $f_{j+B},\ldots, f_{j+W+B}$ and attempts to jointly recover all the underlying sources  at $t=j+W+B$.  The following corollary provides an alternate expression for the achievable rate and makes the connection to the binning technique more explicit. 
\begin{corol}
\label{corol:genUB}
For any first order Markov source process defined in Section~\ref{subsec:Markov}, the upper bound in~\eqref{eq:genUB} can also be expressed as
\begin{align}
R^+(B,W) = \frac{1}{W+1} H(\rvs_{B+1}, \rvs_{B+2}, \ldots, \rvs_{B+W+1}|s_{0}). \label{eq:genUB_Slepian-Wolf}
\end{align}
\end{corol}
The proof of Corollary~\ref{corol:genUB} is provided in Appendix~\ref{app:Cor1}. Although our framework assumes a single isolated erasure burst, we note that the coding scheme enables recovery  in the presence of multiple erasure bursts, provided there is a guard interval  of at least $W+1$ between these bursts. 

Our lower bound involves two key ideas that we illustrate below for the case when $W=1$ and $B=1$. First we develop the following equivalent expression of~\eqref{eq:genLB} which is easier to interpret:
\begin{align}
&R^-(B=1,W=1) = H(\rvs_1|\rvs_0)+ \frac{1}{2}I(\rvs_1;\rvs_{3}|\rvs_0)\\
&= H(\rvs_1|\rvs_0) + \frac{1}{2}H(\rvs_3|\rvs_0) -\frac{1}{2}H(\rvs_3 | \rvs_0, \rvs_1)\\
&=\frac{1}{2}H(\rvs_1, \rvs_2 | \rvs_0) + \frac{1}{2}H(\rvs_3 | \rvs_0) - \frac{1}{2}H(\rvs_3 | \rvs_1)\label{eq:MC_Prop1}\\
&= \frac{1}{2}H(\rvs_1|\rvs_0,\rvs_2) + \frac{1}{2}H(\rvs_3|\rvs_0)\label{eq:MC_Prop2}
\end{align}
where both~\eqref{eq:MC_Prop1} and~\eqref{eq:MC_Prop2} follow from the first-order Markov Chain property $\rvs_0 \rightarrow\rvs_1 \rightarrow \rvs_2 \rightarrow \rvs_3.$

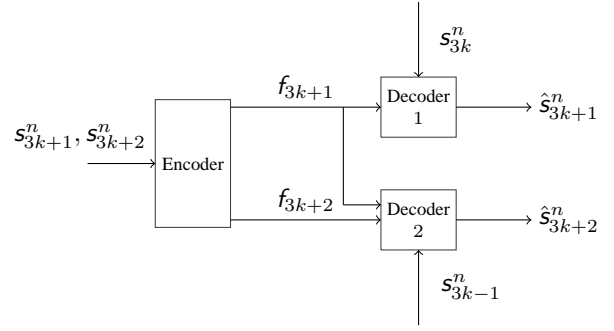
\begin{figure}
\begin{center}
\vspace{1em}
\begin{tikzpicture}

\def \a {0}{
\draw [white] (-1,-\a-.8) -- (-1,-\a-.8) node {$\color{black}\rvs^n_{3k+1},\rvs^n_{3k+2}$};
\fill [fill=white, draw=gray!50!black] (0,-.3-\a) -- (0,-2-\a) -- (1,-2-\a) -- (1,-.3-\a) -- (0,-.3-\a);
\draw [white] (.5,-\a-1.15) -- (.5,-\a-1.15) node {$\color{black}\scriptstyle\textrm{Encoder}$};
\draw [->] (-.9,-1.15) -- (0,-1.15);
\draw [->] (1,-.4) -- (3,-.4);
\fill [fill=white, draw=gray!50!black] (3,-.8-\a) -- (3,-\a) -- (4,-\a) -- (4,-.8-\a) -- (3,-.8-\a);
\draw [white] (3.5,-\a-0.25) -- (3.5,-\a-0.25) node {$\color{black}\scriptstyle\textrm{Decoder}$};
\draw [white] (3.5,-\a-0.55) -- (3.5,-\a-0.55) node {$\color{black}\scriptstyle{1}$};
\draw [->] (4,-.4) -- (5,-.4);
\draw [white] (5.5,-\a-0.4) -- (5.5,-\a-0.4) node {$\color{black}\hat{\rvs}^n_{3k+1}$};
\draw [white] (2,-\a-0.15) -- (2,-\a-0.15) node {$\color{black}\rvf_{3k+1}$};
\draw [->] (3.5,-\a+1) -- (3.5,-\a);
\draw [white] (4,-\a+.5) -- (4,-\a+.5) node {$\color{black}\rvs^n_{3k}$};
}
\def \a {1.5}{
\draw [->] (1,-.4-\a) -- (3,-.4-\a);
\fill [fill=white, draw=gray!50!black] (3,-.8-\a) -- (3,-\a) -- (4,-\a) -- (4,-.8-\a) -- (3,-.8-\a);
\draw [white] (3.5,-\a-0.25) -- (3.5,-\a-0.25) node {$\color{black}\scriptstyle\textrm{Decoder}$};
\draw [white] (3.5,-\a-0.55) -- (3.5,-\a-0.55) node {$\color{black}\scriptstyle{2}$};
\draw [->] (4,-.4-\a) -- (5,-.4-\a);
\draw [white] (5.5,-\a-0.4) -- (5.5,-\a-0.4) node {$\color{black}\hat{\rvs}^n_{3k+2}$};
\draw [white] (2,-\a-0.15) -- (2,-\a-0.15) node {$\color{black}\rvf_{3k+2}$};
\draw [->] (3.5,-\a-1.8) -- (3.5,-\a-.8);
\draw [white] (4.2,-\a-1.3) -- (4.2,-\a-1.3) node {$\color{black}\rvs^n_{3k-1}$};
}

\draw [-] (2.5,-.4) -- (2.5,-1.7);
\draw [->] (2.5,-1.7) -- (3,-1.7);

\end{tikzpicture}
\caption{A Multi-terminal Source Coding Problem related to the proposed streaming setup. The erasure at time $t=3k$ leads to two virtual decoders with different side information as shown.}
\label{fig:multiterminal}
\end{center}
\end{figure}
Our first idea is to introduce a {\em periodic erasure channel} where every third packet gets erased i.e., $g_k =\star$ for $t=3k$, $k=0,1,2,\ldots$. 
We claim that even on such a channel every third source sequence must be recovered.
Suppose the destination does not receive $f_0$ but observes $g_1 = f_1$ and $g_2 = f_2$. It must recover $\rvs_2^n$ at $t=2$. At this point, because of the Markov nature of the source process, it becomes synchronized with the encoder i.e., the effect of earlier erasures is no longer relevant. Thus it treats the new erasure at time $t=3$ as the only erasure it has observed so far. Upon receiving $f_4$ and $f_5$  it must recover $\rvs_5^n$ at $t=5$. More generally, it is able to recover $\rvs_{3k+2}^n$ at $t=3k+2$ upon sequentially observing $\{f_{3i+1}, f_{3i+2}\}_{0 \le i \le k}$ and missing $\{f_{3i}\}_{0 \le i \le k}$. From the source coding theorem we must have
\begin{align}
2kR  &\ge H(f_1, f_2, f_4, f_5, \ldots, f_{3k-2},f_{3k-1})\\
&\ge H(\rvs_2^n, \rvs_5^n,\ldots, \rvs_{3k-1}^n)\\
&\ge n(k-1) H(\rvs_{3}|\rvs_0)
\end{align}
which, upon taking $k\rightarrow \infty$ yields $R \ge \frac{1}{2}H(\rvs_3 | \rvs_0)$.

The above argument only takes into account one constraint --- when there is an erasure, the destination needs to 
recover the source sequences with $W=1$.   Hence it is missing the  term of $\frac{1}{2}H(\rvs_1|\rvs_0,\rvs_2)$  that appears in~\eqref{eq:MC_Prop2}. To recover this term, we need to take into account for the second constraint --- in absence of erasures the destination must recover each source sequence instantaneously. 

Our second key idea is to introduce a multi-terminal source coding problem with one encoder and two  decoders that simultaneously captures both these constraints. This is illustrated in Fig.~\ref{fig:multiterminal}. 
The encoder is revealed $(\rvs_{3k+1}^n, \rvs_{3k+2}^n)$ and produces outputs $f_{3k+1}$ and $f_{3k+2}$.
Decoder $1$  needs to recover $\rvs_{3k+1}^n$ given $f_{3k+1}$ and $\rvs_{3k}^n$ while decoder $2$ needs to recover $\rvs_{3k+2}^n$
given $\rvs_{3k-1}^n$ and $(f_{3k+1}, f_{3k+2})$. Thus decoder $1$ corresponds to the steady state behaviour of the system when there is no loss while decoder $2$ corresponds to the recovery immediately after an erasure. For the above multi-terminal problem, we establish a simple lower bound on the symmetric rate $R = \frac{1}{n}H(f_{3k-1})= \frac{1}{n}H(f_{3k})$ as follows:
\begin{align}
&2nR \ge H(f_{3k+1}, f_{3k+2}) \notag\\ &\ge H(f_{3k+1},f_{3k+2}|\rvs_{3k-1}^n)\\
&= H(f_{3k+1}, f_{3k+2}, \rvs_{3k+2}^n | \rvs_{3k-1}^n)\notag\\&\quad- H(\rvs_{3k+2}^n | f_{3k+1}, f_{3k+2}, \rvs_{3k-1}^n)\\
&\ge H(f_{3k+1}, \rvs_{3k+2}^n|\rvs_{3k-1}^n) - n\eps_n \label{eq:Fano_S3k}\\
&\ge H(\rvs_{3k+2}^n|\rvs_{3k-1}^n) + H(f_{3k+1}|\rvs_{3k+2}^n, \rvs_{3k-1}^n)-n\eps_n\notag\\
&\ge H(\rvs_{3k+2}^n|\rvs_{3k-1}^n) \!+\! H(f_{3k+1}|\rvs_{3k+2}^n, \rvs_{3k}^n, \rvs_{3k-1}^n)-n\eps_n\label{eq:Cond_Red}\\
&\ge H(\rvs_{3k+2}^n|\rvs_{3k-1}^n) \!+\! H(\rvs^n_{3k+1}|\rvs_{3k+2}^n, \rvs_{3k}^n, \rvs_{3k-1}^n)-2n\eps_n\label{eq:Fano_S3k2}\\
&\ge H(\rvs_{3k+2}^n|\rvs_{3k-1}^n) \!+\! H(\rvs^n_{3k+1}|\rvs_{3k+2}^n, \rvs_{3k}^n)-2n\eps_n\label{eq:Markov_S3k2}\\
&= nH(\rvs_3|\rvs_0) + nH(\rvs_1 |\rvs_2, \rvs_0) -2n\eps_n \label{eq:3k_der}
\end{align}
where~\eqref{eq:Fano_S3k} follows from the fact that $\rvs_{3k+2}^n$ must be recovered from $(f_{3k+1}, f_{3k+2}, \rvs_{3k-1}^n)$ at decoder $2$ hence Fano's inequality applies and~\eqref{eq:Cond_Red} follows from the fact that conditioning reduces entropy.
~\eqref{eq:Fano_S3k2} follows from Fano's inequality applied to decoder $1$ and finally~\eqref{eq:Markov_S3k2} follows from the Markov chain associated with the source process.  Dividing throughout by $n$ in~\eqref{eq:3k_der}  and taking $n\rightarrow\infty$ recover the desired lower bound~\eqref{eq:MC_Prop2}. 

To apply the above lower bound in the streaming setup, we need to take into account that the decoders have access to codeword indices rather than side-information sequences. Furthermore the encoder has access to all the past source sequences.  The formal proof of the lower bound is presented in Sec.~\ref{sec:UBLB}. While inspired by the above ideas, it is somewhat more direct.

\subsection{Linear Semi-Deterministic Markov Sources }
\label{det-model-main-results}
We propose a special class of source models --- linear semi-deterministic Markov sources --- for which the lower bound in~\eqref{eq:genLB} is tight.  
Our proposed coding scheme is most natural for a subclass of deterministic sources defined below.

\begin{defn}{(Linear Diagonally Correlated Deterministic Sources)}
\label{def:Diagonal}
The alphabet of a {\em linear diagonally correlated deterministic source} consists of $K$ sub-symbols i.e.,
\begin{multline}
\bs_i = (\bs_{i,0},\ldots, \bs_{i,K}) \in \cS_0 \times \cS_1 \times \ldots \times \cS_K, \label{eq:S_dec}
\end{multline}
where each $\cS_i = \{0,1\}^{N_i}$ is a binary sequence. 
Suppose that the sub-sequence $\{\bs_{i,0}\}_{i \ge 0}$ is an i.i.d. sequence sampled uniformly over $\cS_0$ and for $1 \le j \le K$, the sub-symbol $\bs_{i,j}$ is a linear deterministic function\footnote{All multiplication is over the binary field.} of $\bs_{i-1, j-1}$ i.e., 
\begin{equation}
\bs_{i,j} = \bR_{j,j-1} \cdot \bs_{i-1,j-1}, \qquad 1\le j \le K. \label{eq:S_dec_2}
\end{equation}
for fixed  matrices $\bR_{1,0}, \bR_{2,1} \ldots, \bR_{K, K-1}$ each of full row-rank i.e., $\mrm{rank}(\bR_{j,j-1}) = N_j$. 
\end{defn}

For such a class of sources we establish that the lower bound in Theorem~\ref{thm:genUB_LB} is tight and the binning based scheme is sub-optimal.

\begin{prop}
\label{prop:LDD}
For the class of Linear\footnote{The assumption of linearity in Def.~\ref{def:Diagonal} is not required to achieve the lower bound. However we use linearity to generalize to the class of semi-deterministic sources in Thm.~\ref{thm:SemiDet}.} Diagonally Correlated Deterministic Sources in Def.~\ref{def:Diagonal} the rate-recovery function is  given by:
\begin{align}
&R(B,W) = R^-(B,W) \notag\\
&\hspace{1cm}= H(\rvbs_1 | \rvbs_0) + \frac{1}{W+1} I(\rvbs_B; \rvbs_{B+W+1}|\rvbs_0)\label{eq:LDD}\\
&\hspace{1cm}=N_0+\frac{1}{W+1}\sum_{k=1}^{\min\{K-W,B\}}N_{W+k}\label{eq:LDD2}. 
\end{align}
\end{prop}

Sec.~\ref{sec:deterministic} provides the proof of Prop.~\ref{prop:LDD}. Our coding scheme exploits the special structure of such sources and achieves a  rate that is strictly lower than the binning based scheme.  We call this technique {\em prospicient coding} because it exploits non-causal knowledge of some future symbols. We make the following remark, which will be established in the sequel.
\begin{remark}
\label{rem:K}
In the proof of the coding theorem for Prop.~\ref{prop:LDD}, it suffices to consider the case when $K = B+W$. The extension to the case when $K < B+W$ is trivial and the extension to the case when $K> B+W$ also follows in a straightforward manner. 
\end{remark}

The proposed coding scheme can also be generalized to a broader class of semi-deterministic sources. 
\begin{defn}{(Linear Semi-Deterministic Sources)}
\label{def:SemiDet}
The alphabet of a linear semi-deterministic source\footnote{Since each sub-symbol is a (fixed length) binary sequence we use the bold-face font $\bs_{i,j}$ to represent it. Similarly since each source symbol is a collection of sub-symbols we use a bold-face font to represent it. This should  not be confused with a length $n$ source sequence at time $i$, which will be represented as $\bs_i^n$. } consists of two sub-symbols i.e., 
\begin{equation}
\bs_i = (\bs_{i,0},\bs_{i,1}) \in \cS_0 \times \cS_1,
\end{equation}
where each $\cS_i = \{0,1\}^{N_i}$ for $i=0,1$. The sequence $\{\bs_{i,0}\}$ is an i.i.d.\ sequence sampled uniformly over $\cS_0$ whereas \begin{equation}
\bs_{i,1} = \begin{bmatrix}\bA &\bB\end{bmatrix} \cdot \begin{bmatrix} \bs_{i-1,0} \\ \bs_{i-1,1} \end{bmatrix}\label{eq:SemiDet}
\end{equation} for some fixed matrices $\bA$ and $\bB$.
\end{defn}

We show that through a suitable  linear transform, that is both invertible and memoryless, this apparently more general source model can be transformed into a diagonally correlated deterministic Markov source. The propsicient coding can be applied to this class. 

\begin{thm}
\label{thm:SemiDet}
For the class of Linear Semi-Deterministic Sources in Def.~\ref{def:SemiDet} the rate-recovery function is given by:
\begin{multline}
\label{eq:SemiDet-Rate}
R(B,W) = R^-(B,W) \\= H(\rvbs_1 | \rvbs_0) + \frac{1}{W+1} I(\rvbs_B; \rvbs_{B+W+1}|\rvbs_0).
\end{multline}
\end{thm}

The proof of Theorem~\ref{thm:SemiDet} is provided in Sec.~\ref{sec:LSD}.

\subsection{Gaussian Sources}
\label{sub:Gaussian}
Our proposed framework can be easily extended to a continuous valued source process with a fidelity measure. While a complete treatment of the lossy case is beyond the scope of the present paper, we study one natural extension of the diagonally correlated source model in Def.~\ref{def:Diagonal} to  Gaussian sources.  

Consider an Gaussian source process that is i.i.d.\ both in temporal and spatial dimensions.  
i.e., at time $i$, a sequence consisting of $n$ symbols $\rvs^n_{i}$, is sampled i.i.d.\ 
according to a zero mean unit variance Gaussian distribution $N(0,1)$. 

The encoder's output at time $i$ is denoted by the index $\rvf_{i} = \cF(\rvs_0^n, \ldots, \rvs_i^n) \in [1,2^{nR}]$ as before.  At time $i$,  upon receiving the channel outputs until time $i$, the decoder is interested in reproducing a collection of past $K$ sources. 

\begin{align}
\rvbt^n_{i}=\begin{pmatrix}
\rvs^n_{i}\\
\rvs^n_{i-1}\\
\vdots\\
\rvs^n_{i-K}\end{pmatrix}
\label{Gauss-model}
\end{align}
within a distortion vector $\bd=(d_0,d_1,\cdots,d_K)^{T}$. 

Thus for any $i \ge 0$ and $0 \le j \le K$,  if  $\hat{\rvs}_{i-j}^n$ is the reconstruction sequence of $\rvs_{i-j}^n$ at time $i$, we must have that $E\left[||\rvs_{i-j}^n- \hat{\rvs}_{i-j}^n||^2\right] \le nd_j$.  We will assume that ${d_0\le d_1 \le \cdots \le d_K}$ holds.  Furthermore as will be discussed in the coding theorem, it suffices to restrict $K = B+W$.
 
As before, the channel can introduce an  erasure-burst of length $B$ in an arbitrary interval $[j, j+B-1]$. The decoder is not required to output a reproduction of the sequences $\rvbt_i^n$ for $i \in [j, j+B+W-1]$. The {\em lossy rate-recovery function} denoted by $R(B, W, \bd)$ is the minimum rate required to satisfy these constraints.

\begin{thm}
For the Gaussian source model with a distortion vector $\bd = (d_0,\ldots, d_K)$ with $0 < d_i \le 1$,  the lossy recovery-rate function is given by\footnote{All logarithms are taken to base 2.}
\begin{align}
&R(B,W,\bd) \notag\\
&= \frac{1}{2}\log \bigg(\frac{1}{d_0}\bigg) + \frac{1}{W+1}\sum_{k=1}^{\min\{K-W,B\}}\frac{1}{2}\log \bigg(\frac{1}{d_{W+k}}\bigg). \label{eq:det-rate-Gaussian}
\end{align}
\label{thm:gauss-rate}
\end{thm}

The coding scheme for the proposed model involves using a successively refinable code for each sequence $\rvs_i^n$ to produce $K+1$ layers and mapping the sequence of layered codewords to a diagonally correlated deterministic source. The proof of Theorem~\ref{thm:gauss-rate} is provided in Sec.~\ref{sec:Gauss}.

\begin{figure}
\centering
\includegraphics[scale=0.6]{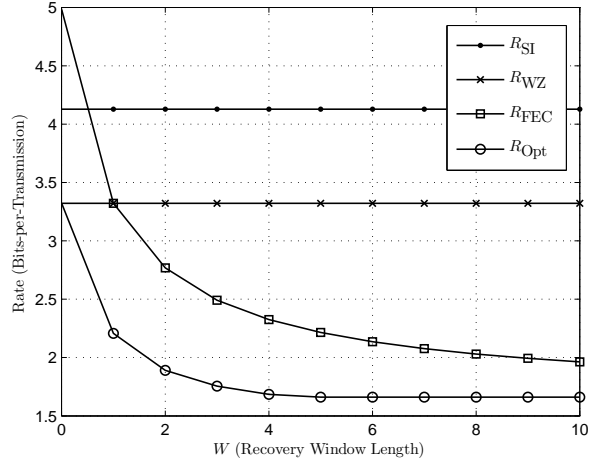}
\caption{Comparison of rate-recovery of sub-optimal systems to minimum possible rate-recovery function for different recovery window length $W$.}
\label{fig:comparison}
\end{figure} 

In Fig~\ref{fig:comparison}, the rate-recovery functions  of various schemes are compared. The sub-optimal schemes considered are \emph{Still-Image compression (SI)}, \emph{Wyner-Ziv Compression with delayed side-information (WZ)} and \emph{Predictive Coding plus FEC (FEC)} which are studied in detail in Sec.~\ref{subsec:comparison}. We assume $K=5$, ${B=2}$ and the distortion vector ${\bd=\{0.1,0.25,0.4,0.55,0.7,.85\}^{T}}$.  It can be observed from Fig~\ref{fig:comparison} that except when $W=0$ none of the other schemes are optimal. The \emph{Predictive Coding plus FEC} scheme, which is a natural { separation based scheme} is sub-optimal even for relatively large values of $W$.

\subsection{Symmetric Sources}
\label{subsec:symm}
A {\em symmetric}  source is defined as a Markov source such that the underlying Markov chain is also reversible i.e., the random variables  satisfy   $(\rvs_0,\ldots, \rvs_t) \stackrel{\text{d}}{=} (\rvs_t, \ldots, \rvs_0)$, where the equality is in the sense of distribution~\cite{markov}.  Of particular interest to us is the following property satisfied for each $t \ge 1$.
\begin{align}
p_{\rvs_{t+1}, \rvs_{t}} (s_a, s_b) = p_{\rvs_{t-1},\rvs_t}(s_a, s_b), \quad \forall s_a, s_b \in \cS\label{eq:symmetric}
\end{align}
i.e., we can ``exchange" the source pair $(\rvs_{t+1}^n,\rvs_t^n)$ with $(\rvs_{t-1}^n,\rvs_t^n)$ without affecting the joint distribution. 
An important class of sources that are symmetric are the binary sources: $\rvs_t^n = \rvs_{t-1}^n \oplus \rvz_t^n$, where $\{\rvz_t^n\}_{t\ge1}$ is  an i.i.d.\ binary source process (in both temporal and spatial  dimensions)  with the marginal distribution $\Pr(\rvz_{t,i}=0)=p$, the marginal distribution $\Pr(\rvs_{t,i}=0) = \Pr(\rvs_{t,i}=1) =\frac{1}{2}$ and $\oplus$ denotes modulo-2 addition. 

\begin{thm}
\label{thm:binning}
For the class of symmetric sources that satisfy~\eqref{eq:symmetric} the rate-recovery function, restricted to the class of memoryless encoders, is given by
\begin{align}
R(B,W) = \frac{1}{W+1} H(\rvs_{B+1}, \rvs_{B+2}, \ldots, \rvs_{B+W+1}|s_{0}).\label{eq:rate-rec-binning}
\end{align}
\end{thm}
Note that the achievability follows immediately from~\eqref{eq:genUB_Slepian-Wolf}. Thus it only remains to show that the lower bound~\eqref{eq:genLB} can to be improved.  We have only been able to obtain  this improvement for the class of memoryless encoders. For the general encoder structure~\eqref{eq:f-enc} this remains an open problem. At first glance one may expect that when memoryless encoders are considered, the binning based scheme is always optimal. Interestingly this is not true. The prospicient encoders for the diagonally correlated source models in section~\ref{det-model-main-results} are memoryless and yet improve upon the binning based lower bound.  Our proof only applies to the class of symmetric sources. 

Our proof presented in Sec.~\ref{sec:Symmetric} involves an interesting connection to a multi-terminal source coding problem called zig-zag source coding~\cite{berger:77,zigZag,oohama:gaussian}. In particular we develop a simple approach to lower bound  the sum-rate of a zig-zag source coding network with symmetric sources that may be of independent interest. 

\subsection{Delay-Constrained Decoder}
\label{subsec:Delay}
As a variation of the problem setup in Section~\ref{sec:statement} where instantaneous recovery of each source sequence is desired, we consider a {\em delay-constrained} decoder in this section.  A receiver with a delay constraint of $T$ recovers 
\begin{multline}
\hat{{s}}^{n}_i=\mathcal{G}_{i}(g_0, g_1, ..., g_i, \ldots, g_{i+T}),\\ i \notin \{j, \ldots, j+B+W-1\}\end{multline} if an erasure-burst of length $B$ occurs in the interval $[j,j+B-1]$. The rest of the setup does not change. The rate-recovery function is a function of three parameters: $W$, $B$ and $T$ i.e., $R(B, W, T)$. Note that $T=0$ reduces to the case treated in the rest of the paper. 
\begin{thm}
\label{thm:no_W}
The rate-recovery function, when $W=0$ is given by ${R_0}(B, T) = R(B, W=0, T)$ where
\begin{align}
{R_0}(B, T) &=\frac{1}{T+1}H(s_{B+1}|s_{0})+\frac{T}{T+1}H(s_{1}|s_{0}) \label{eq:R0_Cap1}\\
&= \frac{1}{T+1} H(s_{B+1}, s_{B+2}, \ldots s_{B+T+1}| s_0 ).
\label{eq:R0_Cap}
\end{align}
\end{thm}
The minimum rate is achieved by applying a Slepian-Wolf code to each source sequence and jointly decoding the source sequences $(\rvs^n_{j+B}, \rvs^n_{j+B+1},\ldots, \rvs^n_{j+B+T} )$ at time $j+B+T$, following a burst
in-between $(j, j+1, \ldots, j+B-1)$.

The complete proof of Theorem.~\ref{thm:no_W} is provided in Sec.~\ref{sec:no_W}.

\section{General Upper and Lower Bounds on Rate-Recovery Function}
\label{sec:UBLB}
We first establish the achievability of $R^+(B,W)$ and then the lower bound $R^-(B,W)$ in Theorem~\ref{thm:genUB_LB}.
\subsection{Achievability}
From Corollary.~\ref{corol:genUB} it suffices to show that\begin{equation}
R^+  = \frac{1}{W+1} H(\rvs_{B+1}, \ldots, \rvs_{B+W+1}|\rvs_0) + \eps\label{eq:SW2} 
\end{equation}
is achievable for any arbitrary $\eps > 0$.

We use a Slepian-Wolf codebook $\cC_i$ which is generated by randomly partitioning the set of all typical sequences~\cite{csiszarKorner}  
$T_\eps^n(\rvs)$ into $2^{nR^+}$ bins. For each $i \ge 0$ the partitioning is done independently and all the partitions are revealed to the  decoder ahead of time. 

Upon observing $\rvs_i^n$ the encoder declares an error if $\rvs_i^n \notin T_\eps^n(\rvs)$. Otherwise it finds the bin to which $\rvs_i^n$ belongs to and sends the corresponding bin index $f_i$. We consider two cases for recovering at time $t=i$. First, suppose that the sequence $\rvs_{i-1}^n$ has already been recovered. Then the destination attempts to recover $\rvs_i^n$ from  $(f_i, \rvs_{i-1}^n)$. This succeeds with high probability if $R^+ > H(\rvs_1 |\rvs_0)$, which is guaranteed via~\eqref{eq:SW2}. Next suppose that $\rvs_{i-1}^n$ has not been recovered by the destination but $\rvs_i^n$ needs to be recovered. This only happens when $\rvs_i^n$ is the first sequence to be recovered after the erasure burst. In particular  the erasure burst must happen between ${[ i-B'-W,i-W-1]}$ for some $B' \le B$. The decoder thus has access to $\rvs_{i-B'-W-1}^n$, before the start of the erasure burst. Upon receiving $f_{i-W},\ldots, f_i$ the destination simultaneously attempts to recover $(\rvs_{i-W}^n,\ldots, \rvs_i^n)$ given $(\rvs_{i-B-W-1}^n, f_{i-W},\ldots, f_i)$. This succeeds with high probability if,
\begin{align}
(W+1)R &= \sum_{j=i-W}^i H(f_j)\notag\\
& > H(\rvs_{i-W},\ldots, \rvs_i|\rvs_{i-W-B-1})
\end{align}which is also guaranteed by~\eqref{eq:SW2}.

\subsection{Lower Bound}
\label{subsec:genLB}

Our proof is an extension of the intuition developed in section~\ref{subsec:gen}.

For any sequence of $(n,2^{nR})$ codes we show that there is a sequence $\eps_n$ that vanishes as $n \rightarrow 0$ such that
\begin{align}
R \ge H(\rvs_1|\rvs_0) + \frac{1}{W+1} I(\rvs_p; \rvs_B| \rvs_0) - (W+1)\eps_n \label{eq:rate_UB}
\end{align}
where throughout we let $p=B+W+1$.

We consider a periodic erasure channel of period $p$ where the first $B$ packets are erased i.e., 
for each $k\ge 0$, suppose that  an erasure happens at time interval $t=\{kp, kp+1, \ldots, kp+B-1\}$. 
Consider:
\begin{align}
&(W+1)n(t+1)R \notag\\ 
&= H(\rvf_{B}^{p-1},\rvf_{p+B}^{2p-1},\ldots, \rvf_{(t-1)p+B}^{tp-1}, \rvf_{tp+B}^{(t+1)p-1})\\
&= H(\rvf_{B}^{p-1})+ \sum_{k=1}^t H(\rvf_{kp+B}^{(k+1)p-1} | \rvf_{B}^{p-1},\rvf_{p+B}^{2p-1},\ldots, \rvf_{(k-1)p+B}^{kp-1})\notag\\
&\ge H(\rvf_{B}^{p-1})+ \sum_{k=1}^t H(\rvf_{kp+B}^{(k+1)p-1} | \rvf_{0}^{kp-1}) \label{eq:periodic_LB}.
\end{align}
where the last step follows from the fact that conditioning reduces entropy. 

We bound the term $H(\rvf_{kp+B}^{(k+1)p-1} | \rvf_{0}^{kp-1})$ for each $k \ge 1$.
By definition, the source sequence $s_{(k+1)p-1}^n$ must be recovered from 
$(f_{0}^{kp-1}, f_{kp+B}, f_{kp+B+1}, ..., f_{(k+1)p-1})$. Applying  Fano's inequality we have that 
\begin{align}
H(\rvs_{(k+1)p-1}^n| \rvf_0^{kp-1},\rvf_{kp+B}, \ldots,\rvf_{(k+1)p-1}) \le n\eps_n \label{eq:Fano2}
\end{align} and
\begin{align}
&H(\rvf_{kp+B}^{(k+1)p-1}~|~\rvf_{0}^{kp-1}) \ge H(\rvs_{(k+1)p-1}^n~|~\rvf_0^{kp-1}) + \notag\\ &\qquad H(\rvf_{kp+B}^{(k+1)p-2}~|~\rvs_{(k+1)p-1}^n,\rvf_{0}^{kp-1}) - n\eps_n\label{eq:FanoApp}, 
\end{align}
where~\eqref{eq:FanoApp} follows from applying Fano's inequality of~\eqref{eq:Fano2}. Now we bound each of the two terms in~\eqref{eq:FanoApp}. First we note that:
\begin{align}
&\quad H(\rvs_{(k+1)p-1}^n|\rvf_0^{kp-1}) \ge H(\rvs_{(k+1)p-1}^n|\rvf_0^{kp-1}, \rvs_{kp-1}^n)\\
&= H(\rvs_{(k+1)p-1}^n|\rvs_{kp-1}^n) = nH(\rvs_{p}|\rvs_0),\label{eq:Markov2}
\end{align}where the last step follows from the Markov relation $\rvf_0^{kp-1} \rightarrow \rvs_{kp-1}^n \rightarrow \rvs_{(k+1)p-1}^n$.

Furthermore the second term in~\eqref{eq:FanoApp} can be lower bounded using the following series of inequalities.
\begin{align}
& H\left(\rvf_{kp+B}^{(k+1)p-2}~\big|~\rvs_{(k+1)p-1}^n,\rvf_{0}^{kp-1}\right) \label{eq:step1}\\
&\ge H\left(\rvf_{kp+B}^{(k+1)p-2}~\big|~\rvs_{(k+1)p-1}^n,\rvf_{0}^{kp+B-1}\right) \label{eq:cond_f3k}\\
\!&\!\ge \!H\left(\!\rvf_{kp+B}^{(k+1)p-2},\! \rvs^n_{kp+B}, \!\ldots,\!\rvs_{(k+1)p-2}^n\!\big|\!\rvs_{(k+1)p-1}^n,\!\rvf_{0}^{kp+B-1}\!\right) \notag\\ &- H\left(\rvs_{kp+B}^n, \ldots,\rvs_{(k+1)p-2}^n\big|\rvs_{(k+1)p-1}^n,\rvf_{0}^{(k+1)p-2}\right)\\
&\!=\!H\left(\!\rvf_{kp+B}^{(k+1)p-2},\! \rvs_{kp+B}^n, \! \ldots,\!\rvs_{(k+1)p-2}^n\!|\!\rvs_{(k+1)p-1}^n,\rvf_{0}^{kp+B-1}\!\right) \notag\\ &\qquad- Wn\eps_n \label{eq:FanoApp2}\\
&\!\ge \!H\left(\!\rvs_{kp+B}^n\!,\! \ldots\!,\!\rvs_{(k+1)p-2}^n\!\big|\!\rvs_{(k+1)p-1}^n,\!\rvf_{0}^{kp+B-1}\!\right)- Wn\eps_n \notag\\
&\!\!\ge\!\! H\!\left(\!\rvs_{kp+B}^n\!,\! \ldots\!,\!\rvs_{(k+1)p-2}^n\!\big|\!\rvs_{(k+1)p-1}^n\!,\rvf_{0}^{kp+B-1},\! \rvs_{kp+B-1}^n\!\right) \notag\\ &\qquad- Wn\eps_n \notag\\
&\!\!= \!\!H\!\left(\!\rvs_{kp+B}^n, \!\rvs_{kp+B+1}^n,\! \ldots,\!\rvs_{(k+1)p-2}^n\!\big|\!\rvs_{(k+1)p-1}^n,\! \rvs_{kp+B-1}^n\!\right) \notag\\ &\qquad- Wn\eps_n\label{eq:CondMarkov}\end{align}\begin{align}
&\!=\!\!n H(\rvs_{B+1}\!,\!\rvs_{B+2}\!,\!\ldots,\!\rvs_{p-1}\!|\!\rvs_B,\rvs_p\!)\!-\! Wn\eps_n\notag\\
&= \!n H(\rvs_{B+1}, \rvs_{B+2}, \ldots, \rvs_{p-1},\rvs_p |\rvs_B)\!\!- \!\!n H(\rvs_p |\rvs_B)\! -\! Wn\eps_n\notag\\
&= n (W+1)H(\rvs_1 |\rvs_0)- n H(\rvs_p |\rvs_B) - Wn\eps_n, \label{eq:term2}
\end{align} 
where~\eqref{eq:FanoApp2} follows from the fact that $\{\rvs_{kp+B+1}^n, \ldots,\rvs_{(k+1)p-2}^n\}$ must be decoded from $\rvf_0^{(k+1)p-2}$ and hence Fano's inequality again applies and~\eqref{eq:CondMarkov}  follows from the fact that 
\begin{align}
\rvf_0^{kp+B-1}\rightarrow \rvs_{kp+B-1}^n\rightarrow (\rvs_{kp+B}^n, \ldots,\rvs_{(k+1)p-2}^n).
\end{align} Combining~\eqref{eq:FanoApp},~\eqref{eq:Markov2} and~\eqref{eq:term2} we have that
\begin{align}
&H\!\left(\rvf_{kp+B}^{(k+1)p-1}~\big|~\rvf_{0}^{kp-1}\right)\notag\\  
&\ge n H(\rvs_p |\rvs_0) + n (W+1)H(\rvs_1 |\rvs_0)- n H(\rvs_p |\rvs_B)\notag\\
&\hspace{1cm} - (W+1)n\eps_n
\label{eq:LB_term}
\end{align}Finally substituting~\eqref{eq:LB_term} into~\eqref{eq:periodic_LB} we have that,
\begin{align}
&(W+1)n(t+1)R \notag\\ 
&\ge H(\rvf_1^{p-1}) - (W+1)n\eps_n + n t H(\rvs_{p}|\rvs_0) \notag\\ &\qquad  + nt (W+1)H(\rvs_1 |\rvs_0)- nt H(\rvs_p |\rvs_B) \label{eq:LB_bound}\\
&= H(\rvf_1^{p-1}) - (W+1)n\eps_n \notag\\ &\qquad + n t \left((W+1)H(\rvs_{1}|\rvs_0) +I(\rvs_{p};\rvs_{B}|\rvs_{0})\right)
\end{align} 
As we take $n\rightarrow \infty$ and then $t\rightarrow \infty$ we recover~\eqref{eq:rate_UB}.

\section{Diagonally Correlated Deterministic Sources}
\label{sec:deterministic}
We establish Prop.~\ref{prop:LDD} in this section.

\begin{figure*}
\vspace{1em}
\begin{center}
\input{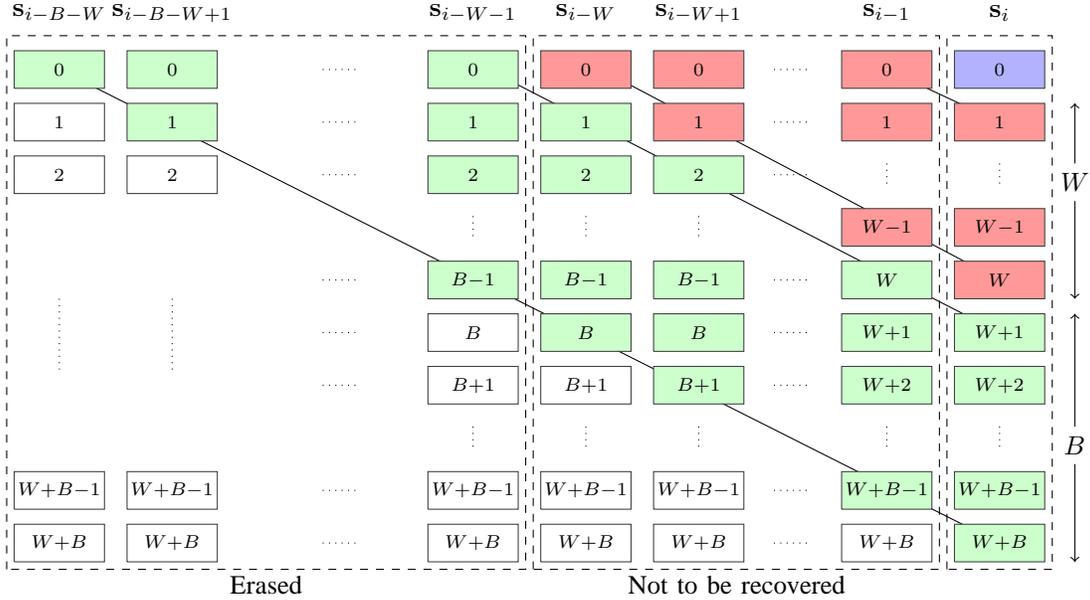}
\end{center}
\caption{Schematic of Diagonally Correlated Deterministic Markov Source. The first row of sub-symbols are innovation symbols. They are generated independently of all past symbols. On each diagonal the sub-symbol is a deterministic function of the sub-symbols above it.}
\label{Fig-Diag-Source}
\end{figure*}

\subsection{Source Model}
We  consider the semi-deterministic source model with a special diagonal correlation structure as described in Def.~\ref{def:Diagonal}. 
The diagonal correlation structure appears to be the most natural structure to consider in developing insights into our proposed coding scheme. 
As we will see later in Theorem~\ref{thm:SemiDet}, the underlying coding scheme can also be generalized to a broader class of linear semi-deterministic sources. Furthermore this class of semi-deterministic sources also provides a solution to the Gaussian source model as discussed in Theorem~\ref{thm:gauss-rate}.

We first provide an alternate characterization of the sources defined in Def.~\ref{def:Diagonal}. Let us define,

\begin{align}
\bR_{k,l}=\bR_{k,k-1}\bR_{k-1,k-2}\ldots\bR_{l+2,l+1}\bR_{l+1,l}\label{Prod},
\end{align}
where $k>l$. Note that since each $\bR_{j,j-1}$ is assumed to have a full row-rank (c.f. Def.~\ref{def:Diagonal}) the matrix $\bR_{k,l}$ is a $N_k \times N_l$ full-rank matrix of rank $N_k$. From Def.~\ref{def:Diagonal}

\begin{align}
\bs_{i}=\begin{pmatrix}
\bs_{i,0}\\
\bR_{1,0}\bs_{i-1,0}\\
\bR_{2,0}\bs_{i-2,0}\\
\vdots\\
\bR_{K,0}\bs_{i-K,0}
\end{pmatrix},\label{model-s}
\end{align}
where $\{\bs_{i-K,0}, \bs_{i-K+1,0},..., \bs_{i,0}\}$ are innovation sub-symbols of each source. This is expressed in Fig.~\ref{Fig-Diag-Source}. Any diagonal in Fig.~\ref{Fig-Diag-Source} consists of the same set of innovation bits.
In particular the innovation bits are introduced on the upper-left most entry of the diagonal. As we traverse down, each sub-symbol consists of some fixed linear combinations of these innovation bits. Furthermore the sub-symbol $\bs_{i,j}$ is completely determined given the sub-symbol $\bs_{i-1, j-1}$. 

In this section, we first argue that analyzing the coding scheme for the case $K=B+W$ is sufficient. Then we explain the prospicient coding scheme which achieves the rate specified in~\eqref{eq:LDD2}. Finally, the proof of the rate-optimality of the prospicient coding scheme is provided by establishing the equality of the rate expression~\eqref{eq:LDD2} and the general lower bound in~\eqref{eq:LDD}.  

\subsection{Sufficiency of $K=B+W$ }
\label{subsec:K_suff}
We first argue that for our coding scheme, it suffices to  assume that each source symbol $\bs_i$ consists of one innovation sub-symbol and a total of ${K=B+W}$ deterministic symbols. In particular when ${K<B+W}$, by simply adding ${K-B-W}$ zeros, the source can be turned into a source with ${B+W}$ deterministic sub-symbols. 

For the case ${K>B+W}$ we argue that it suffices to construct a coding scheme with $K=B+W$. The remainder of the sub-symbols can be trivially computed by the receiver.  In particular, at any time $i$, either $\bs_{i-1}$ or $\bs_{i-B-W-1}$ is guaranteed to be available to the destination. In the former case, except the innovation bits of $\bs_{i}$, all other bits are known. Thus all the deterministic sub-symbols, including those corresponding to ${K>B+W}$ can be computed.  In the latter case, because of the diagonal structure of the source, the sub-symbols $\bs_{i,j}$, for $j\ge B+W+1$,  are deterministic functions of $\bs_{i-B-W-1}$ (c.f.~\eqref{model-s}), and therefore, are known and can be ignored. Thus without loss of generality we assume that ${K=B+W}$ is sufficient.

\begin{figure*}
\begin{center}
\vspace{1em}
\noindent
\resizebox{5in}{3in}{
\begin{tikzpicture}
\fill[gray!30!white] (0,-1) -- (12,-9) -- (0,-9) -- (0,-1);

\draw [dashed](.85,-9.2) rectangle (5.15,.5);
\draw [dashed](5.35,-9.2) rectangle (12.65,.5);

\draw [-](1.5,-1) -- (13.5,-9);
\draw [-](4.5,-1) -- (13.5,-7);
\draw [-](6,-1) -- (15,-7);
\draw [-](7.5,-1) -- (15,-6);
\draw [-](10.5,-1) -- (12,-2);
\draw [-](13.5,-1) -- (15,-2);
\draw [-](12,-1) -- (13.5,-2);
\draw [-](13.5,-4) -- (15,-5);

\draw [white] (0,0) -- (0,0) node {$\color{black}s_{i-p}$};
\draw [white] (1.5,0) -- (1.5,0) node {$\color{black}s_{i-p+1}$};

\draw [dotted](2.5,0) -- (3.5,0);

\draw [white] (4.5,0) -- (4.5,0) node {$\color{black}s_{i-W-1}$};
\draw [white] (6,0) -- (6,0) node {$\color{black}s_{i-W}$};
\draw [white] (7.5,0) -- (7.5,0) node {$\color{black}s_{i-W+1}$};

\draw [dotted](8.5,0) -- (9.5,0);

\draw [white] (10.5,0) -- (10.5,0) node {$\color{black}s_{i-2}$};
\draw [white] (12,0) -- (12,0) node {$\color{black}s_{i-1}$};
\draw [white] (13.5,0) -- (13.5,0) node {$\color{black}s_{i}$};
\draw [white] (15,0) -- (15,0) node {$\color{black}s_{i+1}$};

\draw [white] (-.5,-1) -- (-.5,-1) node {$\color{black}{\textrm{\small 0}}$};
\draw [white] (-.5,-2) -- (-.5,-2) node {$\color{black}\textrm{\small 1}$};
\draw [white] (-.5,-4) -- (-.5,-4) node {$\color{black}\textrm{\small B}$};
\draw [white] (-.5,-5) -- (-.5,-5) node {$\color{black}\textrm{\small B+1}$};
\draw [white] (-.5,-6) -- (-.5,-6) node {$\color{black}\textrm{\small B+2}$};
\draw [white] (-.5,-8) -- (-.5,-8) node {$\color{black}\textrm{\small p-2}$};
\draw [white] (-.5,-9) -- (-.5,-9) node {$\color{black}\textrm{\small p-1}$};

\draw [fill=gray!30!white, draw=gray!50!black] (0,-1) circle (8pt);
\draw [fill=gray!30!white, draw=gray!50!black] (0,-2) circle (7.5pt);
\draw [dotted] (0,-2.8) -- (0,-3.2);
\draw [fill=gray!30!white, draw=gray!50!black] (0,-4) circle (6pt);
\draw [fill=gray!30!white, draw=gray!50!black] (0,-5) circle (5.5pt);
\draw [fill=gray!30!white, draw=gray!50!black] (0,-6) circle (5pt);
\draw [dotted] (0,-6.8) -- (0,-7.2);
\draw [fill=gray!30!white, draw=gray!50!black] (0,-8) circle (4pt);
\draw [fill=gray!30!white, draw=gray!50!black] (0,-9) circle (3.5pt);

\draw [fill=white, draw=gray!50!black] (1.5,-1) circle (8pt);
\draw [fill=gray!30!white, draw=gray!50!black] (1.5,-2) circle (7.5pt);
\draw [dotted] (1.5,-2.8) -- (1.5,-3.2);
\draw [fill=gray!30!white, draw=gray!50!black] (1.5,-4) circle (6pt);
\draw [fill=gray!30!white, draw=gray!50!black] (1.5,-5) circle (5.5pt);
\draw [fill=gray!30!white, draw=gray!50!black] (1.5,-6) circle (5pt);
\draw [dotted] (1.5,-6.8) -- (1.5,-7.2);
\draw [fill=gray!30!white, draw=gray!50!black] (1.5,-8) circle (4pt);
\draw [fill=gray!30!white, draw=gray!50!black] (1.5,-9) circle (3.5pt);

\draw [fill=white, draw=gray!50!black] (4.5,-1) circle (8pt);
\draw [fill=white, draw=gray!50!black] (4.5,-2) circle (7.5pt);
\draw [dotted] (4.5,-2.8) -- (4.5,-3.2);
\draw [fill=gray!30!white, draw=gray!50!black] (4.5,-4) circle (6pt);
\draw [fill=gray!30!white, draw=gray!50!black] (4.5,-5) circle (5.5pt);
\draw [fill=gray!30!white, draw=gray!50!black] (4.5,-6) circle (5pt);
\draw [dotted] (4.5,-6.8) -- (4.5,-7.2);
\draw [fill=gray!30!white, draw=gray!50!black] (4.5,-8) circle (4pt);
\draw [fill=gray!30!white, draw=gray!50!black] (4.5,-9) circle (3.5pt);

\filldraw [fill=green!10!white, draw=green!50!black] (6,-1) circle (8pt);
\draw [fill=white, draw=gray!50!black] (6,-2) circle (7.5pt);
\draw [dotted] (6,-2.8) -- (6,-3.2);
\draw [fill=white, draw=gray!50!black] (6,-4) circle (6pt);
\draw [fill=gray!30!white, draw=gray!50!black] (6,-5) circle (5.5pt);
\draw [fill=gray!30!white, draw=gray!50!black] (6,-6) circle (5pt);
\draw [dotted] (6,-6.8) -- (6,-7.2);
\draw [fill=gray!30!white, draw=gray!50!black] (6,-8) circle (4pt);
\draw [fill=gray!30!white, draw=gray!50!black] (6,-9) circle (3.5pt);

\draw [fill=red!40!white, draw=red!50!black] (7.5,-1) circle (8pt);
\draw [fill=white, draw=gray!50!black] (7.5,-2) circle (7.5pt);
\draw [dotted] (7.5,-2.8) -- (7.5,-3.2);
\draw [fill=white, draw=gray!50!black] (7.5,-4) circle (6pt);
\draw [fill=white, draw=gray!50!black] (7.5,-5) circle (5.5pt);
\draw [fill=gray!30!white, draw=gray!50!black] (7.5,-6) circle (5pt);
\draw [dotted] (7.5,-6.8) -- (7.5,-7.2);
\draw [fill=gray!30!white, draw=gray!50!black] (7.5,-8) circle (4pt);
\draw [fill=gray!30!white, draw=gray!50!black] (7.5,-9) circle (3.5pt);

\draw [dotted] (2.7,-1) -- (3.3,-1);
\draw [dotted] (2.7,-2) -- (3.3,-2);
\draw [dotted] (2.7,-4) -- (3.3,-4);
\draw [dotted] (2.7,-5) -- (3.3,-5);
\draw [dotted] (2.7,-6) -- (3.3,-6);
\draw [dotted] (2.7,-8) -- (3.3,-8);
\draw [dotted] (2.7,-9) -- (3.3,-9);

\draw [dotted] (8.7,-1) -- (9.3,-1);
\draw [dotted] (8.7,-2) -- (9.3,-2);
\draw [dotted] (8.7,-4) -- (9.3,-4);
\draw [dotted] (8.7,-5) -- (9.3,-5);
\draw [dotted] (8.7,-6) -- (9.3,-6);
\draw [dotted] (8.7,-8) -- (9.3,-8);
\draw [dotted] (8.7,-9) -- (9.3,-9);

\draw [fill=white, draw=gray!50!black] (12,-1) circle (8pt);
\draw [fill=white, draw=gray!50!black] (12,-2) circle (7.5pt);
\draw [dotted] (12,-2.8) -- (12,-3.2);
\draw [fill=white, draw=gray!50!black] (12,-4) circle (6pt);
\draw [fill=white, draw=gray!50!black] (12,-5) circle (5.5pt);
\draw [fill=white, draw=gray!50!black] (12,-6) circle (5pt);
\draw [dotted] (12,-6.8) -- (12,-7.2);
\draw [fill=white, draw=gray!50!black] (12,-8) circle (4pt);
\draw [fill=gray!30!white, draw=gray!50!black] (12,-9) circle (3.5pt);

\draw [fill=white, draw=gray!50!black] (10.5,-1) circle (8pt);
\draw [fill=white, draw=gray!50!black] (10.5,-2) circle (7.5pt);
\draw [dotted] (10.5,-2.8) -- (10.5,-3.2);
\draw [fill=white, draw=gray!50!black] (10.5,-4) circle (6pt);
\draw [fill=white, draw=gray!50!black] (10.5,-5) circle (5.5pt);
\draw [dotted] (10.5,-5.8) -- (10.5,-6.2);
\draw [fill=white, draw=gray!50!black] (10.5,-7) circle (4.5pt);
\draw [fill=gray!30!white, draw=gray!50!black] (10.5,-8) circle (4pt);
\draw [fill=gray!30!white, draw=gray!50!black] (10.5,-9) circle (3.5pt);

\draw [fill=white, draw=gray!50!black] (13.5,-1) circle (8pt);
\draw [fill=white, draw=gray!50!black] (13.5,-2) circle (7.5pt);
\draw [dotted] (13.5,-2.8) -- (13.5,-3.2);
\draw [fill=white, draw=gray!50!black] (13.5,-4) circle (6pt);
\draw [fill=white, draw=gray!50!black] (13.5,-5) circle (5.5pt);
\draw [fill=white, draw=gray!50!black] (13.5,-6) circle (5pt);
\filldraw [fill=green!10!white, draw=green!50!black] (13.5,-7) circle (4.5pt);

\draw [dotted] (13.5,-7.8) -- (13.5,-8.2);

\filldraw [fill=green!10!white, draw=green!50!black] (13.5,-9) circle (3.5pt);

\draw [fill=white, draw=gray!50!black] (15,-1) circle (8pt);
\draw [fill=white, draw=gray!50!black] (15,-2) circle (7.5pt);
\draw [dotted] (15,-2.8) -- (15,-3.2);
\draw [fill=white, draw=gray!50!black] (15,-4) circle (6pt);
\draw [fill=white, draw=gray!50!black] (15,-5) circle (5.5pt);
\draw [fill=white, draw=gray!50!black] (15,-6) circle (5pt);
\filldraw [fill=red!40!white, draw=red!50!black] (15,-7) circle (4.5pt);

\draw [dotted] (15,-7.8) -- (15,-8.2);

\filldraw [fill=red!40!white, draw=red!50!black] (15,-9) circle (3.5pt);

\draw (6,-3) ellipse (15pt and 40pt);
\draw (13.5,-8) ellipse (15pt and 40pt);
\draw (7.5,-3) ellipse (15pt and 40pt);
\draw (15,-8) ellipse (15pt and 40pt);

\draw [white] (3,-9.5) -- (3,-9.5) node {$\color{black}\textrm{Erased}$};
\draw [white] (9,-9.5) -- (9,-9.5) node {$\color{black}\textrm{Not to be recovered}$};

\draw [->](14,-3.8) -- (14,-2);
\draw [->](14,-4.2) -- (14,-6);
\draw [white] (14,-4) -- (14,-4) node {$\color{black}\textrm{W}$};

\draw [->](14,-7.8) -- (14,-7);
\draw [->](14,-8.2) -- (14,-9);
\draw [white] (14,-8) -- (14,-8) node {$\color{black}\textrm{B}$};

\end{tikzpicture}}
\caption{Schematic of Coding Scheme- Codeword structure. We set $p= B+W+1$. }
\label{Fig-S-1}
\end{center}
\end{figure*}
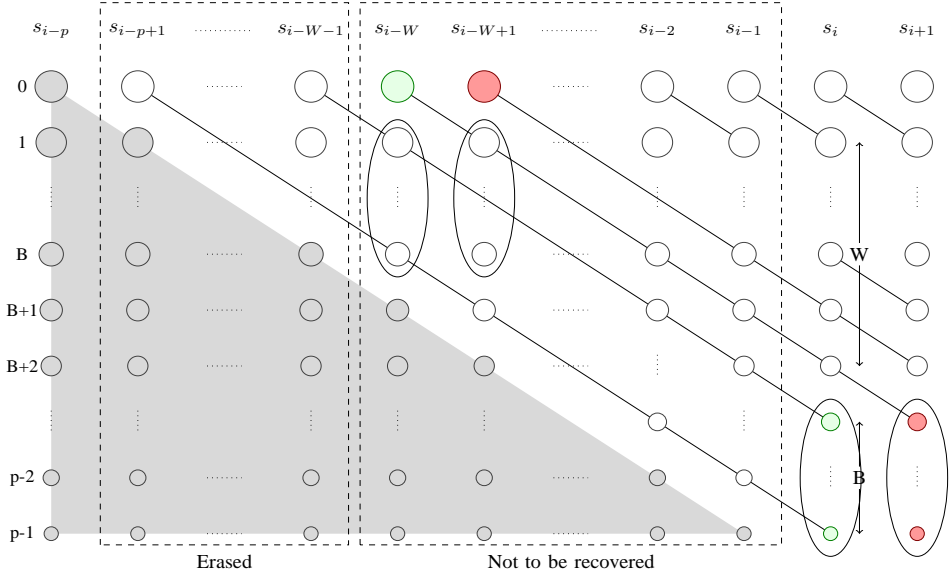

\subsection{Prospicient Coding}

Our coding scheme is based on the following observation, illustrated in  Fig.~\ref{Fig-S-1}.  
Suppose that an erasure happens between $t \in [i-W-B-1,i-W-1]$ and after the ``don't care'' period of $[i-W, i-1]$ we need to recover $\bs_i^n$. 
Based on the structure of the source, illustrated in Fig.~\ref{Fig-S-1} we make the following observations:
\begin{itemize}
\item Sub-symbols $\bs_{i,1},\ldots, \bs_{i,W}$ can be directly computed from the innovation sub-symbols $\bs_{i-1,0},\ldots, \bs_{i-W,0}$, respectively.
\item Sub-symbols $\bs_{i,W+1},\ldots, \bs_{i, W+B}$ can be computed from sub-symbols $\bs_{i-W,1},\ldots, \bs_{i-W, B}$, respectively. 
\end{itemize}

Thus if we send the first $B+1$ sub-symbols at each time i.e., $\bx_i = (\bs_{i,0},\ldots, \bs_{i,B})$ then we are guaranteed that the destination will be able to decode $\bs_i^n$ when an erasure happens between $[i-B-W, i-W-1]$. To achieve the optimal rate, we further compress $\bx_i$ as discussed below. Our coding scheme consists of two steps. 

\subsubsection{Source Rearrangement} The source symbols $\bs_i$ consisting of innovation and deterministic sub-symbols as in Def.~\ref{def:Diagonal} are first rearranged to produce an auxiliary set of  codewords
\begin{align}
{\bc}_{i}\!=\!
\begin{pmatrix}
{\bc}_{i,0}\\
{\bc}_{i,1}\\
{\bc}_{i,2}\\
\vdots\\
{\bc}_{i,B}\end{pmatrix}=\begin{pmatrix}
\bs_{i,0}\\
\bs_{i+W,W+1}\\
\bs_{i+W,W+2}\\
\vdots\\
\bs_{i+W,W+B}\end{pmatrix}\!
=\!\begin{pmatrix}
\bs_{i,0}\\
\bR_{W+1,1}\bs_{i,1}\\
\bR_{W+2,2}\bs_{i,2}\\
\vdots\\
\bR_{W+B,B}\bs_{i,B}\end{pmatrix},
\label{hat-c}
\end{align}
where the last relation follows from~\eqref{model-s}.

Note that the codeword ${\bc}_i$  consists of the innovation symbol  $\bs_{i,0}$,
as well as symbols  $\bs_{i+W,W+1},\ldots, \bs_{i+W,W+B}$ that enable the recovery of symbols in $\bs_{i+W}$.
The codeword ${\bc}_{i-W}$ consists of the green circles in Fig.~\ref{Fig-S-1}. 

\begin{figure*}
\begin{center}
\vspace{1em}
\begin{tikzpicture}

\draw [-] (0.5,-.25) -- (7.5,-3.75);
\draw [-] (1.9,-.25) -- (8.9,-3.75);
\draw [-] (0.5,-.95) -- (6.1,-3.75);

\fill [fill=green!20!white, draw=gray!50!black] (0,-.5) -- (0,0) -- (1,0) -- (1,-.5) -- (0,-.5);
\fill [fill=blue!35!white, draw=gray!50!black] (0,-1.2) -- (0,-.7) -- (1,-.7) -- (1,-1.2) -- (0,-1.2);
\fill [fill=white, draw=gray!50!black] (0,-1.9) -- (0,-1.4) -- (1,-1.4) -- (1,-1.9) -- (0,-1.9);
\draw [dotted] (0.5,-2.2) -- (0.5,-2.5);
\fill [fill=gray!15!white, draw=gray!50!black] (0,-3.3) -- (0,-2.8) -- (1,-2.8) -- (1,-3.3) -- (0,-3.3);
\fill [fill=white, draw=gray!50!black] (0,-4) -- (0,-3.5) -- (1,-3.5) -- (1,-4) -- (0,-4);

\fill [fill=red!50!white, draw=gray!50!black] (1.4,-.5) -- (1.4,0) -- (2.4,0) -- (2.4,-.5) -- (1.4,-.5);
\fill [fill=green!20!white, draw=gray!50!black] (1.4,-1.2) -- (1.4,-.7) -- (2.4,-.7) -- (2.4,-1.2) -- (1.4,-1.2);
\fill [fill=blue!35!white, draw=gray!50!black] (1.4,-1.9) -- (1.4,-1.4) -- (2.4,-1.4) -- (2.4,-1.9) -- (1.4,-1.9);
\draw [dotted] (1.9,-2.2) -- (1.9,-2.5);
\fill [fill=white, draw=gray!50!black] (1.4,-3.3) -- (1.4,-2.8) -- (2.4,-2.8) -- (2.4,-3.3) -- (1.4,-3.3);
\fill [fill=gray!15!white, draw=gray!50!black] (1.4,-4) -- (1.4,-3.5) -- (2.4,-3.5) -- (2.4,-4) -- (1.4,-4);

\fill [fill=white, draw=gray!50!black] (2.8,-.5) -- (2.8,0) -- (3.8,0) -- (3.8,-.5) -- (2.8,-.5);
\fill [fill=red!50!white, draw=gray!50!black] (2.8,-1.2) -- (2.8,-.7) -- (3.8,-.7) -- (3.8,-1.2) -- (2.8,-1.2);
\fill [fill=green!20!white, draw=gray!50!black] (2.8,-1.9) -- (2.8,-1.4) -- (3.8,-1.4) -- (3.8,-1.9) -- (2.8,-1.9);
\draw [dotted] (3.3,-2.2) -- (3.3,-2.5);
\fill [fill=white, draw=gray!50!black] (2.8,-3.3) -- (2.8,-2.8) -- (3.8,-2.8) -- (3.8,-3.3) -- (2.8,-3.3);
\fill [fill=white, draw=gray!50!black] (2.8,-4) -- (2.8,-3.5) -- (3.8,-3.5) -- (3.8,-4) -- (2.8,-4);

\draw [dotted] (4.6,-.25) -- (4.8,-.25);
\draw [dotted] (4.6,-.95) -- (4.8,-.95);
\draw [dotted] (4.6,-1.65) -- (4.8,-1.65);
\draw [dotted] (4.6,-3.05) -- (4.8,-3.05);
\draw [dotted] (4.6,-3.75) -- (4.8,-3.75);

\fill [fill=white, draw=gray!50!black] (5.6,-.5) -- (5.6,0) -- (6.6,0) -- (6.6,-.5) -- (5.6,-.5);
\fill [fill=white, draw=gray!50!black] (5.6,-1.2) -- (5.6,-.7) -- (6.6,-.7) -- (6.6,-1.2) -- (5.6,-1.2);
\fill [fill=white, draw=gray!50!black] (5.6,-1.9) -- (5.6,-1.4) -- (6.6,-1.4) -- (6.6,-1.9) -- (5.6,-1.9);
\draw [dotted] (6.1,-2.2) -- (6.1,-2.5);
\fill [fill=green!20!white, draw=gray!50!black] (5.6,-3.3) -- (5.6,-2.8) -- (6.6,-2.8) -- (6.6,-3.3) -- (5.6,-3.3);
\fill [fill=blue!35!white, draw=gray!50!black] (5.6,-4) -- (5.6,-3.5) -- (6.6,-3.5) -- (6.6,-4) -- (5.6,-4);

\fill [fill=white, draw=gray!50!black] (7,-.5) -- (7,0) -- (8,0) -- (8,-.5) -- (7,-.5);
\fill [fill=white, draw=gray!50!black] (7,-1.2) -- (7,-.7) -- (8,-.7) -- (8,-1.2) -- (7,-1.2);
\fill [fill=white, draw=gray!50!black] (7,-1.9) -- (7,-1.4) -- (8,-1.4) -- (8,-1.9) -- (7,-1.9);
\draw [dotted] (7.5,-2.2) -- (7.5,-2.5);
\fill [fill=red!50!white, draw=gray!50!black] (7,-3.3) -- (7,-2.8) -- (8,-2.8) -- (8,-3.3) -- (7,-3.3);
\fill [fill=green!20!white, draw=gray!50!black] (7,-4) -- (7,-3.5) -- (8,-3.5) -- (8,-4) -- (7,-4);

\fill [fill=white, draw=gray!50!black] (8.4,-.5) -- (8.4,0) -- (9.4,0) -- (9.4,-.5) -- (8.4,-.5);
\fill [fill=white, draw=gray!50!black] (8.4,-1.2) -- (8.4,-.7) -- (9.4,-.7) -- (9.4,-1.2) -- (8.4,-1.2);
\fill [fill=white, draw=gray!50!black] (8.4,-1.9) -- (8.4,-1.4) -- (9.4,-1.4) -- (9.4,-1.9) -- (8.4,-1.9);
\draw [dotted] (8.9,-2.2) -- (8.9,-2.5);
\fill [fill=white, draw=gray!50!black] (8.4,-3.3) -- (8.4,-2.8) -- (9.4,-2.8) -- (9.4,-3.3) -- (8.4,-3.3);
\fill [fill=red!50!white, draw=gray!50!black] (8.4,-4) -- (8.4,-3.5) -- (9.4,-3.5) -- (9.4,-4) -- (8.4,-4);

\draw [dotted] (10.2,-.25) -- (10.4,-.25);
\draw [dotted] (10.2,-.95) -- (10.4,-.95);
\draw [dotted] (10.2,-1.65) -- (10.4,-1.65);
\draw [dotted] (10.2,-3.05) -- (10.4,-3.05);
\draw [dotted] (10.2,-3.75) -- (10.4,-3.75);

\fill [fill=white, draw=gray!50!black] (11.2,-.5) -- (11.2,0) -- (12.2,0) -- (12.2,-.5) -- (11.2,-.5);
\fill [fill=white, draw=gray!50!black] (11.2,-1.2) -- (11.2,-.7) -- (12.2,-.7) -- (12.2,-1.2) -- (11.2,-1.2);
\fill [fill=white, draw=gray!50!black] (11.2,-1.9) -- (11.2,-1.4) -- (12.2,-1.4) -- (12.2,-1.9) -- (11.2,-1.9);
\draw [dotted] (11.7,-2.2) -- (11.7,-2.5);
\fill [fill=white, draw=gray!50!black] (11.2,-3.3) -- (11.2,-2.8) -- (12.2,-2.8) -- (12.2,-3.3) -- (11.2,-3.3);
\fill [fill=white, draw=gray!50!black] (11.2,-4) -- (11.2,-3.5) -- (12.2,-3.5) -- (12.2,-4) -- (11.2,-4);

\draw [white] (.5,0.5) -- (.5,0.5) node {$\color{black}\scriptstyle \bc_{i-W}$};
\draw [white] (0.5,-0.25) -- (0.5,-0.25) node {$\color{black}\scriptstyle 0$};
\draw [white] (0.5,-.95) -- (0.5,-.95) node {$\color{black}\scriptstyle 1$};
\draw [white] (0.5,-1.65) -- (0.5,-1.65) node {$\color{black}\scriptstyle 2$};
\draw [white] (0.5,-3.05) -- (0.5,-3.05) node {$\color{black}{\scriptstyle B-1}$};
\draw [white] (0.5,-3.75) -- (0.5,-3.75) node {$\color{black}\scriptstyle B$};

\draw [white] (1.9,0.5) -- (1.9,0.5) node {$\color{black}\scriptstyle \bc_{i-W+1}$};
\draw [white] (1.9,-0.25) -- (1.9,-0.25) node {$\color{black}\scriptstyle 0$};
\draw [white] (1.9,-.95) -- (1.9,-.95) node {$\color{black}\scriptstyle 1$};
\draw [white] (1.9,-1.65) -- (1.9,-1.65) node {$\color{black}\scriptstyle 2$};
\draw [white] (1.9,-3.05) -- (1.9,-3.05) node {$\color{black}{\scriptstyle B-1}$};
\draw [white] (1.9,-3.75) -- (1.9,-3.75) node {$\color{black}\scriptstyle B$};

\draw [white] (3.3,0.5) -- (3.3,0.5) node {$\color{black}\scriptstyle \bc_{i-W+2}$};
\draw [white] (3.3,-0.25) -- (3.3,-0.25) node {$\color{black}\scriptstyle 0$};
\draw [white] (3.3,-.95) -- (3.3,-.95) node {$\color{black}\scriptstyle 1$};
\draw [white] (3.3,-1.65) -- (3.3,-1.65) node {$\color{black}\scriptstyle 2$};
\draw [white] (3.3,-3.05) -- (3.3,-3.05) node {$\color{black}{\scriptstyle B-1}$};
\draw [white] (3.3,-3.75) -- (3.3,-3.75) node {$\color{black}\scriptstyle B$};

\draw [white] (6.1,0.5) -- (6.1,0.5) node {$\color{black}\scriptstyle \bc_{i-W+B}$};
\draw [white] (6.1,-0.25) -- (6.1,-0.25) node {$\color{black}\scriptstyle 0$};
\draw [white] (6.1,-.95) -- (6.1,-.95) node {$\color{black}\scriptstyle 1$};
\draw [white] (6.1,-1.65) -- (6.1,-1.65) node {$\color{black}\scriptstyle 2$};
\draw [white] (6.1,-3.05) -- (6.1,-3.05) node {$\color{black}{\scriptstyle B-1}$};
\draw [white] (6.1,-3.75) -- (6.1,-3.75) node {$\color{black}\scriptstyle B$};

\draw [white] (7.5,0.5) -- (7.5,0.5) node {$\color{black}\scriptstyle \bc_{i-W+B+1}$};
\draw [white] (7.5,-0.25) -- (7.5,-0.25) node {$\color{black}\scriptstyle 0$};
\draw [white] (7.5,-.95) -- (7.5,-.95) node {$\color{black}\scriptstyle 1$};
\draw [white] (7.5,-1.65) -- (7.5,-1.65) node {$\color{black}\scriptstyle 2$};
\draw [white] (7.5,-3.05) -- (7.5,-3.05) node {$\color{black}{\scriptstyle B-1}$};
\draw [white] (7.5,-3.75) -- (7.5,-3.75) node {$\color{black}\scriptstyle B$};

\draw [white] (8.9,0.5) -- (8.9,0.5) node {$\color{black}\scriptstyle  \bc_{i-W+B+2}$};
\draw [white] (8.9,-0.25) -- (8.9,-0.25) node {$\color{black}\scriptstyle 0$};
\draw [white] (8.9,-.95) -- (8.9,-.95) node {$\color{black}\scriptstyle 1$};
\draw [white] (8.9,-1.65) -- (8.9,-1.65) node {$\color{black}\scriptstyle 2$};
\draw [white] (8.9,-3.05) -- (8.9,-3.05) node {$\color{black}{\scriptstyle B-1}$};
\draw [white] (8.9,-3.75) -- (8.9,-3.75) node {$\color{black}\scriptstyle B$};

\draw [white] (11.7,0.5) -- (11.7,0.5) node {$\color{black} \scriptstyle \bc_{i}$};
\draw [white] (11.7,-0.25) -- (11.7,-0.25) node {$\color{black}\scriptstyle 0$};
\draw [white] (11.7,-.95) -- (11.7,-.95) node {$\color{black}\scriptstyle 1$};
\draw [white] (11.7,-1.65) -- (11.7,-1.65) node {$\color{black}\scriptstyle 2$};
\draw [white] (11.7,-3.05) -- (11.7,-3.05) node {$\color{black}{\scriptstyle B-1}$};
\draw [white] (11.7,-3.75) -- (11.7,-3.75) node {$\color{black}\scriptstyle B$};
\end{tikzpicture}
\caption{Schematic of Coding Scheme- Rate reduction.}
\label{Fig-Code-Cons}
\end{center}
\end{figure*}
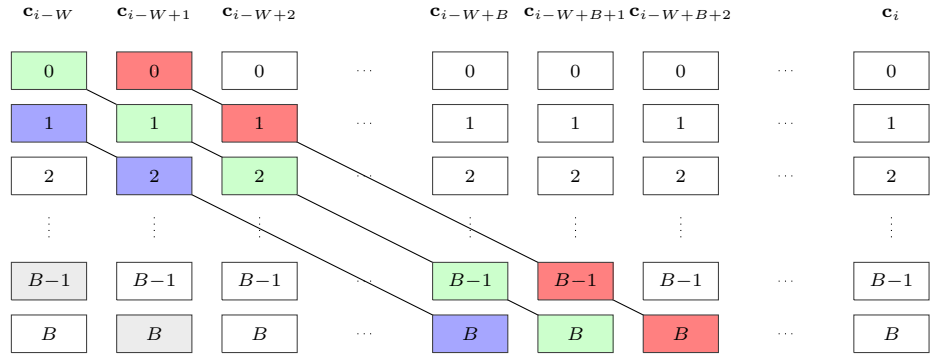

It can be verified from~\eqref{hat-c} that the rate associated with the codewords  $\bc_i$ is given by
\begin{align}
R_0=N_0+\sum_{k=W+1}^{W+B}N_i, \label{r-hat}
\end{align}
which is  larger than the rate-expression in~\eqref{eq:LDD2}. In particular it is missing the $\frac{1}{W+1}$  factor in the second term. This factor can be recovered by binning the sequences $\bc_i^n$ as described next. 

\subsubsection{Slepian-Wolf Coding} There is a strong temporal correlation between the sequences ${\bc}_i^n$ in~\eqref{hat-c}. 
As shown in Fig.~\ref{Fig-Code-Cons} as we proceed along any diagonal  the sub-symbols ${\bc}_{i,j}$ and ${\bc}_{i+1,j+1}$ contain the same underlying set of innovation bits  i.e., from sub-symbol $\bs_{i-j,0}$. 

To exploit the correlation, we independently bin the codeword sequences $\bc_i^n$ into $2^{nR}$ bins at each time. We let $R=R(B,W)+\eps$ is as given in~\eqref{eq:LDD2},  and only transmits the bin index of the associated codeword i.e., $\rvf_i = \cF(\bc_i^n) \in \{1,2, \ldots, 2^{nR}\}$.

It remains to show that given the bin index $\rvf_i$, the decoder is able to recover the underlying codeword symbols $\bc_i^n$. 


\subsubsection*{Analysis of Slepian-Wolf Coding}

Recall that we only transmit the bin index $\rvf_i$ of $\bc_i^n$. The receiver first recovers the underlying sequence $\bc_i^n$ as follows:
\begin{itemize}
\item [1)] If the receiver has access to $\bs^n_{i-1}$ in addition to  $\rvf_{i}$ it  can recover $\bc_{i}^n$ if 
\begin{align}
R&\ge H(\bc_{i}|\bs_{i-1})=H(\bc_{i,0})=N_0.
\label{eq:c-1}
\end{align}  
where the second equality follows since $\bc_{i,1}$, ..., $\bc_{i,W}$ are all deterministic functions of  $\bs_{i,1}$, ..., $\bs_{i,W}$, which in turn are deterministic functions of $\bs_{i-1}$. Clearly~\eqref{eq:c-1} is satisfied by our choice of $R$ in~\eqref{eq:LDD2}. 
\item [2)] The decoder has access to $\bs_{i-B-W-1}$ and $\{\rvf_{i-W}, \rvf_{i-W+1}, ..., \rvf_{i}\}$. The decoder is able to recover $\{\bc_{i-W}, ..., \bc_{i}\}$ if
\begin{align}
(W+1)R&\ge H(\bc_{i},\bc_{i-1},..., \bc_{i-W}|\bs_{i-B-W-1})\notag\\
&=\sum_{k=0}^{W+1}H(\bc_{i-k,0})+\sum_{k=1}^{B}H(\bc_{i-W,k})\label{sum}\\
&=(W+1)N_0+\sum_{k=1}^{B}N_{W+k}\label{eq:c-2},
\end{align}
where \eqref{sum} comes from the diagonal correlation property illustrated in  Fig.~\ref{Fig-Code-Cons}.
Our choice of $R$~\eqref{eq:LDD2} guarantees that~\eqref{eq:c-2} is satisfied. 
\end{itemize}

\subsection{Rate-Optimality of the Coding Scheme}

We specialize the general lower bound established in Theorem~\ref{thm:genUB_LB} 
to the case of diagonally correlated deterministic sources. 
Using~\eqref{model-s} and $p=B+W+1$ we have 
\begin{align}
R&\ge H(\bs_{1}|\bs_{0})+\frac{1}{W+1}I(\bs_{p};\bs_{B}|\bs_{0})\nonumber\\
&=H(\bs_{i}|\bs_{i-1})\!+\!\frac{1}{W+1}\left\{H(\bs_{i}|\bs_{i-p}) \notag - H(\bs_{i}|\bs_{i-W-1})\right\}\nonumber\\
&=H(\bs_{i}|\bs_{i-1})+ \notag\\
&\hspace{2pt}\frac{1}{W+1}H(\bs_{i,0},\bR_{1,0}\bs_{i-1,0},\ldots,\bR_{p-1,0}\bs_{i-p+1,0})\notag\\
&\hspace{2pt}- \frac{1}{W+1}H(\bs_{i,0},\bR_{1,0}\bs_{i-1,0},\ldots, \bR_{W,0}\bs_{i-W,0})\label{eq:connect}
\end{align}
According to the fact that innovation bits of each source are drawn i.i.d.\ \eqref{eq:connect} reduces to
\begin{align} 
R&\ge H(\bs_{i,0})+ \frac{1}{W+1}\left(H(\bs_{i,0})+\sum_{k=1}^{p-1}H(\bR_{k,0}\bs_{i-k,0})\right) \notag\\&-\frac{1}{W+1}\left(H(\bs_{i,0})+\sum_{k=1}^{W}H(\bR_{k,0}\bs_{i-k,0})\right)\label{indep}\\
&=N_0+ \frac{1}{W+1}  \left(\sum_{k=1}^{p-1}N_{k} -  \sum_{k=1}^{W}N_{k}\right)\label{N-j}\\
&=N_0 + \frac{1}{W+1}  \sum_{k=W+1}^{p-1}N_{k}\label{lower-dc},
\end{align}
where \eqref{N-j} follows from the fact that $\bR_{k,0}$ are $N_k\times N_0$ full-rank matrices of rank $N_k$. Since~\eqref{lower-dc} equals \eqref{eq:LDD2} for $K=B+W$ the optimality of the proposed scheme is established.

\section{Linear Semi-Deterministic Sources}
\label{sec:LSD}
We consider the class of linear deterministic sources as defined in Def.~\ref{def:SemiDet} in this section. Recall that for such a source the deterministic component $\bs_{i,d}$ is obtained from the previous sub-symbol $\bs_{i-1}$ through a linear transformation i.e.,
$$\bs_{i,d} = \begin{bmatrix}\bA & \bB \end{bmatrix} \begin{bmatrix} \bs_{i-1,0} \\ \bs_{i-1,d}\end{bmatrix}.$$

As discussed below, the transfer matrix $\begin{bmatrix}\bA & \bB \end{bmatrix}$ can be converted into a block-diagonal form through suitable invertible linear transformations, thus resulting in a diagonally correlated deterministic source. The prospicient coding scheme discussed earlier can then be applied to such a transformed source.

\subsection{Case 1}
Our transformation is most natural for the case when $\bA$ is a full row-rank matrix. So we treat this case first. Let
\begin{align}
N_1 \triangleq \textrm{Rank}(\bA)\le \min \{N_0,N_d\} \label{N1-min}.
\end{align}
In this section we restrict to the special case where $N_1=N_d$, i.e. $\bA$ is a full-row-rank matrix with $N_d$ independent non-zero rows. For this case, we explain the coding scheme by describing the encoder and decoder shown in Fig~\ref{fig:case1}.
\begin{figure}
\begin{center}
\vspace{1em}
\begin{tikzpicture}
\def \a {0}
\def \b {0}{
\draw [dashed] (-0.8,\b-.6) -- (-.8,\b+.6) -- (3.2,\b+.6) -- (3.2,\b-.6) -- (-0.8,\b-.6) ; 
\draw [white] (-.6,\b-.7) -- (-.6,\b-.7) node {$\color{black}\scriptstyle\textrm{Encoder}$};
\fill [fill=white, draw=gray!50!black] (\a-.5,\b-.5) -- (\a-.5,\b+.5) -- (\a+.5,\b+.5) -- (\a+.5,\b-.5) -- (\a-.5,\b-.5);
\draw [white] (\a-1.5,\b+0.25) -- (\a-1.5,\b+0.25) node {$\color{black} \rvbs^n_{i}$};
\draw [white] (\a,\b) -- (\a,\b) node {$\color{black}\scriptstyle\mathcal{L}$};
\draw [-] (-1.5+\a,\b) -- (-.5,\b);
}

\def \a {0}
\def \b {-2.5}{
\draw [dashed] (-0.8,\b-.6) -- (-.8,\b+.6) -- (3.2,\b+.6) -- (3.2,\b-.6) -- (-0.8,\b-.6) ; 
\draw [white] (-.6,\b+.7) -- (-.6,\b+.7) node {$\color{black}\scriptstyle\textrm{Decoder}$};
\fill [fill=white, draw=gray!50!black] (\a-.5,\b-.5) -- (\a-.5,\b+.5) -- (\a+.5,\b+.5) -- (\a+.5,\b-.5) -- (\a-.5,\b-.5);
\draw [white] (\a-1.5,\b+0.25) -- (\a-1.5,\b+0.25) node {$\color{black} \rvbs^n_{i}$};
\draw [white] (\a,\b) -- (\a,\b) node {$\color{black}\scriptstyle\mathcal{L}^{-1}$};
\draw [-] (-1.5+\a,\b) -- (-.5,\b);
}

\def \a {2.25}
\def \b {0}{
\fill [fill=white, draw=gray!50!black] (\a-.75,\b-.5) -- (\a-.75,\b+.5) -- (\a+.75,\b+.5) -- (\a+.75,\b-.5) -- (\a-.75,\b-.5);
\draw [white] (\a-1.25,\b+0.25) -- (\a-1.25,\b+0.25) node {$\color{black} \tilde{\rvbs}^n_{i}$};
\draw [white] (\a,\b+.2) -- (\a,\b+.2) node {$\color{black}\scriptstyle\textrm{Prospicient}$};
\draw [white] (\a,\b-.2) -- (\a,\b-.2) node {$\color{black}\scriptstyle\textrm{Encoder}$};
\draw [-] (\a-1.75,\b) -- (\a-.75,\b);
\draw [-] (\a+.75,\b) -- (\a+1.75,\b);
\draw [-] (\a+1.75,\b) -- (\a+1.75,\b-.5);
}

\def \a {2.25}
\def \b {-2.5}{
\fill [fill=white, draw=gray!50!black] (\a-.75,\b-.5) -- (\a-.75,\b+.5) -- (\a+.75,\b+.5) -- (\a+.75,\b-.5) -- (\a-.75,\b-.5);
\draw [white] (\a-1.25,\b+0.25) -- (\a-1.25,\b+0.25) node {$\color{black} \tilde{\rvbs}^n_{i}$};
\draw [white] (\a,\b+.2) -- (\a,\b+.2) node {$\color{black}\scriptstyle\textrm{Prospicient}$};
\draw [white] (\a,\b-.2) -- (\a,\b-.2) node {$\color{black}\scriptstyle\textrm{Decoder}$};
\draw [-] (\a-1.75,\b) -- (\a-.75,\b);
\draw [-] (\a+.75,\b) -- (\a+1.75,\b);
\draw [-] (\a+1.75,\b) -- (\a+1.75,\b+.5);
}

\def \a {4}
\def \b {-1.25}{
\fill [fill=gray!10!white, draw=gray!50!black] (\a-.5,\b-.75) -- (\a-.5,\b+.75) -- (\a+.5,\b+.75) -- (\a+.5,\b-.75) -- (\a-.5,\b-.75);
\draw [white] (\a,\b+.3) -- (\a,\b+.3) node {$\color{black}\scriptstyle\textrm{Burst}$};
\draw [white] (\a,\b) -- (\a,\b) node {$\color{black}\scriptstyle\textrm{Erasure}$};
\draw [white] (\a,\b-.3) -- (\a,\b-.3) node {$\color{black}\scriptstyle\textrm{Channel}$};
}
\end{tikzpicture}
\caption{ Block diagram of the system described in Case 1.}
\label{fig:case1}
\end{center}
\end{figure}
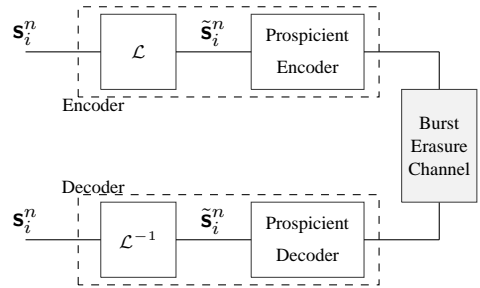 
\subsubsection{Encoder}
As in Fig.~\ref{fig:case1}, the encoder applies a memoryless transformation block $\mathcal{L}(.)$ onto each symbol $\bs_i$ to yield $\tilde{\bs}_{i} = \cL(\bs_i)$, a diagonally correlated deterministic source. We discuss the $\mathcal{L}(\cdot)$ mapping below.

Suppose that $\bX$ is  a matrix of dimensions $N_0\times N_d$. Define
\begin{align}
\bM \triangleq \begin{bmatrix}
\bI& \bX\\
\mathrm{0}& \bI
\end{bmatrix}
\end{align}
and observe that
\begin{align}
\bM^{-1}= \begin{bmatrix}
\bI& -\bX\\
\mathrm{0}& \bI
\end{bmatrix}.
\end{align}

For a certain $\bX$  to be specified later, let
\begin{align}
\bs_{i,d} &=\begin{bmatrix}
\bA& \bB
\end{bmatrix}\bM^{-1}\bM\begin{bmatrix}
\bs_{i-1,0}\\
\bs_{i-1,d}
\end{bmatrix}\\
&=\begin{bmatrix}
\bA& \bB-\bA\bX
\end{bmatrix}\begin{bmatrix}
\bs_{i-1,0}+\bX \bs_{i-1,d}\\
\bs_{i-1,d}
\end{bmatrix}\label{case-one-1}
\end{align} Since $\bA$ is a full-rank matrix, we may select $\bX$ such that 
\begin{align}
\bB-\bA\bX=\mathrm{0}
\end{align} With this choice of $\bX$, \eqref{case-one-1} reduces to 
\begin{align}
\bs_{i,d}&=\begin{bmatrix}
\bA& \mathrm{0}
\end{bmatrix}\begin{bmatrix}
\bs_{i-1,0}+\bX \bs_{i-1,d}\\
\bs_{i-1,d}
\end{bmatrix}
\end{align} Now, define the linear transformation $\mathcal{L}(.)$ as follows. 
\begin{align}
\tilde{\bs}_{i}=\begin{pmatrix}
\tilde{\bs}_{i,0}\\
\tilde{\bs}_{i,1}
\end{pmatrix}\triangleq \begin{pmatrix}
\bs_{i,0}+\bX \bs_{i,d}\\
\bs_{i,d}
\end{pmatrix}
\end{align} Note that 1) The transformation $\mathcal{L}(.)$ is memoryless and requires no knowledge of the past source sequences, 2) The innovation bits $\bs_{i,0}$ are independently drawn and independent of $\bs_{i,d}$. Hence $\tilde{\bs}_{i,0}$ are drawn i.i.d.\ according to Bernoulli-$(1/2)$, and  are independent of $\bs_{i,d}$, 3) The map between the two sources $\bs_{i}$ and $\tilde{\bs}_{i}$ are one-to-one.

Observe that $\tilde{\bs}_{i}$ is diagonally correlated Markov source with $N_0$ innovation bits $\tilde{\bs}_{i,0}$ and $N_d$ deterministic bits $\tilde{\bs}_{i,1}$ that satisfy
\begin{align}
\tilde{\bs}_{i,1}=\bA\tilde{\bs}_{i-1,0}.
\end{align}
We transmit the source sequence $\{\tilde{\bs}_i\}$ using the prospicient coding scheme. 
\subsubsection{Decoder}
At the receiver, first the Prospicient decoder recovers the diagonally correlated source $\tilde{\bs}_{i}$ at any time except error propagation window. Then whenever $\tilde{\bs}_{i}$ is available, the decoder directly constructs $\bs_{i}$ as
\begin{align}
\bs_{i}=\mathcal{L}^{-1}(\tilde{\bs}_{i})=\bM^{-1}\tilde{\bs}_{i}.
\end{align}

\subsubsection{Rate-optimality}
\label{sec:rate-optimal}
Suppose that our two step approach in Fig.~\ref{fig:case1} is sub-optimal. 
Then, in order to transmit the $\tilde{\bs}_{i}$ through the channel, one can first transform it into ${\bs}_{i}$ via $\mathcal{L}^{-1}$ and achieve lower rate than the prospicient coding scheme. However this is impossible because prospicient scheme is optimal. This shows the optimality of the coding scheme.  

\subsection{Case 2}
Now we consider the general case of semi-deterministic Markov sources defined in Def.~\ref{def:SemiDet}. 
As illustrated in Fig.~\ref{fig:LfLb} the reduction to the diagonally correlated source is done in two steps using two linear transforms: $\cL_f(\cdot)$ and $\cL_b(\cdot)$.
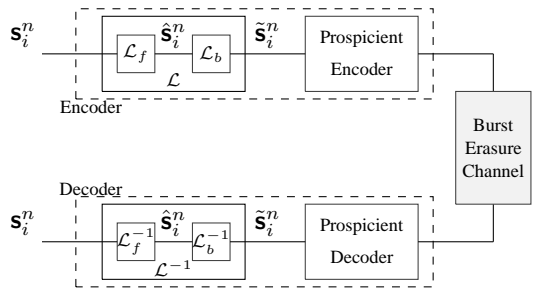
\begin{figure}
\begin{center}
\vspace{1em}
\begin{tikzpicture}
\def \a {0}
\def \b {0}{
\draw [dashed] (-0.8,\b-.6) -- (-.8,\b+.6) -- (3.95,\b+.6) -- (3.95,\b-.6) -- (-0.8,\b-.6) ;
\draw [-] (-0.45,\b-.5) -- (-.45,\b+.5) -- (1.45,\b+.5) -- (1.45,\b-.5) -- (-0.45,\b-.5) ;  
\draw [white] (.5,\b-.35) -- (.5,\b-.35) node {$\color{black}\scriptstyle\mathcal{L}$};
\draw [white] (-.6,\b-.7) -- (-.6,\b-.7) node {$\color{black}\scriptstyle\textrm{Encoder}$};
\fill [fill=white, draw=gray!50!black] (\a-.25,\b-.25) -- (\a-.25,\b+.25) -- (\a+.25,\b+.25) -- (\a+.25,\b-.25) -- (\a-.25,\b-.25);
\draw [white] (\a-1.5,\b+0.25) -- (\a-1.5,\b+0.25) node {$\color{black} \rvbs^n_{i}$};
\draw [white] (\a,\b) -- (\a,\b) node {$\color{black}\scriptstyle\mathcal{L}_f$};
\draw [-] (-1.25+\a,\b) -- (-.25,\b);
}

\def \a {1}
\def \b {0}{
\fill [fill=white, draw=gray!50!black] (\a-.25,\b-.25) -- (\a-.25,\b+.25) -- (\a+.25,\b+.25) -- (\a+.25,\b-.25) -- (\a-.25,\b-.25);
\draw [white] (\a-.5,\b+0.25) -- (\a-.5,\b+0.25) node {$\color{black} \hat{\rvbs}^n_{i}$};
\draw [white] (\a,\b) -- (\a,\b) node {$\color{black}\scriptstyle\mathcal{L}_b$};
\draw [-] (-.75+\a,\b) -- (\a-.25,\b);
}

\def \a {0}
\def \b {-2.5}{
\draw [dashed] (-0.8,\b-.6) -- (-.8,\b+.6) -- (3.95,\b+.6) -- (3.95,\b-.6) -- (-0.8,\b-.6) ; 
\draw [-] (-0.45,\b-.5) -- (-.45,\b+.5) -- (1.45,\b+.5) -- (1.45,\b-.5) -- (-0.45,\b-.5) ;
\draw [white] (.5,\b-.35) -- (.5,\b-.35) node {$\color{black}\scriptstyle\mathcal{L}^{-1}$}; 
\draw [white] (-.6,\b+.7) -- (-.6,\b+.7) node {$\color{black}\scriptstyle\textrm{Decoder}$};
\fill [fill=white, draw=gray!50!black] (\a-.25,\b-.25) -- (\a-.25,\b+.25) -- (\a+.25,\b+.25) -- (\a+.25,\b-.25) -- (\a-.25,\b-.25);
\draw [white] (\a-1.5,\b+0.25) -- (\a-1.5,\b+0.25) node {$\color{black} \rvbs^n_{i}$};
\draw [white] (\a,\b) -- (\a,\b) node {$\color{black}\scriptstyle\mathcal{L}^{-1}_f$};
\draw [-] (-1.25+\a,\b) -- (-.25,\b);
}

\def \a {1}
\def \b {-2.5}{
\fill [fill=white, draw=gray!50!black] (\a-.25,\b-.25) -- (\a-.25,\b+.25) -- (\a+.25,\b+.25) -- (\a+.25,\b-.25) -- (\a-.25,\b-.25);
\draw [white] (\a-.5,\b+0.25) -- (\a-.5,\b+0.25) node {$\color{black} \hat{\rvbs}^n_{i}$};
\draw [white] (\a,\b) -- (\a,\b) node {$\color{black}\scriptstyle\mathcal{L}_b^{-1}$};
\draw [-] (-.75+\a,\b) -- (\a-.25,\b);
}
\def \a {3}
\def \b {0}{
\fill [fill=white, draw=gray!50!black] (\a-.75,\b-.5) -- (\a-.75,\b+.5) -- (\a+.75,\b+.5) -- (\a+.75,\b-.5) -- (\a-.75,\b-.5);
\draw [white] (\a-1.25,\b+0.25) -- (\a-1.25,\b+0.25) node {$\color{black} \tilde{\rvbs}^n_{i}$};
\draw [white] (\a,\b+.2) -- (\a,\b+.2) node {$\color{black}\scriptstyle\textrm{Prospicient}$};
\draw [white] (\a,\b-.2) -- (\a,\b-.2) node {$\color{black}\scriptstyle\textrm{Encoder}$};
\draw [-] (\a-1.75,\b) -- (\a-.75,\b);
\draw [-] (\a+.75,\b) -- (\a+1.75,\b);
\draw [-] (\a+1.75,\b) -- (\a+1.75,\b-.5);
}

\def \a {3}
\def \b {-2.5}{
\fill [fill=white, draw=gray!50!black] (\a-.75,\b-.5) -- (\a-.75,\b+.5) -- (\a+.75,\b+.5) -- (\a+.75,\b-.5) -- (\a-.75,\b-.5);
\draw [white] (\a-1.25,\b+0.25) -- (\a-1.25,\b+0.25) node {$\color{black} \tilde{\rvbs}^n_{i}$};
\draw [white] (\a,\b+.2) -- (\a,\b+.2) node {$\color{black}\scriptstyle\textrm{Prospicient}$};
\draw [white] (\a,\b-.2) -- (\a,\b-.2) node {$\color{black}\scriptstyle\textrm{Decoder}$};
\draw [-] (\a-1.75,\b) -- (\a-.75,\b);
\draw [-] (\a+.75,\b) -- (\a+1.75,\b);
\draw [-] (\a+1.75,\b) -- (\a+1.75,\b+.5);
}

\def \a {4.75}
\def \b {-1.25}{
\fill [fill=gray!10!white, draw=gray!50!black] (\a-.5,\b-.75) -- (\a-.5,\b+.75) -- (\a+.5,\b+.75) -- (\a+.5,\b-.75) -- (\a-.5,\b-.75);
\draw [white] (\a,\b+.3) -- (\a,\b+.3) node {$\color{black}\scriptstyle\textrm{Burst}$};
\draw [white] (\a,\b) -- (\a,\b) node {$\color{black}\scriptstyle\textrm{Erasure}$};
\draw [white] (\a,\b-.3) -- (\a,\b-.3) node {$\color{black}\scriptstyle\textrm{Channel}$};
}
\end{tikzpicture}
\caption{Block diagram for general linear semi-deterministic Markov sources.}
\label{fig:LfLb}
\end{center}
\end{figure}

\begin{lemma}
\label{lem:full-rank}
Any semi-deterministic Markov source specified in Def.~\ref{def:SemiDet}, or equivalently by~\eqref{eq:SemiDet}, can be transformed into an equivalent source $\hat{\bs}_{i}$ consisting of innovation component $\bs_{i,0}\in \{0,1\}^{N_{0}}$ and $K$ deterministic components that satisfy~\eqref{full-rank-rep}, at the top of the next page, 
\begin{figure*}[!htb]
\begin{align}
&\hat{\bs}_{i,d}=\begin{pmatrix}
{\bs}_{i,1}\\
{\bs}_{i,2}\\
\vdots\\
{\bs}_{i,K-1}\\
{\bs}_{i,K}
\end{pmatrix}=\begin{pmatrix}
\bR_{1,0} & \bR_{1,1} & \cdots & \bR_{1,K-2} &\bR_{1,K-1} &\bR_{1,K}\\
\mathrm{0} & \bR_{2,1} & \cdots & \bR_{2,K-2} & \bR_{2,K-1} &\bR_{2,K}\\
\vdots & \vdots & \ddots & \vdots &\vdots &\vdots\\
\mathrm{0} & \mathrm{0} & \cdots & \bR_{K-1,K-2} &\bR_{K-1,K-1}& \bR_{K-1,K}\\
\mathrm{0} & \mathrm{0} & \cdots & \mathrm{0} & {\bR}_{K,K-1} &{\bR}_{K,K}\\
\end{pmatrix}\begin{pmatrix}
{\bs}_{i-1,0}\\
{\bs}_{i-1,1}\\
\vdots\\
{\bs}_{i-1,K-2}\\
{\bs}_{i-1,K-1}\\
{\bs}_{i-1,K}
\end{pmatrix}. \label{full-rank-rep}
\end{align}
\begin{align}
&\tilde{\bs}_{i,d}=\begin{pmatrix}
\tilde{\bs}_{i,1}\\
\tilde{\bs}_{i,2}\\
\vdots\\
\tilde{\bs}_{i,K-1}\\
\tilde{\bs}_{i,K}
\end{pmatrix}=\begin{pmatrix}
\bR_{1,0} & \mathrm{0} & \cdots & \mathrm{0} &\mathrm{0} \\
\mathrm{0} & \bR_{2,1} & \cdots & \mathrm{0} & \mathrm{0} \\
\vdots & \vdots & \ddots & \vdots &\vdots \\
\mathrm{0} & \mathrm{0} & \cdots & \bR_{K-1,K-2} &\mathrm{0} \\
\mathrm{0} & \mathrm{0} & \cdots & \mathrm{0} & {\bR}_{K,K-1} 
\end{pmatrix}\begin{pmatrix}
\tilde{\bs}_{i-1,0}\\
\tilde{\bs}_{i-1,1}\\
\vdots\\
\tilde{\bs}_{i-1,K-2}\\
\tilde{\bs}_{i-1,K-1}
\end{pmatrix}. \label{block-diag-rep-2}
\end{align}        
\end{figure*} using a one-to-one linear transformation $\cL_f$
where \begin{enumerate}
\item ${\bs}_{i,j}\in \{0,1\}^{N_j}$ for $j \in\{0,\ldots, K\}$ where 
\begin{align}
N_0 \ge N_1 \ge \ldots \ge N_K, \label{ranks}
\end{align}
and $\sum_{k=1}^{K}N_k=N_d$.
\item$\bR_{j,j-1}$ is $N_j\times N_{j-1}$ full-rank matrix of rank $N_j$ for $j \in\{1,\ldots, K-1\}$.
\item The matrix $\bR_{K,K-1}$ is either full-rank of rank $N_K$ or zero matrix.
\end{enumerate}
\end{lemma}  
The transformation to $\hat{\bs}_i$ involves repeated application of the technique in case 1.
The proof is provided in Appendix~\ref{app:full-rank}. The proof provides an explicit construction of $\cL_f$.

\begin{lemma}
\label{lem:bf}
Consider the source $\hat{\bs}_{i}=\mathcal{L}_f(\bs_{i})$ where $\bs_{i}$ is a semi-deterministic Markov source and $\hat{\bs}_i$ is defined in~\eqref{full-rank-rep}. There exists a one-to-one linear transformation $\mathcal{L}_{b}$ which maps $\hat{\bs}_i$ to a diagonally correlated deterministic Markov source $\tilde{\bs}_{i}$ that satisfies~\eqref{block-diag-rep-2}.
\end{lemma}

To illustrate  the idea, here we study a simple example. The complete proof is available in Appendix~\ref{app:bf}. Assume $K=2$ and consider the source $\hat{\bs}_{i}$ consisting of $N_0$ innovation bits $\bs_{i,0}$ and $N_1+N_2$ deterministic bits as 
\begin{align}
\hat{\bs}_{i,d}=\begin{pmatrix}
\bs_{i,1}\\
\bs_{i,2}
\end{pmatrix}=\begin{pmatrix}
\bR_{1,0}& \bR_{1,1} & \bR_{1,2}\\
\mathrm{0}& \bR_{2,1} & \bR_{2,2}
\end{pmatrix}\begin{pmatrix}
\bs_{i-1,0}\\
\bs_{i-1,1}\\
\bs_{i-1,2}
\end{pmatrix}\label{det-2}
\end{align} 
where $\bR_{1,0}$ and $\bR_{2,1}$ are full-rank (non-zero) matrices of rank $N_1$ and $N_2$, respectively.


The following steps transforms the source $\hat{\bs}_{i}$ into diagonally correlated Markov source. \\
{\bf Step $1$: } Define 
\begin{align}
\begin{pmatrix}
\tilde{\bs}_{i,1}\\
\tilde{\bs}_{i,2}
\end{pmatrix}\triangleq \begin{pmatrix}
\bI_{N_{1}}& \bX_1\\
\mathrm{0}& \bI_{N_2}
\end{pmatrix}\begin{pmatrix}
\bs_{i,1}\\
\bs_{i,2}
\end{pmatrix}
\end{align} and 
\begin{align}
\bD_1 \triangleq \begin{pmatrix}
\bI_{N_0}& \mathrm{0}&\mathrm{0}\\
\mathrm{0}& \bI_{N_1} & \bX_1\\
\mathrm{0}& \mathrm{0}& \bI_{N_2}
\end{pmatrix}
\end{align}  and note that
\begin{align}
\bD^{-1}_1 = \begin{pmatrix}
\bI& \mathrm{0}&\mathrm{0}\\
\mathrm{0}& \bI & -\bX_1\\
\mathrm{0}& \mathrm{0}& \bI
\end{pmatrix}
\end{align} By these definitions it is not hard to check that 
\begin{align}
&\begin{pmatrix}
\tilde{\bs}_{i,1}\\
\tilde{\bs}_{i,2}
\end{pmatrix}\notag\\
&=\begin{pmatrix}
\bI& \bX_1\\
\mathrm{0}& \bI
\end{pmatrix}\begin{pmatrix}
\bR_{1,0}& \bR_{1,1} & \bR_{1,2}\\
\mathrm{0}& \bR_{2,1} & \bR_{2,2}
\end{pmatrix}\bD^{-1}_{1}\begin{pmatrix}
\bs_{i-1,0}\\
\tilde{\bs}_{i-1,1}\\
\tilde{\bs}_{i-1,2}
\end{pmatrix}\notag\\
&=\begin{pmatrix}
\bR_{1,0}& \tilde{\bR}_{1,1} & \tilde{\bR}_{1,2}\\
\mathrm{0}& \bR_{2,1} & \bR_{2,2}-\bR_{2,1}\bX_1
\end{pmatrix}\begin{pmatrix}
\bs_{i-1,0}\\
\tilde{\bs}_{i-1,1}\\
\tilde{\bs}_{i-1,2}
\end{pmatrix} \label{examp-1}
\end{align} where 
\begin{align}
\tilde{\bR}_{1,1}= {\bR}_{1,1}+\bX_1\bR_{2,1}
\end{align}
\begin{align}
\tilde{\bR}_{1,2}= {\bR}_{1,2}+\bX_1\bR_{2,2} -\bX_1\bR_{2,1}\bX_1-\bR_{1,1}\bX_1 
\end{align} $\bR_{2,1}$ is full-row-rank of rank $N_2$ and $\bR_{2,2}$ is $N_2 \times N_2$ matrix, thus $\bX_1$ can be selected such that
\begin{align}
\bR_{2,2}-\bR_{2,1}\bX_1=\mathrm{0}
\end{align} and \eqref{examp-1} reduces to
\begin{align}
&\begin{pmatrix}
\tilde{\bs}_{i,1}\\
\tilde{\bs}_{i,2}
\end{pmatrix}=\begin{pmatrix}
\bR_{1,0}& \tilde{\bR}_{1,1} & \tilde{\bR}_{1,2}\\
\mathrm{0}& \bR_{2,1} & \mathrm{0}
\end{pmatrix}\begin{pmatrix}
\bs_{i-1,0}\\
\tilde{\bs}_{i-1,1}\\
\tilde{\bs}_{i-1,2}
\end{pmatrix} 
\end{align}
{\bf Step 2: } Define
\begin{align}
\tilde{\bs}_{i-1,0} \triangleq \begin{pmatrix}
\bI& \bX_{1,2}&\bX_{2,2}
\end{pmatrix}\begin{pmatrix}
\bs_{i-1,0}\\
\tilde{\bs}_{i-1,1}\\
\tilde{\bs}_{i-1,2}
\end{pmatrix}
\end{align} and 
\begin{align}
\bD_2 \triangleq \begin{pmatrix}
\bI& \bX_{1,2}&\bX_{2,2}\\
\mathrm{0}& \bI & \mathrm{0}\\
\mathrm{0}& \mathrm{0}& \bI
\end{pmatrix}
\end{align} and note that
\begin{align}
\bD^{-1}_2 = \begin{pmatrix}
\bI& -\bX_{1,2}&-\bX_{2,2}\\
\mathrm{0}& \bI & \mathrm{0}\\
\mathrm{0}& \mathrm{0}& \bI
\end{pmatrix}
\end{align}  It can be observed that
\begin{align}
&\begin{pmatrix}
\tilde{\bs}_{i,1}\\
\tilde{\bs}_{i,2}
\end{pmatrix}=\begin{pmatrix}
\bR_{1,0}& \tilde{\bR}_{1,1} & \tilde{\bR}_{1,2}\\
\mathrm{0}& \bR_{2,1} & \mathrm{0}
\end{pmatrix}\bD^{-1}_2\begin{pmatrix}
\tilde{\bs}_{i-1,0}\\
\tilde{\bs}_{i-1,1}\\
\tilde{\bs}_{i-1,2}
\end{pmatrix}\\
&= \begin{pmatrix}
\bR_{1,0}& \tilde{\bR}_{1,1}-\bR_{1,0}\bX_{1,2} & \tilde{\bR}_{1,2}-\bR_{1,0}\bX_{2,2} \\
\mathrm{0}& \bR_{2,1} & \mathrm{0}
\end{pmatrix}\notag\\
&\hspace{2cm}\times \begin{pmatrix}
\tilde{\bs}_{i-1,0}\\
\tilde{\bs}_{i-1,1}\\
\tilde{\bs}_{i-1,2}
\end{pmatrix}
\end{align} Similarly, $\bX_{1,2}$ and $\bX_{2,2}$ are selected such that 
\begin{align}
\tilde{\bR}_{1,1}-\bR_{1,0}\bX_{1,2} &= \mathrm{0}\\
\tilde{\bR}_{1,2}-\bR_{1,0}\bX_{2,2} &= \mathrm{0}
\end{align} Therefore, the source $\tilde{\bs}_i$ consists of $N_0$ innovation bits and $N_1+N_2$ deterministic bits as
\begin{align}
\begin{pmatrix}
\tilde{\bs}_{i,1}\\
\tilde{\bs}_{i,2}
\end{pmatrix}&=\begin{pmatrix}
\bR_{1,0}& \mathrm{0} & \mathrm{0} \\
\mathrm{0}& \bR_{2,1} & \mathrm{0}
\end{pmatrix}\begin{pmatrix}
\tilde{\bs}_{i-1,0}\\
\tilde{\bs}_{i-1,1}\\
\tilde{\bs}_{i-1,2}
\end{pmatrix}\\
&=\begin{pmatrix}
\bR_{1,0}& \mathrm{0}\\
\mathrm{0}& \bR_{2,1} 
\end{pmatrix}\begin{pmatrix}
\tilde{\bs}_{i-1,0}\\
\tilde{\bs}_{i-1,1}
\end{pmatrix}.
\end{align} Clearly, $\tilde{\bs}_{i}=\mathcal{L}_{b}(\hat{\bs}_i)$ is a diagonally correlated deterministic Markov source and the mapping is invertible.

Exploiting Lemmas \ref{lem:full-rank} and \ref{lem:bf}, any linear semi-deterministic source $\bs_{i}$ is first transformed into a diagonally correlated deterministic Markov source $\tilde{\bs}_{i}=\mathcal{L}_{b}(\mathcal{L}_f(\bs_{i}))$ and then is transmitted through the channel using prospicient coding scheme. The block diagram of encoder and decoder is shown in Fig~\ref{fig:LfLb}. The optimality of the scheme can be shown using a similar argument in Sec.~\ref{sec:rate-optimal}.  

\section{Gaussian Sources: Proof of Theorem~\ref{thm:gauss-rate}}
\label{sec:Gauss}
In this section we investigate the Gaussian source model with window recovery constraints explained in \ref{sub:Gaussian}. First we argue that it is sufficient to consider the case $K=B+W$ and then the coding scheme and rate-optimality of the scheme is studied. Finally the rate-recovery function of different schemes are provided for comparison.  

\subsection{Sufficiency of $K=B+W$}
First, we argue that it is sufficient to consider the case $K=B+W$. In particular, if $K<B+W$, we can assume that the decoder, instead of recovering the source $\rvbt_{i}=(\rvs_{i}, \rvs_{i-1}, \ldots, \rvs_{i-K})^{T}$ at time $i$ within distortion $\bd$,  aims to recover the source $\rvbt'_{i}=(\rvs_{i}, ..., \rvs_{i-K'})^{T}$ within distortion $\bd'$ where $K'=B+W$ and 
\begin{align}
d'_{j}=\begin{cases}
d_{j} &\mbox{for } j\in\{1, 2, ..., K\}\\
1&\mbox{for } j\in\{K+1, ..., K'\}
\end{cases}
\end{align}
In the case $K>B+W$, at each time $i$, $\rvs_{i-j}$ is required to be recovered within distortion $d_{j}$ for $j\in\{B+W+1, \ldots, K\}$, however there is always a better reconstruction available from the past. In particular, according to the problem description the decoder at each time $i$ has recovered $\hat{\rvbt}_{i-1}$ or $\hat{\rvbt}_{i-B-W-1}$.
In the former case, $\{\hat{\rvs}_{i-j}\}_{d_{j-1}}$ is available from time $i-1$ and $d_{j-1}\le d_{j}$ and in the latter case, $\{\hat{\rvs}_{i-j}\}_{d_{j-W-B-1}}$ is available from time $i-B-W-1$ and $d_{j-W-B-1}\le d_{j}$. Thus one can simply solve the problem for the case $K=B+W$. 

\subsection{Coding Scheme}
In this section, we propose a coding scheme based on \emph{successive refinement} of the Gaussian source. The block diagram of the scheme is shown in  Fig.~\ref{fig:Ach}. 
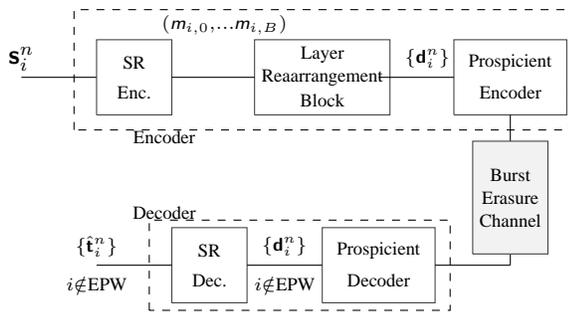
\begin{figure}
\vspace{1em}
\begin{center}
\begin{tikzpicture}
\def \a {-1}
\def \b {0}{
\draw [-] (-1.5+\a,\b) -- (-.5,\b);
\draw [dashed] (-1.8,\b-.7) -- (-1.8,\b+.9) -- (4.8,\b+.9) -- (4.8,\b-.7) -- (-1.8,\b-.7) ; 
\draw [white] (-.6,\b-.8) -- (-.6,\b-.8) node {$\color{black}\scriptstyle\textrm{Encoder}$};
\fill [fill=white, draw=gray!50!black] (\a-.5,\b-.5) -- (\a-.5,\b+.5) -- (\a+.5,\b+.5) -- (\a+.5,\b-.5) -- (\a-.5,\b-.5);
\draw [white] (\a-1.5,\b+0.25) -- (\a-1.5,\b+0.25) node {$\color{black} \rvbs^n_{i}$};
\draw [white] (\a,\b+.2) -- (\a,\b+.2) node {$\color{black}\scriptstyle\textrm{SR}$};
\draw [white] (\a,\b-.2) -- (\a,\b-.2) node {$\color{black}\scriptstyle\textrm{Enc.}$};
}

\def \a {0}
\def \b {-2.5}{
\draw [dashed] (-0.8,\b-.6) -- (-.8,\b+.6) -- (3.2,\b+.6) -- (3.2,\b-.6) -- (-0.8,\b-.6) ; 
\draw [white] (-.6,\b+.7) -- (-.6,\b+.7) node {$\color{black}\scriptstyle\textrm{Decoder}$};
\fill [fill=white, draw=gray!50!black] (\a-.5,\b-.5) -- (\a-.5,\b+.5) -- (\a+.5,\b+.5) -- (\a+.5,\b-.5) -- (\a-.5,\b-.5);
\draw [white] (\a-1.5,\b+0.25) -- (\a-1.5,\b+0.25) node {$\color{black} \scriptstyle\{\hat{\rvbt}^n_{i}\}$};
\draw [white] (\a-1.5,\b-0.25) -- (\a-1.5,\b-0.25) node {$\color{black} \scriptstyle i\notin \textrm{EPW}$};
\draw [white] (\a,\b+.2) -- (\a,\b+.2) node {$\color{black}\scriptstyle\textrm{SR}$};
\draw [white] (\a,\b-.2) -- (\a,\b-.2) node {$\color{black}\scriptstyle\textrm{Dec.}$};
\draw [-] (-1.5+\a,\b) -- (-.5,\b);
}

\def \a {1.5}
\def \b {0}{
\fill [fill=white, draw=gray!50!black] (\a-.9,\b-.5) -- (\a-.9,\b+.5) -- (\a+.9,\b+.5) -- (\a+.9,\b-.5) -- (\a-.9,\b-.5);
\draw [white] (\a-1.3,\b+.7) -- (\a-1.3,\b+.7) node {$\color{black}\scriptstyle (\rvm_{i,0},\ldots \rvm_{i,B})$};
\draw [white] (\a,\b+.3) -- (\a,\b+.3) node {$\color{black}\scriptstyle\textrm{Layer}$};
\draw [white] (\a,\b) -- (\a,\b) node {$\color{black}\scriptstyle\textrm{Reaarrangement}$};
\draw [white] (\a,\b-.3) -- (\a,\b-.3) node {$\color{black}\scriptstyle\textrm{Block}$};
\draw [-] (\a-2,\b) -- (\a-.9,\b);
}

\def \a {4}
\def \b {0}{
\fill [fill=white, draw=gray!50!black] (\a-.75,\b-.5) -- (\a-.75,\b+.5) -- (\a+.75,\b+.5) -- (\a+.75,\b-.5) -- (\a-.75,\b-.5);
\draw [white] (\a-1.1,\b+0.25) -- (\a-1.1,\b+0.25) node {$\color{black} \scriptstyle \{\rvbd^n_{i}\}$};
\draw [white] (\a,\b+.2) -- (\a,\b+.2) node {$\color{black}\scriptstyle\textrm{Prospicient}$};
\draw [white] (\a,\b-.2) -- (\a,\b-.2) node {$\color{black}\scriptstyle\textrm{Encoder}$};
\draw [-] (\a-1.75,\b) -- (\a-.75,\b);
\draw [-] (\a,\b-.5) -- (\a,\b-2);
}

\def \a {2.25}
\def \b {-2.5}{
\fill [fill=white, draw=gray!50!black] (\a-.75,\b-.5) -- (\a-.75,\b+.5) -- (\a+.75,\b+.5) -- (\a+.75,\b-.5) -- (\a-.75,\b-.5);
\draw [white] (\a-1.25,\b+0.25) -- (\a-1.25,\b+0.25) node {$\color{black} \scriptstyle\{\rvbd_{i}^n\}$};
\draw [white] (\a-1.25,\b-0.25) -- (\a-1.25,\b-0.25) node {$\color{black} \scriptstyle i\notin \textrm{EPW}$};
\draw [white] (\a,\b+.2) -- (\a,\b+.2) node {$\color{black}\scriptstyle\textrm{Prospicient}$};
\draw [white] (\a,\b-.2) -- (\a,\b-.2) node {$\color{black}\scriptstyle\textrm{Decoder}$};
\draw [-] (\a-1.75,\b) -- (\a-.75,\b);
\draw [-] (\a+.75,\b) -- (\a+1.75,\b);
\draw [-] (\a+1.75,\b) -- (\a+1.75,\b+.5);
}

\def \a {4}
\def \b {-1.6}{
\fill [fill=gray!10!white, draw=gray!50!black] (\a-.5,\b-.75) -- (\a-.5,\b+.75) -- (\a+.5,\b+.75) -- (\a+.5,\b-.75) -- (\a-.5,\b-.75);
\draw [white] (\a,\b+.3) -- (\a,\b+.3) node {$\color{black}\scriptstyle\textrm{Burst}$};
\draw [white] (\a,\b) -- (\a,\b) node {$\color{black}\scriptstyle\textrm{Erasure}$};
\draw [white] (\a,\b-.3) -- (\a,\b-.3) node {$\color{black}\scriptstyle\textrm{Channel}$};
}
\end{tikzpicture}
\caption{Schematic of Encoder and Decoder for Gaussian Case. (EPW$\equiv$ Error Propagation Window)}
\label{fig:Ach}
\end{center}
\end{figure}
\subsubsection{Successive Refinement (SR) Encoder}
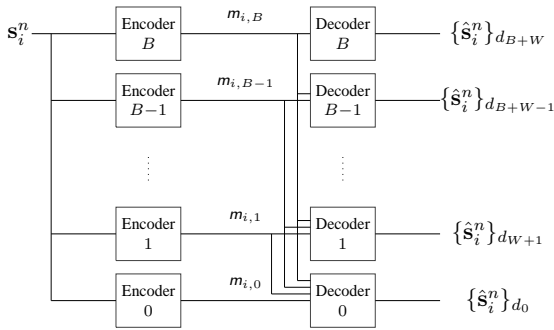
\begin{figure}
\vspace{1em}
\begin{center}
\noindent
\resizebox{3in}{1.7in}{
\begin{tikzpicture}
\draw [white] (-1.5,-.4) -- (-1.5,-.4) node {$\color{black}\bs_{i}^{n}$};
\def \a {0}{
\fill [fill=white, draw=gray!50!black] (0,-.8-\a) -- (0,-\a) -- (1,-\a) -- (1,-.8-\a) -- (0,-.8-\a);
\draw [white] (.5,-\a-0.25) -- (.5,-\a-0.25) node {$\color{black}\textrm{\scriptsize{Encoder}}$};
\draw [white] (.5,-\a-0.55) -- (.5,-\a-0.55) node {$\color{black}\scriptstyle{B}$};
\draw [-] (-1.3,-.4) -- (0,-.4);
\draw [-] (1,-.4) -- (3,-.4);
\fill [fill=white, draw=gray!50!black] (3,-.8-\a) -- (3,-\a) -- (4,-\a) -- (4,-.8-\a) -- (3,-.8-\a);
\draw [white] (3.5,-\a-0.25) -- (3.5,-\a-0.25) node {$\color{black}\textrm{\scriptsize{Decoder}}$};
\draw [white] (3.5,-\a-0.55) -- (3.5,-\a-0.55) node {$\color{black}\scriptstyle{B}$};
\draw [-] (4,-.4) -- (5,-.4);
\draw [white] (5.9,-\a-0.4) -- (5.9,-\a-0.4) node {$\color{black}\{\hat{\bs}_{i}^n\}_{d_{B+W}}$};
\draw [white] (2,-\a-0.15) -- (2,-\a-0.15) node {$\color{black}\scriptstyle\rvm_{i,B}$};
}
\def \a {1}{
\fill [fill=white, draw=gray!50!black] (0,-.8-\a) -- (0,-\a) -- (1,-\a) -- (1,-.8-\a) -- (0,-.8-\a);
\draw [white] (.5,-\a-0.25) -- (.5,-\a-0.25) node {$\color{black}\textrm{\scriptsize{Encoder}}$};
\draw [white] (.5,-\a-0.55) -- (.5,-\a-0.55) node {$\color{black}\scriptstyle{B-1}$};
\draw [-] (-1,-.4-\a) -- (0,-.4-\a);
\draw [-] (1,-.4-\a) -- (3,-.4-\a);
\fill [fill=white, draw=gray!50!black] (3,-.8-\a) -- (3,-\a) -- (4,-\a) -- (4,-.8-\a) -- (3,-.8-\a);
\draw [white] (3.5,-\a-0.25) -- (3.5,-\a-0.25) node {$\color{black}\textrm{\scriptsize{Decoder}}$};
\draw [white] (3.5,-\a-0.55) -- (3.5,-\a-0.55) node {$\color{black}\scriptstyle{B-1}$};
\draw [-] (4,-.4-\a) -- (5,-.4-\a);
\draw [white] (5.9,-\a-0.4) -- (5.9,-\a-0.4) node {$\color{black}\{\hat{\bs}_{i}^n\}_{d_{B+W-1}}$};
\draw [white] (2,-\a-0.15) -- (2,-\a-0.15) node {$\color{black}\scriptstyle\rvm_{i,B-1}$};
}
\draw [dotted] (.5,-2-.6)--(.5,-2-.2);
\draw [dotted] (3.5,-2-.6)--(3.5,-2-.2);
\def \a {3}{
\fill [fill=white, draw=gray!50!black] (0,-.8-\a) -- (0,-\a) -- (1,-\a) -- (1,-.8-\a) -- (0,-.8-\a);
\draw [white] (.5,-\a-0.25) -- (.5,-\a-0.25) node {$\color{black}\textrm{\scriptsize{Encoder}}$};
\draw [white] (.5,-\a-0.55) -- (.5,-\a-0.55) node {$\color{black}\scriptstyle{1}$};
\draw [-] (-1,-.4-\a) -- (0,-.4-\a);
\draw [-] (1,-.4-\a) -- (3,-.4-\a);
\fill [fill=white, draw=gray!50!black] (3,-.8-\a) -- (3,-\a) -- (4,-\a) -- (4,-.8-\a) -- (3,-.8-\a);
\draw [white] (3.5,-\a-0.25) -- (3.5,-\a-0.25) node {$\color{black}\textrm{\scriptsize{Decoder}}$};
\draw [white] (3.5,-\a-0.55) -- (3.5,-\a-0.55) node {$\color{black}\scriptstyle{1}$};
\draw [-] (4,-.4-\a) -- (5,-.4-\a);
\draw [white] (5.9,-\a-0.4) -- (5.9,-\a-0.4) node {$\color{black}\{\hat{\bs}_{i}^n\}_{d_{W+1}}$};
\draw [white] (2,-\a-0.15) -- (2,-\a-0.15) node {$\color{black}\scriptstyle\rvm_{i,1}$};
}
\def \a {4}{
\fill [fill=white, draw=gray!50!black] (0,-.8-\a) -- (0,-\a) -- (1,-\a) -- (1,-.8-\a) -- (0,-.8-\a);
\draw [white] (.5,-\a-0.25) -- (.5,-\a-0.25) node {$\color{black}\textrm{\scriptsize{Encoder}}$};
\draw [white] (.5,-\a-0.55) -- (.5,-\a-0.55) node {$\color{black}\scriptstyle{0}$};
\draw [-] (-1,-.4-\a) -- (0,-.4-\a);
\draw [-] (1,-.4-\a) -- (3,-.4-\a);
\fill [fill=white, draw=gray!50!black] (3,-.8-\a) -- (3,-\a) -- (4,-\a) -- (4,-.8-\a) -- (3,-.8-\a);
\draw [white] (3.5,-\a-0.25) -- (3.5,-\a-0.25) node {$\color{black}\textrm{\scriptsize{Decoder}}$};
\draw [white] (3.5,-\a-0.55) -- (3.5,-\a-0.55) node {$\color{black}\scriptstyle{0}$};
\draw [-] (4,-.4-\a) -- (5,-.4-\a);
\draw [white] (5.9,-\a-0.4) -- (5.9,-\a-0.4) node {$\color{black}\{\hat{\bs}_{i}^n\}_{d_{0}}$};
\draw [white] (2,-\a-0.15) -- (2,-\a-0.15) node {$\color{black}\scriptstyle\rvm_{i,0}$};
}

\draw [-] (-1,-.4) -- (-1,-4.4);
\draw [-] (2.8,-.4) -- (2.8,-4.1);
\draw [-] (2.6,-1.4) -- (2.6,-4.2);
\draw [-] (2.4,-3.4) -- (2.4,-4.3);

\draw [-] (2.8,-1.3) -- (3,-1.3);
\draw [-] (2.8,-3.2) -- (3,-3.2);
\draw [-] (2.6,-3.3) -- (3,-3.3);
\draw [-] (2.8,-4.1) -- (3,-4.1);
\draw [-] (2.6,-4.2) -- (3,-4.2);
\draw [-] (2.4,-4.3) -- (3,-4.3);

\end{tikzpicture}}
\caption{($B+1$)-layer coding scheme based on Successive Refinement.}
\label{Suc-Ref}
\end{center}
\end{figure}

The structure of SR encoder is shown in Fig.~\ref{Suc-Ref}. 
The encoder at time $i$, encodes source signal $\bs^n_{i}$ using a $(B+1)$-layer successive refinement coding scheme~\cite{equitzCover:91, rimoldi:94} to generate $(B+1)$ codewords whose indices are given by $\{\rvm_{i,0}, \rvm_{i,1}, \ldots, \rvm_{i,B-1}, \rvm_{i,B}\}$ where $\rvm_{i,j}\in \{1,2,\ldots,2^{n\tilde{R}_j}\}$ for $j\in\{0,1,\dots,B\}$ and
\begin{align}
&\tilde{R}_{j}=\begin{cases}
\frac{1}{2}\log(\frac{d_{W+1}}{d_{0}})& \mbox{for } j=0\\
\frac{1}{2}\log(\frac{d_{W+j+1}}{d_{W+j}})& \mbox{for } j\in\{1,2,\ldots,B\}\\
\frac{1}{2}\log(\frac{1}{d_{W+j}})& \mbox{for }  j=B,
\end{cases}
\end{align}

The $j$-th layer uses indices \begin{align}M_{i,j} \defeq (\rvm_{i,j},\ldots,\rvm_{i,B}) \label{eq:Mdef}\end{align} for reproduction and the associated rate is given by:
\begin{align}
&R_{j}=\notag\\
&\begin{cases}
\sum_{k=0}^{B}\tilde{R}_{k}= \frac{1}{2}\log(\frac{1}{d_{0}})& \mbox{for } j=0\\
\sum_{k=j}^{B}\tilde{R}_{k}= \frac{1}{2}\log(\frac{1}{d_{W+j}})& \mbox{for } j\in\{1,2,\ldots,B-1\},
\end{cases}
\label{eq:Rdef}\end{align}

and the corresponding distortion associated with layer $j$ equals $d_0$ if $j=0$ and $d_{W+j}$ for $j=1,\ldots, B$.

Clearly, for recovering $\bt_i^n = (\rvs_i^n, \ldots, \rvs_{i-B-W}^n)$ with a distortion tuple $(d_0,\ldots, d_{B+W}),$  it suffices that the destination have access to:
\begin{equation}
\cM_i = \begin{bmatrix} M_{i,0} \\ M_{i-1,0} \\  \vdots \\  M_{i-W,0} \\  M_{i-W-1,1} \\  \vdots \\  M_{i-W-B,B} \end{bmatrix}
\label{eq:cMdef}
\end{equation}

As described next, our coding scheme  uses the prospicient code construction in Section~\ref{sec:deterministic} for the diagonally correlated source model to guarantee that the receiver obtains $\cM_i$ for each $i$ outside the error propagation window.

\subsubsection{Layer Rearrangement Block}

For simplicity first assume that all $\tilde{R}_j$s are integer.  The results can be generalized for non-integer rates as explained in Appendix~\ref{app:gen}.  Each index $\rvm_{i,j}$ is isomorphic to a length $n$ sequence $\bb_{i,j}^n$  over the alphabet $\cB_j \in \{0,1\}^{R_j}$ and the indices associated with layer $j$ are isomorphic to a sequence $\bc_{i,j}^n$ defined as
\begin{equation}
M_{i,j}\leftrightarrow \bc_{i,j}^n = \left(\bb_{i,j}^n, \ldots, \bb_{i,B}^n\right).\label{eq:def_cij}
\end{equation}
and the collection of layers $\cM_{i}$ as defined in~\eqref{eq:cMdef} is isomorphic to
\begin{equation}
\cM_i \leftrightarrow \bd_i^n = \begin{bmatrix}
\bc_{i,0}^n\\
\bc_{i-1,0}^n\\
\vdots\\ 
\bc_{i-W,0}^n\\
\bc_{i-W-1,1}^n\\
\vdots\\
\bc_{i-W-B,B}^n
\end{bmatrix}\label{eq:d_def}
\end{equation}

As shown in Fig.~\ref{fig:Ach}, the sequence $\bd_i^n$ is encoded at time $i$ and recovered outside the error propagation window.

\subsubsection{Prospicient Encoder/Decoder} 
It can be readily verified that $\bd_i^n$ in \eqref{eq:d_def} is a linear diagonally correlated semi-deterministic source as defined in Def.~\ref{def:Diagonal}. Hence
applying the prospecient coding scheme in section~\ref{sec:deterministic}, the achievable rate from Prop.~\ref{prop:LDD} is
\begin{align}
R&=R_{0}+\frac{1}{W+1}\sum_{k=1}^{B}R_{k}\\
&=\frac{1}{2}\log\frac{1}{d_{0}}+\frac{1}{2(W+1)}\sum_{k=1}^{B}\log\frac{1}{d_{W+j}}. \label{eq:R_gauss_layered_achiev}
\end{align} 

\subsubsection{Decoding of $\hat{\bt}_i^n$} Using $\bd_i^n$, which is isomorphic to $\cM_i$ defined in~\eqref{eq:cMdef}, the decoder is guaranteed to recover $(\rvs_i^n, \ldots, \rvs_{i-W}^n)$ with distortions $d_0$ and the sequences $\rvs_{i-W-1}^n, \ldots, \rvs_{i-W-B}^n$ with distortions $d_{W+1},\ldots, d_{W+B}$ respectively.

\subsection{Converse for Theorem~\ref{thm:gauss-rate}}
We need to show that for any sequence of codes that achieve a distortion tuple $(d_0,\ldots, d_{W+B})$ the rate is lower bounded by~\eqref{eq:R_gauss_layered_achiev}.

As in the proof of Theorem~\ref{thm:genUB_LB}, we consider a periodic erasure channel of period $p=B+W+1$ and assume that the first $B$ positions of each period are erased.
Consider,

\begin{align}
&(W+1)n(t+1)R \notag\\
&= H\left(\rvf_{B}^{p-1},\rvf_{p+B}^{2p-1},\ldots, \rvf_{(t-1)p+B}^{tp-1}, \rvf_{tp+B}^{(t+1)p-1}\right)\\
&= H(\rvf_{B}^{p-1}) + \sum_{k=1}^t H(\rvf_{kp+B}^{(k+1)p-1} | \rvf_{B}^{p-1},\rvf_{p+B}^{2p-1},\ldots, \rvf_{(k-1)p+B}^{kp-1})\notag\\
&\ge H(\rvf_{B}^{p-1}) +  \sum_{k=1}^t H\left(\rvf_{kp+B}^{(k+1)p-1} | \rvf_{0}^{kp-1}\right)\label{eq:Gauss-LB}
\end{align}
where the last step follows from the fact that conditioning reduces entropy. 

We next establish the following claim, whose proof is in Appendix~\ref{app:Gauss-Ent}.
\begin{claim}
For each $k \ge 1$ we have that
\begin{align}
&H\left(\rvf_{kp+B}^{(k+1)p-1} | \rvf_{0}^{kp-1}\right) \notag\\&\ge \frac{n}{2}\sum_{i=1}^{B} \log (\frac{1}{ d_{W+i}}) + \frac{n(W+1)}{2}\log(\frac{1}{d_{0}}) \label{eq:gauss-Ent-LB}.
\end{align}
\label{claim:Gauss-Ent}
\end{claim}

Substituting~\eqref{eq:gauss-Ent-LB} into~\eqref{eq:Gauss-LB} 
and taking $n\rightarrow \infty$ and then $t\rightarrow \infty$, we recover 
\begin{align}
R \ge \frac{1}{2}\log_{2}(\frac{1}{d_{0}}) + \frac{1}{2(W+1)}\sum_{j=1}^{B} \log_2 (\frac{1}{ d_{W+j}}).
\end{align} as required.

\subsection{Illustrative Sub-optimal Schemes}
\label{subsec:comparison}
As explained in Sec.\ref{sub:Gaussian} and Fig.~\ref{fig:comparison}, the optimal performance is compared with the following sub-optimal schemes. 
\subsubsection{Still-Image Compression} In this scheme, the encoder ignores the decoder's memory and at time $i\ge 0$ encodes the source $\rvbt_{i}$  in a memoryless manner and sends the codewords through the channel. The rate associated to this scheme is 
\begin{align}
R_{\textrm{SI}}(\bd)&=I(\rvbt_{i};\hat{\rvbt}_{i})=\sum_{k=0}^{K}\frac{1}{2}\log \bigg(\frac{1}{d_{k}}\bigg)
\end{align} In this scheme, the decoder is able to recover the source whenever its codeword is available, i.e. at all the times except when the erasure happens.   
\subsubsection{Wyner-Ziv Compression with Delayed Side Information} At time $i$ the encoders assumes that $\rvbt_{i-B-1}$ is already reconstructed at the receiver within distortion $\bd$. With this assumption, it compresses the source $\rvbt_{i}$ according to Wyner-Ziv scheme and transmits the codewords through the channel. The rate of this scheme is 
\begin{align}
R_{\textrm{WZ}}(B,\bd)&=I(\rvbt_{i};\hat{\rvbt}_{i}|\hat{\rvbt}_{i-B-1})=\sum_{k=0}^{B}\frac{1}{2}\log \bigg(\frac{1}{d_{k}}\bigg)
\end{align} Note that, if at time $i$, $\hat{\rvbt}_{i-B-1}$ is not available, $\hat{\rvbt}_{i-1}$ is available and the decoder can still use it as side-information to construct $\hat{\rvbt}_{i}$ since $I(\rvbt_{i};\hat{\rvbt}_{i}|\hat{\rvbt}_{i-B-1})\ge I(\rvbt_{i};\hat{\rvbt}_{i}|\hat{\rvbt}_{i-1})$. 

As in the case of Still-Image Compression, the Wyner-Ziv scheme also enables the recovery of each source sequence except those with erased codewords.

\subsubsection{Predictive Coding plus FEC} This scheme consists of predictive coding (DPC)~\cite{berger71} followed by a Forward Error Correction (FEC) code to compensate the effect of packet losses of the channel. As the contribution of $B$ erased codewords need to be recovered using $W+1$ available codewords, the rate of this scheme can be computed as follows. 
\begin{align}
R_{\textrm{FEC}}(B,W,\bd)&=\frac{B+W+1}{W+1}I(\rvbt_{i};\hat{\rvbt}_{i}|\hat{\rvbt}_{i-1})\\&=\frac{B+W+1}{2(W+1)}\log \bigg(\frac{1}{d_{0}}\bigg)
\end{align}

\section{Symmetric Sources: Proof of Theorem~\ref{thm:binning}}
\label{sec:Symmetric}
\begin{figure*}
\begin{center}
\begin{tikzpicture}

\def \a {0}{
\draw [white] (-1.5,-\a-.4) -- (-1.5,-\a-.4) node {$\color{black}\bs_{j}^{n}$};
\fill [fill=white, draw=gray!50!black] (0,-.8-\a) -- (0,-\a) -- (1,-\a) -- (1,-.8-\a) -- (0,-.8-\a);
\draw [white] (.5,-\a-0.25) -- (.5,-\a-0.25) node {$\color{black}\scriptstyle\textrm{Encoder}$};
\draw [white] (.5,-\a-0.55) -- (.5,-\a-0.55) node {$\color{black}\scriptstyle{j}$};
\draw [->] (-1.3,-.4) -- (0,-.4);
\draw [->] (1,-.4) -- (3,-.4);
\fill [fill=white, draw=gray!50!black] (3,-.8-\a) -- (3,-\a) -- (4,-\a) -- (4,-.8-\a) -- (3,-.8-\a);
\draw [white] (3.5,-\a-0.25) -- (3.5,-\a-0.25) node {$\color{black}\scriptstyle\textrm{Decoder}$};
\draw [white] (3.5,-\a-0.55) -- (3.5,-\a-0.55) node {$\color{black}\scriptstyle{j}$};
\draw [->] (4,-.4) -- (5,-.4);
\draw [white] (5.5,-\a-0.4) -- (5.5,-\a-0.4) node {$\color{black}\hat{\bs}^n_{j}$};
\draw [white] (2,-\a-0.15) -- (2,-\a-0.15) node {$\color{black}\rvf_{j}$};
\draw [->] (3.5,-\a+1) -- (3.5,-\a);
\draw [white] (4,-\a+.5) -- (4,-\a+.5) node {$\color{black}\bs^n_{j-1}$};
}
\def \a {1.5}{
\draw [white] (-1.5,-\a-.4) -- (-1.5,-\a-.4) node {$\color{black}\bs_{j+1}^{n}$};
\fill [fill=white, draw=gray!50!black] (0,-.8-\a) -- (0,-\a) -- (1,-\a) -- (1,-.8-\a) -- (0,-.8-\a);
\draw [white] (.5,-\a-0.25) -- (.5,-\a-0.25) node {$\color{black}\scriptstyle\textrm{Encoder}$};
\draw [white] (.5,-\a-0.55) -- (.5,-\a-0.55) node {$\color{black}\scriptstyle{j+1}$};
\draw [->] (-1,-.4-\a) -- (0,-.4-\a);
\draw [->] (1,-.4-\a) -- (3,-.4-\a);
\fill [fill=white, draw=gray!50!black] (3,-.8-\a) -- (3,-\a) -- (4,-\a) -- (4,-.8-\a) -- (3,-.8-\a);
\draw [white] (3.5,-\a-0.25) -- (3.5,-\a-0.25) node {$\color{black}\scriptstyle\textrm{Decoder}$};
\draw [white] (3.5,-\a-0.55) -- (3.5,-\a-0.55) node {$\color{black}\scriptstyle{j+1}$};
\draw [->] (4,-.4-\a) -- (5,-.4-\a);
\draw [white] (5.5,-\a-0.4) -- (5.5,-\a-0.4) node {$\color{black}\hat{\bs}^n_{j+1}$};
\draw [white] (2,-\a-0.15) -- (2,-\a-0.15) node {$\color{black}\rvf_{j+1}$};
\draw [->] (3.5,-\a-1.8) -- (3.5,-\a-.8);
\draw [white] (4.2,-\a-1.3) -- (4.2,-\a-1.3) node {$\color{black}\bs^n_{j-B-1}$};
}

\draw [-] (2.5,-.4) -- (2.5,-1.7);
\draw [->] (2.5,-1.7) -- (3,-1.7);
\draw [white] (2,-4) -- (2,-4) node {$\color{black}\textrm{(a)}$};

\def \a {0}{
\draw [white] (6.5,-\a-.4) -- (6.5,-\a-.4) node {$\color{black}\bs_{j}^{n}$};
\fill [fill=white, draw=gray!50!black] (8,-.8-\a) -- (8,-\a) -- (9,-\a) -- (9,-.8-\a) -- (8,-.8-\a);
\draw [white] (8.5,-\a-0.25) -- (8.5,-\a-0.25) node {$\color{black}\scriptstyle\textrm{Encoder}$};
\draw [white] (8.5,-\a-0.55) -- (8.5,-\a-0.55) node {$\color{black}\scriptstyle{j}$};
\draw [->] (6.7,-.4) -- (8,-.4);
\draw [->] (9,-.4) -- (11,-.4);
\fill [fill=white, draw=gray!50!black] (11,-.8-\a) -- (11,-\a) -- (12,-\a) -- (12,-.8-\a) -- (11,-.8-\a);
\draw [white] (11.5,-\a-0.25) -- (11.5,-\a-0.25) node {$\color{black}\scriptstyle\textrm{Decoder}$};
\draw [white] (11.5,-\a-0.55) -- (11.5,-\a-0.55) node {$\color{black}\scriptstyle{j}$};
\draw [->] (12,-.4) -- (13,-.4);
\draw [white] (13.5,-\a-0.4) -- (13.5,-\a-0.4) node {$\color{black}\hat{\bs}^n_{j}$};
\draw [white] (10,-\a-0.15) -- (10,-\a-0.15) node {$\color{black}\rvf_{j}$};
\draw [->] (11.5,-\a+1) -- (11.5,-\a);
\draw [white] (12,-\a+.5) -- (12,-\a+.5) node {$\color{black}\bs^n_{j+1}$};
}
\def \a {1.5}{
\draw [white] (6.5,-\a-.4) -- (6.5,-\a-.4) node {$\color{black}\bs_{j+1}^{n}$};
\fill [fill=white, draw=gray!50!black] (8,-.8-\a) -- (8,-\a) -- (9,-\a) -- (9,-.8-\a) -- (8,-.8-\a);
\draw [white] (8.5,-\a-0.25) -- (8.5,-\a-0.25) node {$\color{black}\scriptstyle\textrm{Encoder}$};
\draw [white] (8.5,-\a-0.55) -- (8.5,-\a-0.55) node {$\color{black}\scriptstyle{j+1}$};
\draw [->] (7,-.4-\a) -- (8,-.4-\a);
\draw [->] (9,-.4-\a) -- (11,-.4-\a);
\fill [fill=white, draw=gray!50!black] (11,-.8-\a) -- (11,-\a) -- (12,-\a) -- (12,-.8-\a) -- (11,-.8-\a);
\draw [white] (11.5,-\a-0.25) -- (11.5,-\a-0.25) node {$\color{black}\scriptstyle\textrm{Decoder}$};
\draw [white] (11.5,-\a-0.55) -- (11.5,-\a-0.55) node {$\color{black}\scriptstyle{j+1}$};
\draw [->] (12,-.4-\a) -- (13,-.4-\a);
\draw [white] (13.5,-\a-0.4) -- (13.5,-\a-0.4) node {$\color{black}\hat{\bs}^n_{j+1}$};
\draw [white] (10,-\a-0.15) -- (10,-\a-0.15) node {$\color{black}\rvf_{j+1}$};
\draw [->] (11.5,-\a-1.8) -- (11.5,-\a-.8);
\draw [white] (12.2,-\a-1.3) -- (12.2,-\a-1.3) node {$\color{black}\bs^n_{j-B-1}$};
}

\draw [-] (10.5,-.4) -- (10.5,-1.7);
\draw [->] (10.5,-1.7) -- (11,-1.7);
\draw [white] (10,-4) -- (10,-4) node {$\color{black}\textrm{(b)}$};

\draw [dotted] (6,1.5) -- (6,-4);

\end{tikzpicture}
\end{center}
\caption{Connection between the streaming problem and Zig-Zag source coding problem. 
The setup on the right is identical to the setup on the left, except with the side information sequence $\rvs_{j-1}^n$ replaced with $\rvs_{j+1}^n$. However the rate
region for both problems turns out to be identical for symmetric Markov sources.}
\label{fig:zigZag}
\end{figure*}
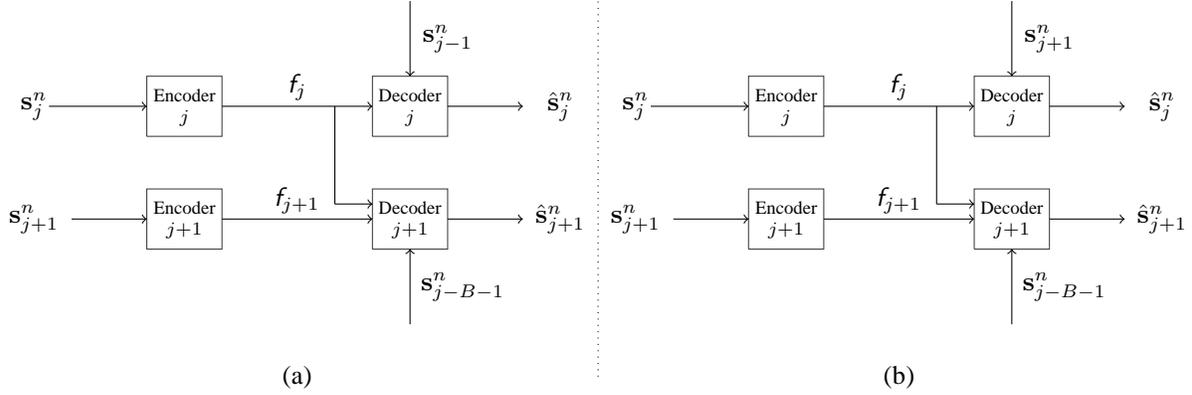

The special case when $W=0$  follows directly from~\eqref{eq:genLB}. We only need to consider the case when $W \ge 1$.
For simplicity in exposition we consider the case when $W=1$. Then we need to show that 
\begin{align}
R(B,W=1) \ge \frac{1}{2} H(\rvs_{B+1}, \rvs_{B+2}|\rvs_0)\label{eq:rate-rec-binning-2}
\end{align}
The proof for general $W \ge 1$ follows along similar lines and will be sketched briefly.

Assume that an erasure-burst spans time indices ${j-B, \ldots, j-1}$. The decoder must recover \begin{equation}\label{eq:decoder1}{\hat\rvs}_{j+1}^n = \cG_{j+1}\left(\rvf_0^{j-B-1}, \rvf_{j}, \rvf_{j+1}\right).\end{equation} 
From Fano's inequality, we have,
\begin{equation}
H\left(\rvs_{j+1}^n ~|~ \rvf_0^{j-B-1}, \rvf_{j}, \rvf_{j+1} \right) \le n\eps_n. 
\label{eq:Fano1}
\end{equation}

Furthermore if there is no erasure until time $j$ then \begin{align}\hat{\rvs}_j^n = \cG_j\left(\rvf_0^{j}\right)\label{eq:decoder0}\end{align} must hold. Hence from Fano's Inequalty,
\begin{equation}
H\left(\rvs_{j}^n ~|~ \rvf_0^{j} \right) \le n\eps_n. 
\label{eq:Fano3}
\end{equation}
Our aim is to combine~\eqref{eq:Fano1} and~\eqref{eq:Fano3} to establish the following lower bound on the sum-rate
\begin{equation}
\label{eq:Rj_lb}
R_j + R_{j+1} \ge H(\rvs_{j+1}|\rvs_j) + H(\rvs_{j}|\rvs_{j-B-1}).
\end{equation}
The lower bound then  follows since \begin{align}
R &\ge \max(R_j, R_{j+1})\\
&\ge \frac{1}{2}(R_j + R_{j+1})\\
&\ge \frac{1}{2}(H(\rvs_{j+1}|\rvs_j) + H(\rvs_{j}|\rvs_{j-B-1}))\\
&=\frac{1}{2}(H(\rvs_{j+1}|\rvs_j,\rvs_{j-B-1}) + H(\rvs_{j}|\rvs_{j-B-1}))\\
&=\frac{1}{2}H(\rvs_{j+1},\rvs_{j}|\rvs_{j-B-1}) =\frac{1}{2}H(\rvs_{B+1},\rvs_{B+2}|\rvs_0)
\end{align}
thus establishing~\eqref{eq:rate-rec-binning-2}.

To establish~\eqref{eq:Rj_lb} we make a connection to a multi-terminal source coding problem in Fig.~\ref{fig:zigZag}.
\subsection{Zig-Zag Source Coding}

Consider the source coding problem with side information illustrated in Fig.~\ref{fig:zigZag}(a). 
In this setup there are four source sequences drawn i.i.d.\ from a joint distribution $p(\rvs_{j+1},\rvs_j, \rvs_{j-1},\rvs_{j-B-1})$.
The two encoders $j$  and $j+1$ are revealed source sequences $\rvs_j^n$ and $\rvs_{j+1}^n$ and the two decoders $j$ and $j+1$ are revealed sources
$\rvs_{j-1}^n$ and $\rvs_{j-B-1}^n$. The encoders operate independently and compress the source sequences to $\rvf_j$ and $\rvf_{j+1}$ at rates $R_j$
and $R_{j+1}$ respectively. Decoder $j$ has access to $(\rvf_j, \rvs_{j-1}^n)$ while decoder ${j+1}$ has access to $(\rvf_j, \rvf_{j+1}, \rvs_{j-B-1}^n)$ and
 are interested in reproducing, 
\begin{align}
\hat{\rvs}_j^n &= \hat{\cG}_j(\rvf_j, \rvs_{j-1}^n) \label{eq:decoder0a}\\
\hat{\rvs}_{j+1}^n&=\hat{\cG}_{j+1}(\rvf_j^{j+1}, \rvs_{j-B-1}^n)\label{eq:decoder1a}\end{align} respectively such that $\Pr(\rvs_i^n \neq \hat{\rvs}_i^n) \le \eps_n$ for ${i= j, j+1}$.
 
When $\rvs_{j-B-1}^n$ is a constant sequence, the problem has been studied in~\cite{berger:77,oohama:gaussian}. A complete single letter characterization involving an auxiliary random variable is obtained. Fortunately in the present case of symmetric sources a simple lower bound can be obtained using the following observation.\begin{lemma}
The set of all achievable rate-pairs $(R_j, R_{j+1})$ for the problem in Fig.~\ref{fig:zigZag}(a) is identical to the set of all achievable rate-pairs for the problem in Fig.~\ref{fig:zigZag}(b)
where the side information sequence $\rvs_{j-1}^n$ at decoder 1 is replaced by the side information sequence $\rvs_{j+1}^n$.
\label{lem:SI}
\end{lemma}

The proof of Lemma~\ref{lem:SI} follows by observing that the capacity region for the problem in Fig.~\ref{fig:zigZag}(a) depends on the joint distribution $p(\rvs_j, \rvs_{j+1}, \rvs_{j-1},\rvs_{j-B-1})$ only via the {\em marginal}
distributions $p(\rvs_j, \rvs_{j-1})$ and $p(\rvs_{j+1}, \rvs_j, \rvs_{j-B-1})$.  When the source is symmetric the distributions $p(\rvs_j, \rvs_{j-1})$ and $p(\rvs_j, \rvs_{j+1})$ are identical.
The formal proof will be omitted.

Thus it suffices to lower bound the achievable sum rate for the problem in Fig.~\ref{fig:zigZag}(b). First upon applying the Slepian-Wolf lower bound to encoder $j+1$
\begin{align}
nR_{j+1}&\ge H(\rvs_{j+1}^n|\rvs_{j-B-1}^n, \rvf_{j}) - n\eps_n\label{eq:sum1}
\end{align}and to bound $R_j$
\begin{align}
&nR_{j}\ge H(\rvf_{j})\ge I(\rvf_{j};\rvs^n_{j}|\rvs_{j-B-1}^n)\nonumber \\
&\ge H(\rvs_{j}^n|\rvs_{j-B-1}^n)\notag - H(\rvs_{j}^n|\rvs_{j-B-1}^n, \rvf_{j})\nonumber \\
&\ge nH(\rvs_{B+1}|\rvs_{0})- H(\rvs_{j}^n|\rvs_{j-B-1}^n, \rvf_{j})\notag\\ &\qquad +H(\rvs_{j}^n|\rvs_{j-B-1}^n,\rvs_{j+1}^n, \rvf_{j})-n\eps_n\label{eq:FanoApp_Sj}\\
&=nH(\rvs_{j}|\rvs_{j-B-1})  - I(\rvs_{j}^n; \rvs_{j+1}^n |\rvs_{j-B-1}^n, \rvf_{j})-n\eps_n\nonumber\\
&=nH(\rvs_{j}|\rvs_{j-B-1})-H(\rvs_{j+1}^n |\rvs_{j-B-1}^n, \rvf_{j}) \notag\\
&\hspace{2cm}+ H( \rvs_{j+1}^n |\rvs_{j-B-1}^n,\rvs_{j}^n, \rvf_{j})-n\eps_n\nonumber\\
&=nH(\rvs_{j}|\rvs_{j-B-1})-H(\rvs_{j+1}^n |\rvs_{j-B-1}^n, \rvf_{j})\notag\\
&\hspace{2cm} + nH(\rvs_{j+1}|\rvs_j) -n\eps_n\label{eq:sum2}
\end{align}
where~\eqref{eq:FanoApp_Sj} follows by applying Fano's inequality since $\rvs_j^n$ can be recovered from $(\rvs_{j+1}^n, \rvf_j)$ and hence
$H(\rvs_{j}^n|\rvs_{j-B-1}^n,\rvs_{j+1}^n, \rvf_{j}) \le n\eps_n$ holds and~\eqref{eq:sum2} follows form the Markov relation ${\rvs_{j+1}^n \rightarrow \rvs_j^n \rightarrow (\rvf_j, \rvs_{j-B-1}^n)}$. 
Observe that~\eqref{eq:Rj_lb} follows by summing~\eqref{eq:sum1} and~\eqref{eq:sum2}.

\subsection{Connection between Streaming and Zig-Zag Coding Problems}
It remains to show that the lower bound on the Zig-Zag coding problem also constitutes a lower bound on the original problem.
\begin{lemma}
Suppose that the encoding function $\rvf_j = \cF_j(\rvs_j^n)$ is memoryless. Suppose that there exist decoding functions $\hat{\rvs}_j^n = {\cG}_j(\rvf_0^j)$
and $\hat{\rvs}_{j+1}^n = {\cG}_{j+1}(f_0^{j-B-1}, \rvf_j, \rvf_{j+1})$ such that  ${\Pr(\hat{\rvs_j}^n \neq \rvs_j^n)}$ and ${\Pr(\hat{\rvs}_{j+1}^n \neq \rvs_{j+1}^n)}$ both vanish to zero as $n\rightarrow\infty$.
Then
\begin{align}
&H(\rvs_{j}^n | \rvs_{j-1}^n, \rvf_j) \le n\eps_n \label{eq:Fano2a}\\
&H(\rvs_{j+1}^n | \rvs_{j-B-1}^n, \rvf_j, \rvf_{j+1}) \le n\eps_n\label{eq:Fano1a}
\end{align} also hold. 
\label{lem:memoryless}
\end{lemma}

\begin{proof}To establish~\eqref{eq:Fano2a} we note that for the memoryless encoder the following Markov chain holds:
\begin{equation}
\label{eq:ZigZag_Markov1}
\rvf_{0}^{j-1}\rightarrow \rvs^n_{j-1} \rightarrow (\rvf_j, \rvs^n_{j}).
\end{equation}
Hence we have that
\begin{align}
n\eps_n \ge H(\rvs_{j}^n|\rvf_{0}^{j})&\ge H(\rvs^n_{j}|\rvf_{0}^{j-1}, \rvs^n_{j-1}, \rvf_{j})\\
&=H(\rvs^n_{j}|\rvs^n_{j-1}, \rvf_{j}) ,
\end{align}where the last step follows via~\eqref{eq:ZigZag_Markov1}.
Similarly using ${\rvf_{0}^{j-B-2}\rightarrow \rvs^n_{j-B-2} \rightarrow (\rvf_{j-1},\rvs^n_{j},\rvf_j)}$, we have 
\begin{align}
n\eps_n &\ge H(\rvs^n_{j}|\rvf_{0}^{j-B-2}, \rvf_{j-1}, \rvf_{j})\notag\\
&\ge H(\rvs^n_{j}|\rvf_{0}^{j-B-2}, \rvs^n_{j-B-2}, \rvf_{j-1}, \rvf_{j})\nonumber \\
&=H(\rvs^n_{j}|\rvs^n_{j-B-2}, \rvf_{j-1}, \rvf_{j}).
\end{align}
\end{proof}
 
The conditions in~\eqref{eq:Fano2a} and~\eqref{eq:Fano1a} show that any rate that is achievable in the original problem is also achieved in the zig-zag source network. Hence a lower bound to the source network also constitutes a lower bound to the original problem.
\subsection{Extension to Arbitrary $W>1$}

Finally we  comment of the extension of the above approach to $W = 2$. We now consider three encoders $t \in \{j, j+1, j+2\}$. Encoder $t$ observes a source sources $\rvs_t^n$ and compresses it into an index $\rvf_j \in [1,2^{nR_j}]$.  The corresponding decoders are revealed $\rvs_{t-1}^n$ for $t \in \{j, j+1\}$ and the decoder $j+2$ is revealed $\rvs_{j-B-1}^n$.  By an argument analogous to Lemma~\ref{lem:SI} the rate region is equivalent to the case when decoders $j$ and $j+1$ are instead revealed $\rvs_{j+1}^n$ and $\rvs_{j+2}^n$ respectively. For this new setup it is easy to show that decoder $j+2$ must reconstruct $(\rvs_j^n, \rvs_{j+1}^n, \rvs_{j+2}^n)$ given $(\rvs_{j-B-1}^n, \rvf_j, \rvf_{j+1},\rvf_{j+2})$. The sum rate must therefore satisfy
$R_j+R_{j+1} + R_{j+2} \ge \frac{1}{3}H(\rvs_j, \rvs_{j+1},\rvs_{j+2}|\rvs_{j-B-1})$. Using an extension of Lemma~\ref{lem:memoryless} we can show that the  proposed lower bound also continues to hold for the original streaming problem. This completes the proof. The extension to any arbitrary $W > 1$ is completely analogous.

\section{Delay Constrained Decoders: Proof of Theorem~\ref{thm:no_W}}
\label{sec:no_W}

\subsection{Achievability}
The achievability of the rate expression~\eqref{eq:R0_Cap1} is established through a Slepian-Wolf coding scheme. A Slepian-Wolf codebook is constructed by partitioning the space of all typical sequences $\rvs_i^n$
into $2^{nR}$ bins and the bin index $\rvf_i$ is transmitted at time $i$. The decoder is required to output $\hat{\rvs}_i^n$ in one of two ways. If it has access to $\rvs_{i-1}^n$ then it finds a sequence jointly typical with $\hat{\rvs}_{i-1}^n$ in the bin index of $\rvf_i$. This succeeds with high probability if $R \ge H(\rvs_1 | \rvs_0)$ which is clearly satisfied in~\eqref{eq:R0_Cap1}.

If the receiver needs to recover from an erasure burst spanning $t \in \{j-B, \ldots, j-1\}$ it has access to $\rvs_{j-B-1}^n$ and needs to use $\rvf_j^{j+T}$ to recover $\rvs_j^n$. It simultaneously attempts to decode
all of $\rvs_j^n, \ldots, \rvs_{j+T}^n$ using $\rvf_j, \ldots, \rvf_{j+T}$ and $\rvs_{j-B-1}^n$. This succeeds if $(T+1)R \ge H(\rvs_j, \ldots, \rvs_{j+T}| \rvs_{j-B-1}) $
 which in turn holds via~\eqref{eq:R0_Cap}.
 
 \subsection{Converse} 

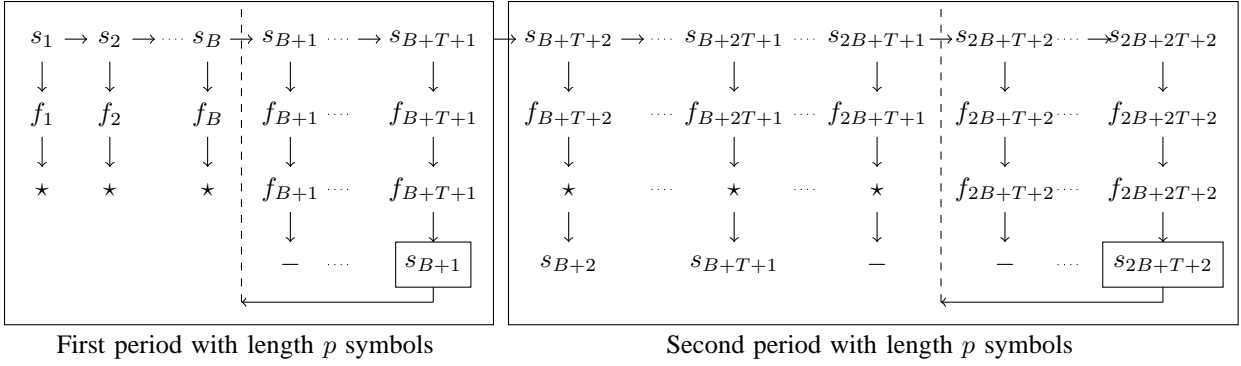
\begin{figure*}[!htb]
\begin{tikzpicture}
\draw [white] (0,0) -- (0,0) node {$\color{black}s_{1}$};
\draw [->](0,-.3) -- (0,-.7);
\draw [white] (0,-1) -- (0,-1) node {$\color{black}f_{1}$};
\draw [->](0,-1.3) -- (0,-1.7);
\draw [white] (0,-2) -- (0,-2) node {$\color{black}\star$};
\draw [->](.3,0) -- (.6,0);

\draw [white] (.9,0) -- (.9,0) node {$\color{black}s_{2}$};
\draw [->](.9,-.3) -- (.9,-.7);
\draw [white] (.9,-1) -- (.9,-1) node {$\color{black}f_{2}$};
\draw [->](.9,-1.3) -- (.9,-1.7);
\draw [white] (.9,-2) -- (.9,-2) node {$\color{black}\star$};
\draw [->](1.2,0) -- (1.5,0);

\draw [dotted](1.6,0) -- (1.9,0);

\draw [white] (2.2,0) -- (2.2,0) node {$\color{black}s_{B}$};
\draw [->](2.2,-.3) -- (2.2,-.7);
\draw [white] (2.2,-1) -- (2.2,-1) node {$\color{black}f_{B}$};
\draw [->](2.2,-1.3) -- (2.2,-1.7);
\draw [white] (2.2,-2) -- (2.2,-2) node {$\color{black}\star$};
\draw [->](2.5,0) -- (2.8,0);
\draw [dashed](2.65,-3.5) -- (2.65,.5);

\draw [white] (3.3,0) -- (3.3,0) node {$\color{black}s_{B+1}$};
\draw [->](3.3,-.3) -- (3.3,-.7);
\draw [white] (3.3,-1) -- (3.3,-1) node {$\color{black}f_{B+1}$};
\draw [->](3.3,-1.3) -- (3.3,-1.7);
\draw [white] (3.3,-2) -- (3.3,-2) node {$\color{black}f_{B+1}$};
\draw [->](3.3,-2.3) -- (3.3,-2.7);
\draw [white] (3.3,-3) -- (3.3,-3) node {$\color{black}-$};

\draw [dotted](3.8,0) -- (4.1,0);
\draw [dotted](3.8,-1) -- (4.1,-1);
\draw [dotted](3.8,-2) -- (4.1,-2);
\draw [dotted](3.8,-3) -- (4.1,-3);
\draw [->](4.2,0) -- (4.5,0);

\draw [white] (5.2,0) -- (5.2,0) node {$\color{black}s_{B+T+1}$};
\draw [->](5.2,-.3) -- (5.2,-.7);
\draw [white] (5.2,-1) -- (5.2,-1) node {$\color{black}f_{B+T+1}$};
\draw [->](5.2,-1.3) -- (5.2,-1.7);
\draw [white] (5.2,-2) -- (5.2,-2) node {$\color{black}f_{B+T+1}$};
\draw [->](5.2,-2.3) -- (5.2,-2.7);
\draw [white] (5.2,-3) -- (5.2,-3) node {$\color{black}s_{B+1}$};
\draw [->](6,0) -- (6.3,0);

\draw (4.7,-3.3) rectangle (5.7,-2.7);
\draw (5.2,-3.3) -- (5.2,-3.5);
\draw [->](5.2,-3.5) -- (2.65,-3.5);

\draw [white] (7,0) -- (7,0) node {$\color{black}s_{B+T+2}$};
\draw [->](7,-.3) -- (7,-.7);
\draw [white] (7,-1) -- (7,-1) node {$\color{black}f_{B+T+2}$};
\draw [->](7,-1.3) -- (7,-1.7);
\draw [white] (7,-2) -- (7,-2) node {$\color{black}\star$};
\draw [->](7,-2.3) -- (7,-2.7);
\draw [white] (7,-3) -- (7,-3) node {$\color{black}s_{B+2}$};
\draw [->](7.7,0) -- (8.,0);

\draw [dotted](8.1,0) -- (8.4,0);
\draw [dotted](8.1,-1) -- (8.4,-1);
\draw [dotted](8.1,-2) -- (8.4,-2);
\draw [white] (9.2,0) -- (9.2,0) node {$\color{black}s_{B+2T+1}$};
\draw [->](9.2,-.3) -- (9.2,-.7);
\draw [white] (9.2,-1) -- (9.2,-1) node {$\color{black}f_{B+2T+1}$};
\draw [->](9.2,-1.3) -- (9.2,-1.7);
\draw [white] (9.2,-2) -- (9.2,-2) node {$\color{black}\star$};
\draw [->](9.2,-2.3) -- (9.2,-2.7);
\draw [white] (9.2,-3) -- (9.2,-3) node {$\color{black}s_{B+T+1}$};

\draw [dotted](10,0) -- (10.3,0);
\draw [dotted](10,-1) -- (10.3,-1);
\draw [dotted](10,-2) -- (10.3,-2);
\draw [white] (11.1,0) -- (11.1,0) node {$\color{black}s_{2B+T+1}$};
\draw [->](11.1,-.3) -- (11.1,-.7);
\draw [white] (11.1,-1) -- (11.1,-1) node {$\color{black}f_{2B+T+1}$};
\draw [->](11.1,-1.3) -- (11.1,-1.7);
\draw [white] (11.1,-2) -- (11.1,-2) node {$\color{black}\star$};
\draw [->](11.1,-2.3) -- (11.1,-2.7);
\draw [white] (11.1,-3) -- (11.1,-3) node {$\color{black}-$};
\draw [->](11.8,0) -- (12.1,0);

\draw [white] (12.8,0) -- (12.8,0) node {$\color{black}s_{2B+T+2}$};
\draw [->](12.8,-.3) -- (12.8,-.7);
\draw [white] (12.8,-1) -- (12.8,-1) node {$\color{black}f_{2B+T+2}$};
\draw [->](12.8,-1.3) -- (12.8,-1.7);
\draw [white] (12.8,-2) -- (12.8,-2) node {$\color{black}f_{2B+T+2}$};
\draw [->](12.8,-2.3) -- (12.8,-2.7);
\draw [white] (12.8,-3) -- (12.8,-3) node {$\color{black}-$};

\draw [dotted](13.5,0) -- (13.8,0);
\draw [dotted](13.5,-1) -- (13.8,-1);
\draw [dotted](13.5,-2) -- (13.8,-2);
\draw [dotted](13.5,-3) -- (13.8,-3);
\draw [->](13.9,0) -- (14.2,0);

\draw [white] (14.9,0) -- (14.9,0) node {$\color{black}s_{2B+2T+2}$};
\draw [->](14.9,-.3) -- (14.9,-.7);
\draw [white] (14.9,-1) -- (14.9,-1) node {$\color{black}f_{2B+2T+2}$};
\draw [->](14.9,-1.3) -- (14.9,-1.7);
\draw [white] (14.9,-2) -- (14.9,-2) node {$\color{black}f_{2B+2T+2}$};
\draw [->](14.9,-2.3) -- (14.9,-2.7);
\draw [white] (14.9,-3) -- (14.9,-3) node {$\color{black}s_{2B+T+2}$};

\draw (14.1,-3.3) rectangle (15.7,-2.7);
\draw (14.9,-3.3) -- (14.9,-3.5);
\draw [->](14.9,-3.5) -- (11.95,-3.5);

\draw [dashed](11.95,-3.5) -- (11.95,.5);

\draw (-.5,-3.8) rectangle (6,.5);
\draw (6.2,-3.8) rectangle (15.9,.5);
\draw [white] (2.7	,-4.1) -- (2.7,-4.1) node {$\color{black}\textrm{First period with length $p$ symbols}$};
\draw [white] (11	,-4.1) -- (11,-4.1) node {$\color{black}\textrm{Second period with length $p$ symbols}$};
\end{tikzpicture}
\caption{ Periodic erasure channel considered in proof of converse.}
\label{fig:setup-converse}
\end{figure*}

The basic idea behind the converse is illustrated in Fig.~\ref{fig:setup-converse}. We consider a periodic erasure channel with period $p = B+T+1$. 
The $k-$th period, for $k \ge 1$, spans the interval $[(k-1)(B+T+1)+1, k(B+T+1)]$.  
In each period the first $B$ packets are erased,
whereas the remaining $T+1$ packets $f_{kB + k + (k-1)T},\ldots, f_{k(B+T+1)}$ are not erased. For sake of convenience we denote
the lower and upper end-points of the $k-$th interval by $l_k = (k-1)(B+T+1)+1$ and $u_k = k(B+T+1)$.  The beginning of the un-erased 
symbols in the $k-$th interval is $e_k = kB+k+(k-1)T$. Furthermore for sake of compactness we denote the $n-$ letter sequence $\rvs^n$ by $\rvbs$ i.e., using the bold-face font.

We provide a heuristic argument that is then formalized below. Consider the first period spanning time $[1,B+T+1]$. Recall that the first $B$ channel packets are erased. The source sequence $\rvbs_{B+1}$ corresponding to the first un-erasred channel packet is recovered at the end of the period i.e., by time $t = B+T+1$. As soon as this is recovered the decoding of the remaining source sequences in $[B+2, B+T+1]$ is transparent to any previous erasures due to the Markov nature of the source. 

Thus for the recovery of sources $\rvbs_{B+2},\ldots, \rvbs_{B+T+1}$,  the relevant erasure burst of length $B$ spans the interval $[B+T+2, 2B+T+1]$. All these source sequences are recovered by their deadline and in particular before the end of the second period.

Thus continuing this argument, if we consider a total of $N$ periods then we have a total of $N(T+1)$  channel packets and recover $\{\rvbs_{e_k},\ldots, \rvbs_{u_k}\}_{1\le k \le {N-1}}$. Thus we have
\begin{align}
&~N(T+1) n H(\rvf) \ge  H\left(\{\rvbs_{e_k},\ldots,\rvbs_{u_k}\}_{k=1}^{N-1}|\rvbs_0\right)\\
&= (N-1)n H(\rvs_{B+1},\ldots, \rvs_{B+T+1}|\rvs_0)
\end{align}
which reduces to~\eqref{eq:R0_Cap} as we take $N\rightarrow\infty$.

For the formal converse first observe that,
\begin{align}
& H(f_{e_1}^{u_1}, f_{e_2}^{u_2},\ldots, f_{e_N}^{u_N}|\rvs_0^n)\notag \\
&\ge H(\rvbs_{e_1}, f_{e_1}^{u_1},f_{e_2}^{u_2},\ldots, f_{e_N}^{u_N}|\rvbs_0) - H(\rvbs_{e_1} | f_{e_1}^{u_1},\rvbs_0)\\
&=H(\rvbs_{e_1}, f_{e_1}^{u_1},f_{e_2}^{u_2},\ldots, f_{e_N}^{u_N}|\rvbs_0) - n\eps_n \label{eq:Fano_SB}\\
&=H(\rvbs_{e_1}, f_{e_1}|\rvbs_0) \notag\\ &\qquad + H(f_{e_1+1}^{u_1}, f_{e_2}^{u_2}, \ldots, f_{e_N}^{u_N}|\rvbs_{e_1}, f_{e_1}, \rvbs_0) - n\eps_n \\
&\ge nH(\rvs_{B+1}|\rvs_0) \notag\\&\qquad + H(f_{e_1+1}^{u_1},f_{e_2}^{u_2},\ldots, f_{e_N}^{u_N}|\rvbs_{e_1}, f_0^{e_1}, \rvbs_0) - n\eps_n \label{eq:Cond_SB}
\end{align}
where we use the Fano's inequality in~\eqref{eq:Fano_SB} since $\rvbs_{B+1}$ 
can be recovered from $(\rvbs_0, f_{B+1}^{B+T+1})$ due to the delay 
constraint of $T$ symbols while~\eqref{eq:Cond_SB} follows from the fact that conditioning in the second term only reduces entropy.

We further simplify the second term in~\eqref{eq:Cond_SB} as follows:
\begin{align}
&H(f_{e_1+1}^{u_1},f_{e_2}^{u_2},\ldots, f_{e_N}^{u_N}|\rvbs_{e_1}, f_0^{e_1}, \rvbs_0) \notag\\
&\ge H(\rvbs_{e_1+1}^{u_1},f_{e_1+1}^{u_1},f_{e_2}^{u_2},\ldots, f_{e_N}^{u_N}|\rvbs_{e_1}, f_0^{e_1}, \rvbs_0)\notag\\&\qquad -H(\rvbs_{e_1+1}^{u_1}| f_{0}^{u_1},f_{e_2}^{u_2})\\
&\ge H(\rvbs_{e_1+1}^{u_1},f_{e_1+1}^{u_1},f_{e_2}^{u_2},\ldots, f_{e_N}^{u_N}|\rvbs_{e_1}, f_0^{e_1}, \rvbs_0)-n\eps_n\label{eq:Fano_SB2}\\
&=H(\rvbs_{e_1+1}^{u_1},f_{e_1+1}^{u_1}|\rvbs_{e_1}, \rvbs_0, f_0^{e_1}) \notag\\&\qquad + H(f_{e_2}^{u_2},\ldots, f_{e_N}^{u_N}|\rvbs_0,\rvbs_{e_1}^{u_1}, f_0^{u_1})-n\eps_n\\
&\ge H(\rvbs_{e_1+1}^{u_1}|\rvbs_{e_1}) + H(f_{e_2}^{u_2},\ldots, f_{e_N}^{u_N}|\rvbs_0,\rvbs_{e_1}^{u_1}, f_0^{u_1})-n\eps_n\label{eq:Markov_SB}\\
&\ge nTH(\rvs_1|\rvs_0) + H(f_{e_2}^{u_2},\ldots, f_{e_N}^{u_N}|\rvbs_0,\rvbs_{e_1}^{u_1}, f_0^{u_1})-n\eps_n
\label{eq:Mem_SB}\end{align}
where~\eqref{eq:Fano_SB2} follows from the application of Fano's inequality since all the sequences $\rvs_{B+2}^n,\ldots, \rvs_{B+T+1}^n$ are recovered by time $u_2 = 2B+2T+2$ when the $B$ packets in the interval $[B+T+2, 2B+T+1]$ are erased,~\eqref{eq:Markov_SB} follows from the fact that $(\rvs_0^n, f_0^{B+1}) \rightarrow \rvs_{B+1}^n \rightarrow (\rvs_{B+2}^n,\ldots, \rvs_{B+T+1}^n)$ and~\eqref{eq:Mem_SB} follows because the source sequences are memoryless and form a Markov chain. 

Following the same steps as~\eqref{eq:Cond_SB} and~\eqref{eq:Mem_SB}, we have
\begin{align}
&H(f_{e_2}^{u_2},\ldots, f_{e_N}^{u_N}|\rvbs_0,\rvbs_{e_1}^{u_1}, f_0^{u_1}) \\ &= nH(\rvs_{B+1}|\rvs_0)\notag\\
&\hspace{.3cm} + H(f_{e_2+1}^{u_2}, f_{e_3}^{u_3},\ldots, f_{e_N}^{u_N}| \rvbs_0, \rvbs_{e_1}^{u_1},\rvbs_{e_2}, f_0^{e_2})\\
&\ge nH(\rvs_{B+1}|\rvs_0) + nTH(\rvs_1|\rvs_0) \notag\\ &\qquad + H(f_{e_3}^{u_3},\ldots, f_{e_N}^{u_N}| \rvbs_0, \rvbs_{e_1}^{u_1},\rvbs_{e_2}^{u_2}, f_0^{u_2})
\end{align}

Continuing these steps, we have that
\begin{align}
N(T+1)nR &\ge H(f_{e_1}^{u_1}, f_{e_2}^{u_2},\ldots, f_{e_N}^{u_N}|\rvs_0^n) \\
&\ge nNH(\rvs_{B+1}|\rvs_0) + nT(N-1) H(\rvs_1|\rvs_0)  \notag \\ &\qquad + H(\rvf_{e_N+1}^{u_N}| \rvbs_0^{e_N},f_0^{e_N})-nN\eps_n
\end{align}

Dividing by $N(T+1)n$ and taking $n\rightarrow\infty$ and thereafter $N\rightarrow\infty$  we recover~\eqref{eq:R0_Cap}.

\section{Conclusions}
\label{sec:concl}
We introduce an information theoretic framework to characterize the fundamental tradeoff between compression efficiency and error propagation in video streaming systems. We introduce the  {\em rate-recovery function} and develop upper and lower bounds on this function. The lower bound is established by drawing connection to a periodic erasure channel and a multi-terminal source coding problem. We show that for the first-order Markov sources the rate-recovery function equals the sum of the ideal predictive coding rate and another term that decreases as $\frac{1}{W+1}$. For the class of linear deterministic Markov sources and i.i.d.\ Gaussian sources with a sliding-window recovery constraint we propose a new coding technique --- prospicient coding ---- that achieve the rate-recovery function. Numerical results indicate significant gains over traditional techniques such as the FEC based schemes. For the class of symmetric sources and memoryless encoding the optimality of a random binning based scheme is established by drawing connection to the Zig-Zag source network problem. The optimality of binning is also established when the error recovery window is of length zero. 

Several open problems remain in our proposed framework. A complete characterization of the rate-recovery function remains to be obtained. Better lower bounds can potentially be obtained by considering more elaborate schemes rather than the binning based technique. Finally extension of this framework to lossy reconstructions beyond what has been considered in this paper is also a very fruitful area of research.

\appendices
\section{Proof of Corollary~\ref{corol:genUB}}
\label{app:Cor1}
According to the chain rule of entropies, the term in \eqref{eq:genUB_Slepian-Wolf} can be written as 
\begin{align}
&H(\rvs_{B+1}, \rvs_{B+2},\ldots, \rvs_{B+W+1}|\rvs_{0})\\
&=H(\rvs_{B+1}|\rvs_{0})+ \sum_{k=1}^{W}H(\rvs_{B+k+1}|\rvs_{0}, \rvs_{B+1},\ldots,\rvs_{B+k})\notag\\
&=H(\rvs_{B+1}|\rvs_{0})+ W H(\rvs_{1}|\rvs_{0})\label{eq:cr1}\\
&=H(\rvs_{B+1}|\rvs_{0})-H(\rvs_{B+1}|\rvs_{B}, \rvs_{0})\notag\\
&\hspace{2cm}+H(\rvs_{B+1}|\rvs_{B}, \rvs_{0})+ W H(\rvs_{1}|\rvs_{0})\label{eq:cr2}\\
&=H(\rvs_{B+1}|\rvs_{0})-H(\rvs_{B+1}|\rvs_{B}, \rvs_{0})\notag\\
&\hspace{2cm}+H(\rvs_{B+1}|\rvs_{B})+ W H(\rvs_{1}|\rvs_{0})\label{eq:cr3}\\
&=I(\rvs_{B+1};\rvs_{B}|\rvs_{0})+(W+1)H(\rvs_{1}|\rvs_{0})\label{eq:app-cor}\\
&=(W+1)R^{+}(B,W)
\end{align} where \eqref{eq:cr1} follows from the Markov property
\begin{align}
\rvs_{0}, \rvs_{B+1}, \ldots,\rvs_{B+k-1}\rightarrow \rvs_{B+k} \rightarrow  \rvs_{B+k+1} \label{eq:Mar}
\end{align} for any $k$ and from the temporally independency and stationarity of the sources which for each $k$ implies that
\begin{align}
H(\rvs_{B+k+1}|\rvs_{B+k})=H(\rvs_{1}|\rvs_{0}).
\end{align} Note that in \eqref{eq:cr2} we add and subtract the same term and \eqref{eq:cr3} also follows from the Markov property of \eqref{eq:Mar} for $k=0$.

\section{Proof of Lemma~\ref{lem:full-rank}}
\label{app:full-rank}
First let us define the following notations. 
\begin{itemize}
\item For a vector $\bx$ of size $x$, define $\bx^{(u,a)}$ and $\bx^{(d,a)}$ such that 
\begin{align}
\bx=\kbordermatrix{\\
a&\bx^{(u,a)}\\
(x-a)&\bx^{(d,a)}},
\label{lem-lem-1}
\end{align}
\item For a matrix $\bX$ of size $x\times y$, define $\bX^{(l,a)}$, $\bX^{(r,a)}$, $\bX^{(u,b)}$ and $\bX^{(d,b)}$ as 
\begin{align}
\bX=\kbordermatrix{&a&(y-a)\\
&\bX^{(l,a)}&\bX^{(r,a)}},
\label{lem-lem-2-1}
\end{align}
and 
\begin{align}
\bX=\kbordermatrix{\\
b&\bX^{(u,b)}\\
(x-b)&\bX^{(d,b)}},
\label{lem-lem-2-2}
\end{align}
\item For a square matrix $\bX$ of size $x$, define matrices $\bX^{(ul,a)}$, $\bX^{(ur,a)}$, $\bX^{(dl,a)}$ and $\bX^{(dr,a)}$ such that
\begin{align}
\bX=\kbordermatrix{&a&(x-a)\\
a&\bX^{(ul,a)}&\bX^{(ur,a)}\\
(x-a)&\bX^{(dl,a)}&\bX^{(dr,a)}}.
\label{lem-lem-4}
\end{align}
\end{itemize}
We introduce an iterative method to define the transformation $\mathcal{L}_f$.\\
{\bf Step $0$: } If $\bA=\mathrm{0}$ or $N_1=N_d$, the source is in the form of~\eqref{full-rank-rep}. Thus $\mathcal{L}_f(\bs_{i})=\bs_{i}$. Otherwise, continue to next step.\\ 
{\bf Step $1$: } Without loss of generality we assume that the first $N_1$ rows of matrix $\bA$ are independent.\footnote{By rearranging the rows of matrices $\bA$ and $\bB$, this assumption can always be satisfied.}
Let $\bR_{1,0}$ denotes the first $N_1$ rows of $\bA$ and
\begin{align}
{\bA}^{(d,N_1)}=\bV_1\bR_{1,0}
\end{align} where $\bV_1$ is an $(N_d-N_1)\times N_1$ matrix relating dependent rows of $\bA$ to $\bR_{1,0}$.
Also define invertible square matrix $\bM_1$ as  
\begin{align}
\bM_1 \triangleq \kbordermatrix{&{N}_1&(N_d-{N}_1)\\
{N}_1&\mathrm{I}&\mathrm{0}\\
(N_d-{N}_1)&-\bV_1&\mathrm{I}}.
\end{align}Note that 
\begin{align}
\bM_1^{-1}= \kbordermatrix{\\
&\mathrm{I}&\mathrm{0}\\
&\bV_1&\mathrm{I}}.
\end{align} Define 
\begin{align}
\begin{pmatrix}
\bs_{i,1}\\
\bar{\bs}_{i,1}
\end{pmatrix} \triangleq \begin{pmatrix}
(\bM_1 \bs_{i,d})^{(u,N_1)}\\
(\bM_1 \bs_{i,d})^{(d,N_1)}
\end{pmatrix}=\bM_1 \bs_{i,d}.
\end{align} We have 
\begin{align}
&\begin{pmatrix}
\bs_{i,1}\\
\bar{\bs}_{i,1}
\end{pmatrix}
=\begin{pmatrix}
\bM_1\bA& \bM_1\bB\bM_1^{-1} 
\end{pmatrix}\begin{pmatrix}
\bs_{i-1,0}\\
\bM_1\bs_{i-1,d}
\end{pmatrix}\\
&=\begin{pmatrix}
\bR_{1,0}& (\bM_1\bB\bM_1^{-1})^{(ul,N_1)} & (\bM_1\bB\bM_1^{-1})^{(ur,N_1)} \\ 
\mathrm{0}& (\bM_1\bB\bM_1^{-1})^{(dl,N_1)} & (\bM_1\bB\bM_1^{-1})^{(dr,N_1)}
\end{pmatrix}\notag\\
&\hspace{1in}\times \begin{pmatrix}
\bs_{i-1,0}\\
(\bM_1\bs_{i-1,d})^{(u,N_1)}\\
(\bM_1\bs_{i-1,d})^{(d,N_1)}
\end{pmatrix}\\
&=\kbordermatrix{&N_0 &N_1&N_d-N_1\\
N_1&\bR_{1,0}& \bR_{1,1}& \bR'_{1,2} \\ 
N_d-N_1&\mathrm{0}& \bA^{(1)}& \bB^{(1)}} \begin{pmatrix}
\bs_{i-1,0}\\
\bs_{i-1,1}\\
\bar{\bs}_{i-1,1}\label{eq:Appen-1}
\end{pmatrix}
\end{align} where $\bA^{(1)}=(\bM_1\bB\bM_1^{-1})^{(dl,N_1)}$ and $\bB^{(1)}=(\bM_1\bB\bM_1^{-1})^{(dr,N_1)}$ and the other matrices are defined similarly.  Till now $\bs_{i,1}$ is defined. \\
{\bf Step $2$: } Define $N_2\triangleq \textrm{Rank}(\bA^{(1)})$. Generally 
\begin{align}
N_2\le \min\{N_1,N_d-N_1\}\label{N2-min}
\end{align} If $N_2=N_d-N_1$ or if $\bA^{(1)}$ is zero matrix, set $\bs_{i,2}=\bar{\bs}_{i,1}$ and 
\begin{align}
\mathcal{L}_f(\bs_{i})=\begin{pmatrix}
\bs_{i,0}\\
\bs_{i,1}\\
\bs_{i,2}
\end{pmatrix}
\end{align}
If $\bA^{(1)}\neq 0$ and $N_2< N_d-N_1$, again we assume that the first $N_2$ rows of $\bA^{(1)}$ denoted by $\bR_{2,1}$ contains independent rows and
\begin{align}
{\bA}^{(1)(d,N_2)}=\bV_2\bR_{2,1}.
\end{align} Also define invertible matrix $\bM_2$ as
\begin{align}
\bM_2\triangleq \kbordermatrix{&{N}_2&(N_d-{N}_1-N_2)\\
{N}_2&\mathrm{I}&\mathrm{0}\\
(N_d-{N}_1-N_2)&-\bV_2&\mathrm{I}}.
\end{align} and 
\begin{align}
\begin{pmatrix}
{\bs}_{i,2}\\
\bar{\bs}_{i,2}
\end{pmatrix}\triangleq \begin{pmatrix}
(\bM_2\bar{\bs}_{i,1})^{u,N_2}\\
(\bM_2\bar{\bs}_{i,1})^{d,N_2}
\end{pmatrix}=\bM_2\bar{\bs}_{i,1}
\end{align}
We have 
\begin{align}
&\begin{pmatrix}
\bs_{i,1}\\
{\bs}_{i,2}\\
\bar{\bs}_{i,2}
\end{pmatrix}=\notag\\
&\begin{pmatrix}
\bR_{1,0}&\bR_{1,1}&\bR'_{1,2}\bM_2^{-1}\\
\mathrm{0}&\bM_2\bA^{(1)}& {\bM_2}\bB^{(1)}\bM_2^{-1}
\end{pmatrix}\begin{pmatrix}
\bs_{i-1,0}\\
\bs_{i-1,1}\\
\bM_2\bar{\bs}_{i-1,1}
\end{pmatrix}\label{step-2-1}
\end{align}
and \eqref{step-2-1} is equivalent to \eqref{step-2-2} which can be written as 
\begin{figure*}
\begin{align}\label{step-2-2}
\begin{pmatrix}
\bs_{i,1}\\
{\bs}_{i,2}\\
\bar{\bs}_{i,2}
\end{pmatrix}=\begin{pmatrix}
\bR_{1,0}&\bR_{1,1}&(\bR'_{1,2}\bM_2^{-1})^{l,N_2}& (\bR'_{1,2}\bM_2^{-1})^{r,N_2}\\
\mathrm{0}&(\bM_2\bA^{(1)})^{u,N_2}& ({\bM_2}\bB^{(1)}\bM_2^{-1})^{ul,N_2} &({\bM_2}\bB^{(1)}\bM_2^{-1})^{ur,N_2}\\
\mathrm{0}&(\bM_2\bA^{(1)})^{d,N_2}& ({\bM_2}\bB^{(1)}\bM_2^{-1})^{dl,N_2} &({\bM_2}\bB^{(1)}\bM_2^{-1})^{dr,N_2}\\ 
\end{pmatrix}\begin{pmatrix}
\bs_{i-1,0}\\
\bs_{i-1,1}\\
(\bM_2\bar{\bs}_{i-1,1})^{(u,N_2)}\\
(\bM_2\bar{\bs}_{i-1,1})^{(l,N_2)}
\end{pmatrix}
\end{align}
\end{figure*}
\begin{align}
\begin{pmatrix}
\bs_{i,1}\\
{\bs}_{i,2}\\
\bar{\bs}_{i,2}
\end{pmatrix}=\begin{pmatrix}
\bR_{1,0}&\bR_{1,1}&\bR_{1,2}& \bR'_{1,3}\\
\mathrm{0}&\bR_{2,1}& \bR_{2,2}&\bR'_{2,3}\\
\mathrm{0}&\mathrm{0}& {\bA}^{(2)}&\bB^{(2)}\\ 
\end{pmatrix}\begin{pmatrix}
\bs_{i-1,0}\\
\bs_{i-1,1}\\
\bs_{i-1,2}\\
\bar{\bs}_{i-1,1}
\end{pmatrix}
\end{align}
Note that $\bs_{i,2}$ is defined in this step.\\
This procedure can be repeated through next steps until $(K-1)$th step where $\bA^{(K-1)}$ is either full-rank of rank $N_K$ or zero matrix. In this step define $\bR_{K,K-1}=\bA^{(K-1)}$ and $\bs_{i,K}=\bar{\bs}_{i,K-1}$. The result is
\begin{align}
\hat{\bs}_{i}=\mathcal{L}_{f}(\bs_{i})=\begin{pmatrix}
\bs_{i,0}\\
\vdots\\
\bs_{i,K}
\end{pmatrix}.
\end{align}
Similar to \eqref{N1-min} and \eqref{N2-min}, \eqref{ranks} can be verified for all the steps. Note that all the steps are invertible. This completes the proof of lemma~\ref{lem:full-rank}.

\section{Proof of Lemma~\ref{lem:bf}}
\label{app:bf}
Consider a source $\hat{\bs}_{i}$ consisting of $N_0$ innovation bits and $K$ deterministic sub-symbols $\hat{\bs}_{i,d}$ defined in \eqref{full-rank-rep}. The following iterative method characterizes the transformation $\mathcal{L}_{b}$. \\
{\bf Step $0$: } If $\bR_{K,K-1}=\mathrm{0}$, we have 
\begin{align}
\bs_{i,K}&=\bR_{K,K} \bs_{i-1,K}\\
&=\bR_{K,K}^{i+1}\bs_{-1,K}
\end{align} Note that $\bs_{-1}$, and thus $\bs_{-1,K}$, is known at the decoder. Therefore, we can eliminate sub-symbol $\bs_{.,K}$ and consider the source $\hat{\hat{\bs}}_{i}$ with $N_0$ innovation bits and deterministic bits characterized by \eqref{full-rank-rep-2}, at the top of this page, and continue to  next step with $K-1$. Note that knowing $\hat{\hat{\bs}}_{i}$, $\hat{\bs}_{i}$ can be constructed. 
\begin{figure*}
\begin{align}
&\hat{\hat{\bs}}_{i,d}=\begin{pmatrix}
\hat{\hat{\bs}}_{i,1}\\
\hat{\hat{\bs}}_{i,2}\\
\vdots\\
\hat{\hat{\bs}}_{i,K-1}
\end{pmatrix}=\begin{pmatrix}
\bR_{1,0} & \bR_{1,1} & \cdots & \bR_{1,K-2} &\bR_{1,K-1} \\
\mathrm{0} & \bR_{2,1} & \cdots & \bR_{2,K-2} & \bR_{2,K-1} \\
\vdots & \vdots & \ddots & \vdots &\vdots \\
\mathrm{0} & \mathrm{0} & \cdots & \bR_{K-1,K-2} &\bR_{K-1,K-1} 
\end{pmatrix}\begin{pmatrix}
\hat{\hat{\bs}}_{i-1,0}\\
\hat{\hat{\bs}}_{i-1,1}\\
\vdots\\
\hat{\hat{\bs}}_{i-1,K-2}\\
\hat{\hat{\bs}}_{i-1,K-1}
\end{pmatrix}, \label{full-rank-rep-2}
\end{align}    
\end{figure*}

If $\bR_{K,K-1}$ is full-rank of rank $K$, continue to next step.\\
{\bf Step $1$: } Define 
\begin{align}
\begin{pmatrix}
\tilde{\bs}_{i,K-1}\\
\tilde{\bs}_{i,K}
\end{pmatrix}\triangleq \begin{pmatrix}
\bI_{N_{K-1}}& \bX_1\\
\mathrm{0}& \bI_{N_k}
\end{pmatrix}\begin{pmatrix}
\bs_{i,K-1}\\
\bs_{i,K}
\end{pmatrix}
\end{align} and 
\begin{align}
\bD_1 \triangleq \kbordermatrix{&{\sum_{j=0}^{K-2}N_j}&N_{K-1}&N_K\\
{\sum_{j=0}^{K-2}N_j}&\bI& \mathrm{0}&\mathrm{0}\\
N_{K-1}&\mathrm{0}& \bI& \bX_1\\
N_{K}&\mathrm{0}& \mathrm{0}& \bI}
\end{align}  and note that
\begin{align}
\bD^{-1}_1 = \begin{pmatrix}
\bI& \mathrm{0}&\mathrm{0}\\
\mathrm{0}& \bI & -\bX_1\\
\mathrm{0}& \mathrm{0}& \bI
\end{pmatrix}
\end{align} Also $\bX_1$ can be defined such that 
\begin{align}
{\bR}_{K,K}-\bR_{K,K-1}\bX_1=\mathrm{0}.
\end{align}
By these definitions, \eqref{full-rank-rep} can be reformulated to get \eqref{full-rank-rep-3}.
\begin{figure*}
\begin{align}
\begin{pmatrix}
{\bs}_{i,1}\\
{\bs}_{i,2}\\
\vdots\\
\tilde{\bs}_{i,K-1}\\
\tilde{\bs}_{i,K}
\end{pmatrix}&=
\bD_1^{(rd,N_0)}\begin{pmatrix}
\bR_{1,0} & \bR_{1,1} & \cdots & \bR_{1,K-2} &\bR_{1,K-1} &\bR_{1,K}\\
\mathrm{0} & \bR_{2,1} & \cdots & \bR_{2,K-2} & \bR_{2,K-1} &\bR_{2,K}\\
\vdots & \vdots & \ddots & \vdots &\vdots &\vdots\\
\mathrm{0} & \mathrm{0} & \cdots & \bR_{K-1,K-2} &\bR_{K-1,K-1}& \bR_{K-1,K}\\
\mathrm{0} & \mathrm{0} & \cdots & \mathrm{0} & {\bR}_{K,K-1} &{\bR}_{K,K}\\
\end{pmatrix}\bD^{-1}_1\begin{pmatrix}
{\bs}_{i-1,0}\\
{\bs}_{i-1,1}\\
\vdots\\
{\bs}_{i-1,K-2}\\
\tilde{\bs}_{i-1,K-1}\\
\tilde{\bs}_{i-1,K}
\end{pmatrix}\notag\\
&=\begin{pmatrix}
\bR_{1,0} & \bR_{1,1} & \cdots & \bR_{1,K-2} &\tilde{\bR}^{(1)}_{1,K-1} &\tilde{\bR}^{(1)}_{1,K}\\
\mathrm{0} & \bR_{2,1} & \cdots & \bR_{2,K-2} & \tilde{\bR}^{(1)}_{2,K-1} &\tilde{\bR}^{(1)}_{2,K}\\
\vdots & \vdots & \ddots & \vdots &\vdots &\vdots\\
\mathrm{0} & \mathrm{0} & \cdots & \bR_{K-1,K-2} &\tilde{\bR}^{(1)}_{K-1,K-1}& \tilde{\bR}^{(1)}_{K-1,K}\\
\mathrm{0} & \mathrm{0} & \cdots & \mathrm{0} & {\bR}_{K,K-1} &\mathrm{0}\\
\end{pmatrix}\begin{pmatrix}
{\bs}_{i-1,0}\\
{\bs}_{i-1,1}\\
\vdots\\
{\bs}_{i-1,K-2}\\
\tilde{\bs}_{i-1,K-1}\\
\tilde{\bs}_{i-1,K}
\end{pmatrix}, \label{full-rank-rep-3}
\end{align}    
\end{figure*}
Matrices $\tilde{\bR}^{(1)}_{(.,.)}$ can be defined accordingly. \\
{\bf Step $j\in [2:K]$: } Define $l=K-j$. At step $j$, the source is transformed into the form of \eqref{full-rank-rep-4}.
\begin{figure*}
\begin{align}
\begin{pmatrix}
{\bs}_{i,1}\\
\vdots\\
{\bs}_{i,l}\\
\tilde{\bs}_{i,l+1}\\
\tilde{\bs}_{i,l+2}\\
\vdots\\
\tilde{\bs}_{i,K}
\end{pmatrix}&
=\underbrace{\begin{pmatrix}
\bR_{1,0} & \cdots & \bR_{1,l-1} &{\bR}_{1,l} &\tilde{\bR}^{(j-1)}_{1,l+1}&\tilde{\bR}^{(j-1)}_{1,l+2}&\cdots &\tilde{\bR}^{(j-1)}_{1,K-1}&\tilde{\bR}^{(j-1)}_{1,K}\\
\vdots & \ddots & \vdots & \vdots &\vdots &\vdots & \ddots & \vdots&\vdots\\
\mathrm{0}& \cdots & \bR_{l,l-1} &  \bR_{l,l} & \tilde{\bR}^{(j-1)}_{l,l+1} & \tilde{\bR}^{(j-1)}_{l,l+2} &\cdots&\tilde{\bR}^{(j-1)}_{l,K-1}&\tilde{\bR}^{(j-1)}_{l,K}\\
\mathrm{0}& \cdots & \mathrm{0}  &  \bR_{l+1,l} & \tilde{\bR}^{(j-1)}_{l+1,l+1} & \tilde{\bR}^{(j-1)}_{l+1,l+2} &\cdots&\tilde{\bR}^{(j-1)}_{l+1,K-1}&\tilde{\bR}^{(j-1)}_{l+1,K}\\
\mathrm{0}& \cdots & \mathrm{0}  &  \mathrm{0}   & {\bR}_{l+2,l+1} & \mathrm{0} &\cdots&\mathrm{0}&\mathrm{0}\\
\vdots & \ddots & \vdots & \vdots &\vdots &\vdots & \ddots & \vdots& \vdots\\
\mathrm{0} & \cdots & \mathrm{0}  & \mathrm{0} &  \mathrm{0}  &  \mathrm{0} &\cdots & \bR_{K,K-1}&\mathrm{0}\\
\end{pmatrix}}_{\Psi^{(j-1)}}\begin{pmatrix}
{\bs}_{i-1,0}\\
\vdots\\
{\bs}_{i-1,l-1}\\
{\bs}_{i-1,l}\\
\tilde{\bs}_{i-1,l+1}\\
\tilde{\bs}_{i-1,l+2}\\
\vdots\\
\tilde{\bs}_{i-1,K-1}\\
\tilde{\bs}_{i-1,K}
\end{pmatrix}, \label{full-rank-rep-4}
\end{align}    
\end{figure*} Now define 
\begin{align}
\bD_j \triangleq \kbordermatrix{&{\sum_{j=0}^{l-1}N_j}&N_{l}&N_{l+1}&\cdots & N_K\\
{\sum_{j=0}^{l-1}N_j}&\bI& \mathrm{0}&\mathrm{0}&\cdots&\mathrm{0}\\
N_{l}&\mathrm{0}& \bI& \bX_{1,j}&\cdots&\bX_{j,j}\\
N_{l+1}&\mathrm{0}& \mathrm{0}& \bI &\cdots &\mathrm{0}\\
\vdots&\vdots&\vdots&\vdots&\ddots&\vdots\\
N_{K}&\mathrm{0}& \mathrm{0}&\mathrm{0} &\cdots &\bI\\
}
\end{align} and note that 
\begin{align}
\bD_j^{-1}= \begin{pmatrix}
\bI& \mathrm{0}&\mathrm{0}&\cdots&\mathrm{0}\\
\mathrm{0}& \bI& -\bX_{1,j}&\cdots&-\bX_{j,j}\\
\mathrm{0}& \mathrm{0}& \bI &\cdots &\mathrm{0}\\
\vdots&\vdots&\vdots&\ddots&\vdots\\
\mathrm{0}& \mathrm{0}&\mathrm{0} &\cdots &\bI\\
\end{pmatrix}.
\end{align}
Also define 
\begin{align}
\tilde{\bs}_{i,l} \triangleq \begin{pmatrix}
\bI& \bX_{1,j}&\bX_{2,j}& \cdots&\bX_{j,j}\\
\end{pmatrix}\begin{pmatrix}
{\bs}_{i,l}\\
\tilde{\bs}_{i,l+1}\\
\tilde{\bs}_{i,l+2}\\
\vdots\\
\tilde{\bs}_{i,K}
\end{pmatrix}
\end{align} By these definitions, \eqref{full-rank-rep-4} reduces to
\begin{align}
\begin{pmatrix}
{\bs}_{i,1}\\
\vdots\\
\bs_{i,l-1}\\
\tilde{\bs}_{i,l}\\
\tilde{\bs}_{i,l+1}\\
\tilde{\bs}_{i,l+2}\\
\vdots\\
\tilde{\bs}_{i,K}
\end{pmatrix}&
=\bD_{j}^{(dr,N_0)}\Psi^{(j-1)}\bD_{j}^{-1}\begin{pmatrix}
{\bs}_{i-1,0}\\
\vdots\\
{\bs}_{i-1,l-1}\\
\tilde{\bs}_{i-1,l}\\
\tilde{\bs}_{i-1,l+1}\\
\tilde{\bs}_{i-1,l+2}\\
\vdots\\
\tilde{\bs}_{i,K}
\end{pmatrix}, \label{full-rank-rep-5}
\end{align} By defining $\bX_{k,j}$s such that for each $k\in\{1,2,\ldots,j\}$
\begin{align}
\tilde{\bR}^{(j-1)}_{l+1,l+k}-\bR_{l+1,l}\bX_{k,j}=\mathrm{0},
\end{align} it is not hard to see that \eqref{full-rank-rep-5} can be rewritten as \eqref{full-rank-rep-6} whose $(l+1)$th row is block-diagonalized. 
\begin{figure*}
\begin{align}
\begin{pmatrix}
{\bs}_{i,1}\\
\vdots\\
\tilde{\bs}_{i,l}\\
\tilde{\bs}_{i,l+1}\\
\tilde{\bs}_{i,l+2}\\
\vdots\\
\tilde{\bs}_{i,K}
\end{pmatrix}&
=\underbrace{\begin{pmatrix}
\bR_{1,0} & \cdots & \bR_{1,l-1} &\tilde{\bR}^{(j)}_{1,l} &\tilde{\bR}^{(j)}_{1,l+1}&\tilde{\bR}^{(j)}_{1,l+2}&\cdots &\tilde{\bR}^{(j)}_{1,K-1}&\tilde{\bR}^{(j)}_{1,K}\\
\vdots & \ddots & \vdots & \vdots &\vdots &\vdots & \ddots & \vdots&\vdots\\
\mathrm{0}& \cdots & \bR_{l,l-1} &  \tilde{\bR}^{(j)}_{l,l} & \tilde{\bR}^{(j)}_{l,l+1} & \tilde{\bR}^{(j)}_{l,l+2} &\cdots&\tilde{\bR}^{(j)}_{l,K-1}&\tilde{\bR}^{(j)}_{l,K}\\
\mathrm{0}& \cdots & \mathrm{0}  &  \bR_{l+1,l} &  \mathrm{0}  &  \mathrm{0}  &\cdots&  \mathrm{0}  &  \mathrm{0} \\
\mathrm{0}& \cdots & \mathrm{0}  &  \mathrm{0}   & {\bR}_{l+2,l+1} & \mathrm{0} &\cdots&\mathrm{0}&\mathrm{0}\\
\vdots & \ddots & \vdots & \vdots &\vdots &\vdots & \ddots & \vdots& \vdots\\
\mathrm{0} & \cdots & \mathrm{0}  & \mathrm{0} &  \mathrm{0}  &  \mathrm{0} &\cdots & \bR_{K,K-1}&\mathrm{0}\\
\end{pmatrix}}_{\Psi^{(j)}}\begin{pmatrix}
{\bs}_{i-1,0}\\
\vdots\\
{\bs}_{i-1,l-1}\\
\tilde{\bs}_{i-1,l}\\
\tilde{\bs}_{i-1,l+1}\\
\tilde{\bs}_{i-1,l+2}\\
\vdots\\
\tilde{\bs}_{i-1,K-1}\\
\tilde{\bs}_{i-1,K}
\end{pmatrix}, \label{full-rank-rep-6}
\end{align}    
\end{figure*}

After these steps, the source $\bs_{i}$ is changed into the diagonally correlated Markov source $\tilde{\bs}_{i}$, with $N_0$ innovation bits $\tilde{\bs}_{i,0}$ and deterministic bits as \eqref{block-diag-rep}.
\begin{figure*}
\begin{align}
&\tilde{\bs}_{i,d}=\begin{pmatrix}
\tilde{\bs}_{i,1}\\
\tilde{\bs}_{i,2}\\
\vdots\\
\tilde{\bs}_{i,K-1}\\
\tilde{\bs}_{i,K}
\end{pmatrix}=\begin{pmatrix}
\bR_{1,0} & \mathrm{0} & \cdots & \mathrm{0} &\mathrm{0} \\
\mathrm{0} & \bR_{2,1} & \cdots & \mathrm{0} & \mathrm{0} \\
\vdots & \vdots & \ddots & \vdots &\vdots \\
\mathrm{0} & \mathrm{0} & \cdots & \bR_{K-1,K-2} &\mathrm{0} \\
\mathrm{0} & \mathrm{0} & \cdots & \mathrm{0} & {\bR}_{K,K-1} 
\end{pmatrix}\begin{pmatrix}
\tilde{\bs}_{i-1,0}\\
\tilde{\bs}_{i-1,1}\\
\vdots\\
\tilde{\bs}_{i-1,K-2}\\
\tilde{\bs}_{i-1,K-1}
\end{pmatrix}, \label{block-diag-rep}
\end{align}    
\end{figure*}
All the steps are invertible and this completes the proof.

\section{Generalization to non-integer $R_j$}
\label{app:gen}
Assume that $\tilde{R}_{j}$ are rational\footnote{For irrational rates, consider a rational numbers in the $\epsilon$-neighborhood of $\tilde{R}_{j}$s as $\epsilon\to 0$.}. There exists an $0< \alpha <1$ such that $\tilde{R}'_{j}=\tilde{R}_{j}/\alpha$ is integer for each $j\in\{0,1,\dots,B\}$. Define $n'=\alpha n$. Each codeword $\rvb_{i,j}\in\{1,\ldots,2^{n\tilde{R}_{j}}\}=\{1,\ldots,2^{n'\tilde{R}'_{j}}\}$ can be represented by $n'$ i.i.d. $\tilde{R}'_{j}$-length bit-sequences $\bb'^{n'}_{i,j}$. Similarly, for $j\in\{1,2,\ldots,B\}$ define 
\begin{align}
&R'_{j}=\sum_{k=j}^{B}\tilde{R}'_{k}=\sum_{k=j}^{B}\frac{\tilde{R}_{k}}{\alpha}= \frac{1}{2\alpha}\log(\frac{1}{d_{W+j}})\\
&R'_{0}=\sum_{k=0}^{B}\tilde{R}'_{k}= \sum_{k=0}^{B}\frac{\tilde{R}_{k}}{\alpha}= \frac{1}{2\alpha}\log(\frac{1}{d_{0}}).   
\end{align} Now the same coding scheme can be applied to get the rate 
\begin{align}
nR&=n'R'_{0}+\frac{1}{W+1}\sum_{k=1}^{B}n'R'_{k}\\
&=\frac{n}{2}\log(\frac{1}{d_{0}})+\frac{n}{2(W+1)}\sum_{k=1}^{B}\log(\frac{1}{d_{W+j}}).
\end{align}

\section{Proof of Claim~\ref{claim:Gauss-Ent}}
\label{app:Gauss-Ent}
We need to lower bound $H(\rvf_{kp+B}^{(k+1)p-1}|\rvf_{0}^{kp-1})$. Consider

\begin{align}
&H(\rvf_{kp+B}^{(k+1)p-1}|\rvf_{0}^{kp-1}) \\
&=I(\rvf_{kp+B}^{(k+1)p-1};\rvt_{(k+1)p-1}^n|\rvf_{0}^{kp-1}) \notag \\&\qquad+ H(\rvf_{kp+B}^{(k+1)p-1}|\rvf_{0}^{kp-1}, \rvt_{(k+1)p-1}^n)\\
&=h(\rvt_{(k+1)p-1}^n|\rvf_{0}^{kp-1}) -h(\rvt_{(k+1)p-1}^n|\rvf_{0}^{kp-1}, \rvf_{kp+B}^{(k+1)p-1}) \notag\\ 
&\qquad +  H(\rvf_{kp+B}^{(k+1)p-1}|\rvf_{0}^{kp-1}, \rvt_{(k+1)p-1}^n)\notag\\
&=h(\rvt_{(k+1)p-1}^n) -h(\rvt_{(k+1)p-1}^n|\rvf_{0}^{kp-1}, \rvf_{kp+B}^{(k+1)p-1}) \notag\\ 
&\qquad +  H(\rvf_{kp+B}^{(k+1)p-1}|\rvf_{0}^{kp-1}, \rvt_{(k+1)p-1}^n) \label{eq:exp1}
\end{align}
where~\eqref{eq:exp1} follows since ${\rvt_{(k+1)p-1}^n = (\rvs_{kp}^n,\ldots, \rvs_{(k+1)p-1}^n)}$ is independent of $\rvf_0^{kp-1}$ as the source sequences $\rvs_i^n$ are generated i.i.d.\ . By expanding $\rvt_{(k+1)p-1}^n$ we have that
\begin{align}
h(\rvt_{(k+1)p-1}^n)  &= h(\rvs_{kp}^n, \ldots, \rvs_{kp+B-1}^n) \notag\\ &+ h(\rvs_{kp+B}^n, \ldots, \rvs_{(k+1)p-1}^n), \label{eq:rel1}
\end{align}and
\begin{align}
&h(\rvt_{(k+1)p-1}^n|\rvf_{0}^{kp-1}, \rvf_{kp+B}^{(k+1)p-1}) \notag\\&=h(\rvs_{kp}^n, \ldots, \rvs_{kp+B-1}^n|\rvf_{0}^{kp-1}, \rvf_{kp+B}^{(k+1)p-1})+\notag\\&
\!h\!(\!\rvs_{kp+B}^n\!,\! \ldots,\! \rvs_{(k+1)p-1}^n\!|\!\rvf_{0}^{kp-1}\!,\! \rvf_{kp+B}^{(k+1)p-1}\!,\!\rvs_{kp}^n, \ldots\!,\! \rvs_{kp+B-1}^n\!) \label{eq:rel2}
\end{align}

We next show the following
\begin{multline}
h(\rvs_{kp}^n, \ldots, \rvs_{kp+B-1}^n)-\\h(\rvs_{kp}^n, \ldots, \rvs_{kp+B-1}^n|\rvf_{0}^{kp-1}, \rvf_{kp+B}^{(k+1)p-1})
\\ \ge \sum_{i=1}^{B}\frac{n}{2}\log{(\frac{1}{d_{W+i}})}\label{eq:Gauss_LB_T1}
\end{multline}
and that
\begin{align}
&h(\rvs_{kp+B}^n, \ldots, \rvs_{(k+1)p-1}^n|\rvf_{0}^{kp-1},\rvs_{kp}^n, \ldots, \rvs_{kp+B-1}^n)-\notag\\ &\!\!h(\!\rvs_{kp+B}^n, \!\ldots,\!\! \rvs_{(k+1)p-1}^n\!|\! \rvf_{0}^{kp-1},\!\rvf_{kp+B}^{(k+1)p-1}\!,\rvs_{kp}^n, \!\ldots, \!\rvs_{kp+B-1}^n\!)\notag\\ &+ H(\rvf_{kp+B}^{(k+1)p-1}|\rvf_{0}^{kp-1}, \rvt_{(k+1)p-1}^n) \ge \frac{n(W+1)}{2}\log(\frac{1}{d_0})\label{eq:Gauss_LB_T2}
\end{align}
The proof of {Claim~\ref{claim:Gauss-Ent}} follows from~\eqref{eq:exp1},~\eqref{eq:rel1},~\eqref{eq:rel2},~\eqref{eq:Gauss_LB_T1} and~\eqref{eq:Gauss_LB_T2}.

To establish~\eqref{eq:Gauss_LB_T1} observe that from the fact that conditioning reduces the differential entropy,
\begin{align}
&h(\rvs_{kp}^n, \ldots, \rvs_{kp+B-1}^n)\notag\\
&\hspace{1cm}-h(\rvs_{kp}^n, \ldots, \rvs_{kp+B-1}^n|\rvf_{0}^{kp-1}, \rvf_{kp+B}^{(k+1)p-1})\notag\\
&\ge \sum_{i=0}^{B-1}\left(h(\rvs_{kp+i}^n)-h(\rvs_{kp+i}^n|\rvf_{0}^{kp-1}, \rvf_{kp+B}^{(k+1)p-1})\right)\label{eq:exp3}
\end{align}
We show that for each $i=0,1,\ldots, B-1$
\begin{align}
h(\rvs_{kp+i}^n)\!-\!h(\rvs_{kp+i}^n|\rvf_{0}^{kp-1}, \rvf_{kp+B}^{(k+1)p-1})\! \!\ge\!\! \frac{n}{2}\log \left(\frac{1}{d_{B+W-i}}\right),\label{eq:Gauss_RD_Bound}
\end{align}
which then establishes~\eqref{eq:Gauss_LB_T1}.

Recall that since there is an erasure burst between time $t \in [kp, kp+B-1]$ the receiver is required to reconstruct \begin{align}\hat{\bt}^n_{(k+1)p-1} = \left[\hat{\bs}_{kp + B+W}^n, \ldots, \hat{\bs}_{kp}^n\right]\end{align}
with a distortion vector $(d_0,\ldots, d_{B+W})$ i.e., a reconstruction of $\hat{\bs}_{kp+i}^n$ is desired with a distortion of $d_{B+W-i}$ for $i=0,1,\ldots, B+W$ when the decoder is revealed $(\rvf_{0}^{kp-1}, \rvf_{kp+B}^{(k+1)p-1})$. Hence
\begin{align}
&h(\rvs_{kp+i}^n) - h(\rvs_{kp+i}^n|\rvf_{0}^{kp-1}, \rvf_{kp+B}^{(k+1)p-1})\notag\\
&=h(\rvs_{kp+i}^n) - h(\rvs_{kp+i}^n|\rvf_{0}^{kp-1}, \rvf_{kp+B}^{(k+1)p-1}, \{\hat{\rvs}^n_{kp+i}\}_{d_{B+W-i}})\\
&\ge h(\rvs_{kp+i}^n) - h(\rvs_{kp+i}^n|\{\hat{\rvs}^n_{kp+i}\}_{d_{B+W-i}})\\
&\ge h(\rvs_{kp+i}^n) - h(\rvs_{kp+i}^n-\{\hat{\rvs}^n_{kp+i}\}_{d_{B+W-i}})\label{eq:Gauss_toSub}
\end{align}
Since we have that 
\begin{align}
E\left[ \frac{1}{n}\sum_{j=1}^n (\rvs_{kp+i,j},\hat{\rvs}_{kp+i,j})^2 \right] \le d_{B+W-i}
\end{align}
It follows from standard arguments that~\cite[Chapter 13]{coverThomas} that
\begin{equation}
h(\rvs_{kp+i}^n-\{\hat{\rvs}^n_{kp+i}\}_{d_{B+W-i}}) \le \frac{n}{2}\log {2\pi e}{(d_{B+W-i})}. \label{eq:Gauss_Jensen}
\end{equation}

Substituting~\eqref{eq:Gauss_Jensen} into~\eqref{eq:Gauss_toSub} and the fact that $h(\rvs_{kp+i}^n)= \frac{n}{2}\log 2\pi e$ establishes~\eqref{eq:Gauss_RD_Bound}.

It finally remains to establish~\eqref{eq:Gauss_LB_T2}.
\begin{align}
&h(\rvs_{kp+B}^n, \ldots, \rvs_{(k+1)p-1}^n|\rvf_{0}^{kp-1},\rvs_{kp}^n, \ldots, \rvs_{kp+B-1}^n)-\notag\\ &\!\!h(\!\rvs_{kp+B}^n, \!\ldots,\!\! \rvs_{(k+1)p-1}^n\!|\! \rvf_{0}^{kp-1},\!\rvf_{kp+B}^{(k+1)p-1},\!\rvs_{kp}^n, \!\ldots, \!\rvs_{kp+B-1}^n\!)\notag\\ &+ H(\rvf_{kp+B}^{(k+1)p-1}|\rvf_{0}^{kp-1}, \rvt_{(k+1)p-1}^n) \notag\\
&\!\!=\!\!I(\!\rvs_{kp+B}^n, \!\ldots, \!\rvs_{(k+1)p-1}^n\!;\!\rvf_{kp+B}^{(k+1)p-1}\!|\!\rvf_{0}^{kp-1},\rvs_{kp}^n, \!\ldots, \!\rvs_{kp+B-1}^n\!)\notag\\ &\quad + H(\rvf_{kp+B}^{(k+1)p-1}|\rvf_{0}^{kp-1}, \rvt_{(k+1)p-1}^n)\\
&=H(\rvf_{kp+B}^{(k+1)p-1}|\rvf_{0}^{kp-1},\rvs_{kp}^n, \ldots, \rvs_{kp+B-1}^n)\\
&\!\ge \!I(\!\rvf_{kp+B}^{(k+1)p-1}\!;\!\rvs_{kp+B}^n,\!\! \ldots,\!\! \rvs_{(k+1)p-1}^n\!|\!\rvf_{0}^{kp+B-1},\!\rvs_{kp}^n,\!\! \ldots, \!\!\rvs_{kp+B-1}^n\!)
\end{align}

The above mutual information term can be bounded as follows:
\begin{align}
\!&\!h(\!\rvs_{kp+B}^n, \!\ldots, \!\rvs_{(k+1)p-1}^n\!|\!\rvf_{0}^{kp+B-1},\!\rvs_{kp}^n,\! \ldots, \!\rvs_{kp+B-1}^n\!)\notag\\ 
&~-h(\rvs_{kp+B}^n, \ldots, \rvs_{(k+1)p-1}^n|\rvf_{0}^{(k+1)p-1},\rvs_{kp}^n, \ldots, \rvs_{kp+B-1}^n)\notag\\
\!&=\!h(\!\rvs_{kp+B}^n, \!\ldots, \!\rvs_{(k+1)p-1}^n\!)\notag\\ 
&~-h(\rvs_{kp+B}^n, \ldots, \rvs_{(k+1)p-1}^n|\rvf_{0}^{(k+1)p-1},\rvs_{kp}^n, \ldots, \rvs_{kp+B-1}^n)\label{eq:Gauss_Indep}\\
&\ge h(\rvs_{kp+B}^n, \ldots, \rvs_{(k+1)p-1}^n)-  \notag \\ 
&\qquad h(\rvs_{kp+B}^n, \ldots, \rvs_{(k+1)p-1}^n|\{\hat{\rvs}_{kp}^n\}_{d_0}, \ldots, \{\hat{\rvs}_{(k+1)p-1}^n\}_{d_0})\label{eq:Gauss_Reconstr_d0}\\
&\ge \sum_{i=0}^{W}\left(h(\rvs_{kp+B+i}^n)-h(\rvs_{kp+B+i}^n-\{\hat{\rvs}_{kp+B+i}^n\}_{d_0}\right)\notag\\
&\ge \sum_{i=0}^{W}\frac{n}{2}\log (\frac{1}{d_0})=\frac{n(W+1)}{2}\log(\frac{1}{d_0})\label{eq:exp9}
\end{align}
where~\eqref{eq:Gauss_Indep} follows from the independence of $(\rvs_{kp+B}^n, \!\ldots, \!\rvs_{(k+1)p-1}^n)$ from the past sequences, and~\eqref{eq:Gauss_Reconstr_d0} follows from the fact that given the entire past $f_0^{(k+1)p-1}$ each source sub-sequence needs to be reconstructed with a distortion of $d_0$ and the last step  follows from the standard approach in the proof of the rate-distortion theorem. This establishes~\eqref{eq:Gauss_LB_T2}.

\bibliographystyle{IEEEtran}	
\bibliography{sm}		

\end{document}